\documentclass[11pt]{article}

\usepackage{graphicx}%
\usepackage{multirow}%
\usepackage{amsmath,amssymb,amsfonts}%
\usepackage{amsthm}%
\usepackage{mathrsfs}%
\usepackage[title]{appendix}%
\usepackage{xcolor}%
\usepackage{textcomp}%
\usepackage{manyfoot}%
\usepackage{booktabs}%
\usepackage{algorithm}%
\usepackage{algorithmicx}%
\usepackage{algpseudocode}%
\usepackage{listings}%

\usepackage{amsfonts,setspace}
\usepackage{amsmath}
\usepackage{comment}
\usepackage{dsfont}
\usepackage{anysize}
\usepackage{multicol}
\usepackage{graphicx}
\usepackage{titlesec}

\usepackage{tikz}
\usepackage{mathdots}
\usepackage{cancel}
\usepackage{color}
\usepackage{siunitx}
\usepackage{array}
\usepackage{float}
\usepackage{multirow}
\usepackage{gensymb}
\usepackage{tabularx}
\usepackage{extarrows}
\usepackage{booktabs}
\usetikzlibrary{fadings}
\usetikzlibrary{patterns}
\usetikzlibrary{shadows.blur}
\usetikzlibrary{shapes}

\usepackage{enumerate}
\usepackage{enumitem}
\usepackage{color}
\usepackage{caption}
\usepackage{subcaption}

\usepackage{hyperref}
\usepackage{cleveref}
\usepackage[skins]{tcolorbox}

\theoremstyle{thmstyleone}%
\newtheorem{theorem}{Theorem}
\newtheorem{proposition}[theorem]{Proposition}%

\theoremstyle{thmstyletwo}%

\theoremstyle{thmstylethree}%
\newtheorem{definition}{Definition}%

\theoremstyle{thmstylethree}%
\newtheorem{lemma}{Lemma}%

\theoremstyle{thmstylethree}%
\newtheorem{corollary}{Corollary}%

\usepackage{todonotes}

\theoremstyle{plain}
\newtheorem{claim}{Claim}

\newtheorem{remark}{Remark}

\titleclass{\subsubsubsection}{straight}[\subsubsection]
\newcounter{subsubsubsection}[subsubsection]
\renewcommand\thesubsubsubsection{\thesubsubsection.\arabic{subsubsubsection}}
\titleformat{\subsubsubsection}
  {\normalfont\normalsize\bfseries}{\thesubsubsubsection}{1em}{}
\titlespacing*{\subsubsubsection}
  {0pt}{3.25ex plus 1ex minus .2ex}{1.5ex plus .2ex}

\newcommand{\R}{\mathbb{R}}

\newcommand{\Q}{\mathbb{Q}}

\newcommand{\biglo}[1]{\Tilde{\mathcal{O}}(#1)}
\newcommand{\inside}[1]{\mathring{#1}}
\newcommand{\bigo}[1]{\mathcal{O}(#1)}
\newcommand{\CONGEST}{\textsf{CONGEST}}

\newcommand{\id}{\textsc{id}}

\newcommand{\dfsorder}{\textsc{dfs-order}}

\newcommand{\leftorder}{\textsc{left-dfs-order}}
\newcommand{\rightorder}{\textsc{right-dfs-order}}

\newcommand{\joinpath}{\texttt{join-path}}

\newcommand{\dfsrule}{\textsc{dfs-rule}}
\newcommand{\separatorproblem}{\textsc{separator-problem}}

\newcommand{\weightproblem}{\textsc{weigths-problem}}
\newcommand{\dfsorderproblem}{\textsc{dfs-order-problem}}

\newcommand{\joinproblem}{\textsc{join-problem}}

\newcommand{\markpathproblem}{\textsc{mark-path-problem}}

\newcommand{\leftpi}{\pi_\ell}
\newcommand{\rightpi}{\pi_{r}}
\newcommand{\minproblem}{\textsc{min-problem}}
\newcommand{\maxproblem}{\textsc{max-problem}}
\newcommand{\ancestorsumproblem}{\textsc{ancestor-sum-problem}}
\newcommand{\descendantsumproblem}{\textsc{descendant-sum-problem}}
\newcommand{\rangeproblem}{\textsc{range-problem}}
\newcommand{\sumsubsetproblem}{\textsc{sum-subset-problem}}
\newcommand{\sumtreeproblem}{\textsc{sum-tree-problem}}
\newcommand{\lcaproblem}{\textsc{LCA-problem}}

\newcommand{\detectfaceproblem}{\textsc{detect-face-problem}}
\newcommand{\detectancestor}{\textsc{ancestor-problem}}
\newcommand{\detectdescendant}{\textsc{descendant-problem}}
\newcommand{\hiddenproblem}{\textsc{hidden-problem}}
\newcommand{\fullagumentationproblem}{\textsc{full-augmentation-problem}}
\newcommand{\notcontainedproblem}{\textsc{not-contained-problem}}
\newcommand{\notcontainsproblem}{\textsc{not-contains-problem}}

\newcommand{\rerootproblem}{\textsc{re-root-problem}}
\usepackage{titlesec}

\newtcolorbox{algorithmbox}[2][]{%
	enhanced,colback=white,colframe=black,coltitle=black,
	sharp corners,boxrule=0.4pt,
	fonttitle=\itshape,
	attach boxed title to top left={yshift=-0.3\baselineskip-0.4pt,xshift=2mm},
	boxed title style={tile,size=minimal,left=0.5mm,right=0.5mm,
		colback=white,before upper=\strut},
	title=#2,#1
}

\begin{document}

\title{Deterministic Distributed DFS and Other Problems via Cycle Separators in Planar Graphs}


\author{Benjamin Jauregui \\
IRIF, Université Paris Cité \\
Universidad de Chile \\
  \texttt{jauregui@irif.fr} \\
  Paris, France 
\and 
Pedro Montealegre \\
Universidad Adolfo Ibañez \\
 \texttt{p.montealegre@uai.cl} \\
  Santiago, Chile 
\and
Ivan Rapaport \\
Universidad de Chile \\ 
\texttt{rapaport@dim.uchile.cl} \\
Santiago, Chile
}

\date{}

\maketitle

\begin{abstract}

One of the most basic techniques in algorithm design consists of breaking a problem into subproblems and then proceeding recursively. In the case of graph algorithms, one way to implement this approach is through separator sets. Given a graph $G=(V,E)$, a subset of nodes 
$S \subseteq V$ is called a separator set of $G$ if the size of each connected component of $G-S$ 
is at most $2/3 \cdot |V|$. The most useful separator sets are those that satisfy certain restrictions of cardinality or structure. For over 40 years, various efficient algorithms have been developed for computing separators of different kinds, particularly in planar graphs. Separator sets, combined with a divide and conquer approach, have been fundamental in the design of efficient algorithms in various settings.

In this work, we present the first deterministic algorithm in the distributed \CONGEST\ model that recursively computes a {\it cycle separator} in planar graphs in \(\Tilde{\mathcal{O}}(D)\) rounds. 
This result, as in the centralized setting, has significant implications for distributed planar algorithms. In fact, from this result, we can construct a deterministic algorithm that computes a DFS tree in $\biglo{D}$ rounds. This matches both the best-known randomized algorithm of Ghaffari and Parter (DISC'17) and, up to polylogarithmic factors, the trivial lower bound of \(\Omega(D)\) rounds.

Besides DFS, our deterministic cycle separator algorithm can be used to derandomize several planar-graph algorithms whose only randomized ingredient is the computation of a cycle separator, such as maximum flow (Abd-Elhaleem, Dory, Parter and Weimann, PODC'25), single-source shortest path (Li and Parter, STOC'19), and reachability (Parter, DISC'20).

 \end{abstract}
\clearpage

\tableofcontents
\clearpage 
\section{Introduction.}

Given a graph \( G = (V,E) \), a subset of nodes \( S \subseteq V \) is called a separator set of \( G \) if the size of each connected component of \( G-S \) is at most \( \frac{2}{3} \cdot |V| \). Separator sets were introduced by Lipton and Tarjan as a tool for designing approximation algorithms for NP-complete problems in planar graphs \cite{lipton:1979}. They proved that every planar graph can be divided into connected components with \( |S| = O(\sqrt{n}) \). These separator sets, combined with a divide-and-conquer strategy, enable solving smaller subproblems recursively, thus yielding polynomial-time approximation algorithms for planar graphs for problems such as maximum independent set, maximum matching, and others~\cite{separatorapplications}. Additionally, analogous results and characterizations have been developed for various classes of graphs, including bounded-genus graphs \cite{GILBERT1984391}, minor-excluding graphs \cite{seymournonplanar}, string graphs \cite{Fox2013ApplicationsOA}, and others.

In this article, we address the problem of computing  cycle separators  in planar graphs within a distributed framework, and their application to depth-first search (DFS). We consider the standard \CONGEST\ model~\cite{peleg2000distributed}, which assumes a network of nodes executing a synchronous, failure-free communication protocol. Communication occurs solely through the network links, with a constraint of transmitting no more than $O(\log n)$ bits per round across each link.

\subsection{Context.}

The development of solutions for the DFS problem has followed a challenging path in distributed computing. The DFS problem is notably difficult to solve in a distributed setting. Beyond the classic algorithm by Awerbuch \cite{awerbuch1985new}, which computes a DFS in $\mathcal{O}(n)$ rounds for general graphs, only the randomized algorithm by Ghaffari and Parter \cite{ghaffari:2017} for planar graphs has marked a significant improvement, nearly reaching the lower bound of $\Omega(D)$. Over the past 40 years, no deterministic improvements have been achieved for DFS in any significant graph class.  

The progress of DFS in planar graphs is not coincidental. In recent years, distributed algorithms for planar graphs have seen significant advancements, starting with the work of Ghaffari and Haeupler \cite{ghaffari2016distributed-1}, who introduced the concept of low-congestion shortcuts specifically for planar graphs. This groundbreaking framework has spurred progress in a variety of algorithmic applications for planar graphs, including minimum cuts, minimum spanning trees (MST), reachability, and other essential computational problems \cite{ghaffari2016distributed-1, ghaffari2016distributed-2, ghaffari:2017, parter}. Further details are provided in \Cref{subsec:shortcuts}.

The low-congestion shortcuts technique paved the way for the development of important algorithmic primitives in planar graphs, particularly for the computation of separator sets. To the best of our knowledge, the first application of separator sets in the \CONGEST\ model was by Ghaffari and Parter, who utilized them to compute a depth-first search (DFS) tree in planar graphs \cite{ghaffari:2017}. Their approach to distributively calculating separator sets was based on a randomized adaptation of a classical result by Lipton and Tarjan \cite{lipton:1979}, which states that every planar graph admits a specific type of separator set known as a \emph{cycle separator}. A cycle separator is defined as a separator set that either forms a cycle in \( G \) or a path in \( G \) such that its endpoints can be connected by an edge without crossing any other edges of \( G \).

More specifically, the result by Ghaffari and Parter~\cite{ghaffari:2017} shows that cycle separators can be computed in 2-node-connected\footnote{A graph \(G=(V,E)\) is \emph{\(c\)-node-connected} (or \emph{\(c\)-vertex-connected}) if \(|V|>c\) and \(G-S\) remains connected for every vertex set \(S\subseteq V\) with \(|S|<c\).} planar graphs in $\biglo{D}$ rounds. This approach was later extended by Li and Parter~\cite{parter}, who developed an algorithm for computing cycle separators in 1-node-connected planar graphs, enabling them to solve the SSSP problem in $\biglo{D^2}$ rounds. Subsequently, similar techniques based on randomized cycle separators were used to improve Single-Source Reachability to $\biglo{D}$ rounds~\cite{parter2020}, and to compute Maximum Flow in weighted directed planar graphs in $\biglo{D^2}$ rounds~\cite{abd2025distributed}.  It is noteworthy that all these algorithms are randomized, relying on probabilistic techniques to achieve their results. However, in many of these algorithms, the computation of cycle separators is the only randomized component, while all remaining procedures are deterministic.

Originally, both low-congestion shortcuts and cycle separator techniques only had randomized implementations. However, in \cite{haeupler2018round}, Haeupler, Hershkowitz and Wajc demonstrated  that it is possible to implement low-congestion shortcuts deterministically. This deterministic approach paves the way for generalizing all existing results that use these techniques for planar graphs. For instance, the deterministic shortcut technique directly enables a deterministic \CONGEST\ model algorithm to compute an MST in planar graphs in \( \biglo{D} \) rounds \cite{haeupler2018round}. In contrast,  there is a lower bound of \( \Omega(D + \sqrt{n}) \) in the number of rounds for MST computation in arbitrary graphs, which holds even for randomized algorithms. 

Nevertheless, the cycle separator technique has not yet been similarly extended; until now, only the randomized implementations by Ghaffari and Parter \cite{ghaffari:2017}, and Li and Parter \cite{parter}, have been available. This naturally raises the question: 

\medskip 

\begin{center}
\emph{Is it possible to compute cycle separators deterministically in \( \biglo{D} \) rounds\\ in the \CONGEST\ model?}
\end{center}

\medskip

In this article, we answer this question affirmatively.

\subsection{Our contribution.}

Our first contribution, stated in 
\Cref{teo:separatorcongest}, introduces the first deterministic algorithm in the \CONGEST\ model that computes cycle separators for planar graphs in $\biglo{D}$ rounds. Moreover, our result is even stronger, as it computes separators across multiple connected components in parallel.

\begin{theorem}\label{teo:separatorcongest}
    There exists a deterministic algorithm in the \CONGEST\ model which, given a planar graph $G=(V,E)$ with diameter $D$ and a partition $\mathcal{P} = \{P_1,...,P_k\}$ of $V$, computes in $\biglo{D}$ rounds  a cycle separator of $G[P_i]$ for each $i\in [k]$. 
\end{theorem}

To compute the separators,  Ghaffari and Parter \cite{ghaffari:2017} employ a sophisticated technique for approximating the number of nodes inside the faces of a planar embedding of $G$, known as the \emph{weight of the face}. Their algorithm simulates nodes of the dual graph of $G$, and uses that simulation to compute a good enough approximation of the weights of the faces. 

In contrast, our approach is more similar to the one used by Tarjan and Lipton in \cite{lipton:1979}, which proves that in every spanning tree $T$ of $G$ there is a path in $T$ that is a separator set of $G$. This involves computing \emph{triangulations} of the fundamental faces of $T$, which roughly consists of adding all edges within a face so that all faces inside it become triangles. Implementing the ideas of Tarjan and Lipton in a distributed manner is not straightforward, as triangulating a face could involve to simulate edges between nodes that are far apart in $G$. Instead, we show that it is sufficient to compute what we call an \emph{augmentation} of a face, which consists of adding some specific non-$G$ edges within a face, but not all the edges necessary for a triangulation.  Then, given a spanning tree $T$, we show that we can use our augmentations to find a cycle separator in one of the fundamental faces of $T$. The second key technical contribution of our deterministic algorithm is the definition of a {\it deterministic} formula that enables us to estimate the number of nodes inside each face without requiring randomized approximations.

As a representative application of \Cref{teo:separatorcongest}, we show the first efficient near-optimal  deterministic distributed algorithm for computing a DFS tree in planar graphs.

\begin{theorem}\label{teo:dfscongest}
    There exists a deterministic algorithm in the \CONGEST\ model which, given a planar graph $G~=~(V,E)$ with diameter $D$ and a node $r$, computes a DFS tree of $G$ rooted at $r$ in $\biglo{D}$ rounds.
\end{theorem}

Distributively computing a rooted spanning tree (and in particular a DFS tree) $T$ means that at the end of the algorithm each node $u$ outputs the identifier of the root, the identifier of its parent in $T$, as well as the distance $d_T(u)$ to the root. The algorithm given in \Cref{teo:dfscongest} is also inspired by a recursive approach used in parallel algorithms based on path separators \cite{dfs1987aggarwal}. In our case, the separator sets computed by the algorithm in \Cref{teo:separatorcongest} are the separators used throughout the algorithm.

Although we follow the same high-level approach as the randomized method of \cite{ghaffari:2017} for constructing a DFS tree, replacing their randomized cycle separators with our deterministic ones is not sufficient to obtain a deterministic DFS algorithm. Computing a DFS tree deterministically poses additional challenges beyond the deterministic computation of a separator set. For example, both \cite{ghaffari:2017} and our algorithm require marking a path of up to $\Theta(n)$ nodes.\footnote{In a distributed model, each node on the path must know that it belongs to the path.} While this task can be solved in $\biglo{D}$ rounds using randomized techniques, it is not clear how to achieve the same deterministically in $\biglo{D}$ rounds, rather than via the trivial $\bigo{n}$-round solution. Therefore, beyond using our deterministic cycle-separator algorithm, we develop new ideas to obtain a fully deterministic version of the scheme of \cite{ghaffari:2017} for constructing a DFS tree. See \Cref{subsec:subroutinescongest} for a full description of the subproblems that must be solved deterministically.

More precisely, we show that the problems solved deterministically in \Cref{subsec:subroutinescongest} allow us to add any cycle separator to a partial DFS tree, yielding a tree that remains a partial DFS tree. Furthermore, we show that there exists a \CONGEST\ algorithm that, given a partial solution $T$ (i.e. a partial DFS tree $T$ that includes some nodes of $G$) and a cycle separator $S$ of $G\setminus T$, computes a partial DFS tree  $T'$  that incorporates the nodes of both $T$ and $S$ in $\biglo{D}$ rounds.

In contrast to DFS, for which a deterministic algorithm for computing cycle separators is not sufficient to obtain a deterministic algorithm for computing a DFS tree, there are several problems in the literature, specially for {\it directed} planar graphs, such as weighted SSSP \cite{parter}, maximum flow\cite{abd2025distributed} and reachability  \cite{parter2020}, for which randomness is used exclusively in the cycle-separator subroutine, while the remainder of the algorithm relies solely on deterministic broadcasts and local computations. Therefore, our deterministic cycle separator algorithm also implies the derandomization of the following algorithms presented in Table 1.

\begin{table}[h!]
    \centering
    \begin{tabular}{|c|c|c|}
    \hline
        Problem & Round Complexity & Reference \\\hline
         SSSP& $\biglo{D^2}$ & \cite{parter} \\\hline 
         Single-Source Reachability & $\biglo{D}$ & \cite{parter2020} \\\hline
         Maximum $st$-Flow & $\biglo{D^2}$ & \cite{abd2025distributed}\\\hline
         Strongly Connected Component Detection & $\biglo{D}$ & \cite{parter2020}\\\hline
         
    \end{tabular}
    \caption{Consequences of \Cref{teo:separatorcongest}: for each problem listed above, the randomized algorithm in the cited reference can be derandomized by substituting its cycle-separator subroutine with our deterministic one, yielding the same deterministic round complexity.}
    \label{tab:detalgo}
\end{table}

For a formal definition of each problem and a brief description of the algorithms, we refer the reader to Section \ref{sec:otherapp}. The above list is not exhaustive, and the cited papers contain other results that may also be derandomized by replacing the randomized cycle-separator subroutine with our deterministic one. We refer the reader to those works for additional applications.

\subsection{Structure of the article.}

The article is structured as follows. In \Cref{sec:pre} we present some basic definitions used in the remaining of the paper.  In \Cref{sec:highlevelideas} we present a high-level description of our algorithms and techniques. In \Cref{sec:technicaldetails}, we give the formal proofs of the fundamental results on planar graph theory needed to conclude the correctness of our algorithms. This is followed by \Cref{sec:congestimplementation}, where we describe in detail the implementation of our \CONGEST\ algorithm for computing cycle separators and in \Cref{sec:dfscongest} we describe the implementation of our deterministic DFS tree computation. We describe other direct derandomizations obtained from our deterministic cycle-separator algorithm in \Cref{sec:otherapp}. Finally, in \Cref{sec:conclusion} we discuss extensions and future research directions.

\section{Basic definitions.}\label{sec:pre}

Given a graph $G=(V,E)$ and a node $v\in V$, we define $N_G(v)$ as the set of neighbors of $v$ in $G$, and we may omit the sub index if $G$ is clear by context. Given a subgraph $H$ of $G$, $N(H)$ denotes the set of nodes in $V(G)\setminus V(H)$ that are adjacent to at least one node of $H$, i.e., $N(H) = \cup_{u\in V(H)}N(u)\setminus V(H)$.  Given $G$ and a node $v\in V$, we may abuse notation and denote $v\in H$ instead of $v\in V(H)$.

Given a spanning tree $T$ of a graph $G$ and two nodes $u,v \in V$, the \emph{\(T\)-path} between \(u\) and \(v\) is the (unique) path in \(T\) connecting the nodes. Given a rooted tree $T$ and a node \(v\) of $T$, we denote by $T_v$ the subtree of $T$ rooted in $v$. We also denote by $n_T(v)$ the number of nodes in $T_v$. We say that a node \(u\) is \emph{ancestor} of a node \(v\) in a tree \(T\) if \(v\in T_u\). We say that a node \(v\) is \emph{descendant} of a node \(u\) in a tree \(T\) if \(u\) is ancestor of \(v\).  The \emph{lowest common ancestor} of nodes $u$ and $v$, denoted as LCA($u,v$), or simply LCA if $u$ and $v$ are clear by context, is defined as the deepest node $\omega$ in $T$ that is ancestor of both $u$ and $v$.

Given a planar graph $G$, \emph{a planar combinatorial embedding} $\mathcal{E}$ of $G$ is a function that assigns to each vertex $v \in V$ an ordering $t_v$ of its neighbors $N(v)$ such that, when the ordered neighbors are traversed clockwise in the plane, a geometric planar combinatorial embedding of $G$ is obtained. To avoid heavy notation, we define a {\it planar configuration} $(G,\mathcal{E},T)$ as a triplet where $G=(V,E)$ is a planar graph, $\mathcal{E}$ is a combinatorial planar embedding of $G$ and $T$ a spanning tree of $G$.

Consider a graph \(G = (V, E)\) with a planar combinatorial embedding \(\mathcal{E}\) and a spanning tree \(T\). The edges \(e \in \binom{V}{2} \setminus E(T)\) are called \emph{fundamental edges of} \(T\), or \(T\)-\emph{fundamental edges}. Note that a \(T\)-fundamental edge may or may not be an edge of \(G\). Specifically, a fundamental edge of \(T\) is said to be \emph{real} if \(e \in E(G)\) and \emph{virtual} if \(e \notin E(G)\). The set of fundamental edges of \(T\) has all the real and virtual fundamental edges of \(T\). When the context clearly indicates the spanning tree \(T\), this dependence may be omitted. 

Given a real fundamental edge \(e=uv\) of \(T\), a \emph{real fundamental face} \(F_e\) of \(T\) is defined as the subgraph of \(G\) induced by the set of nodes on the \(T\)-path between \(u\) and \(v\), plus the nodes within this path according to \(\mathcal{E}\). For a graphical example, see \Cref{fig:fundamentalcycles}. A node is considered \emph{inside} \(F_e\) if it lies on the face \(F_e\) and is not part of the \(T\)-path between \(u\) and \(v\) (see Figure \ref{fig:fundamentalcycles}); otherwise, if the node is not in \(F_e\), it is said to be \emph{outside} \(F_e\). The set of nodes inside \(F_e\) is denoted as \(\inside{F_e}\). The nodes on the path from \(u\) to \(v\) in \(T\) constitute the border of the fundamental face \(F_e\) and are denoted as \(C_e\). To avoid ambiguity, the inside zone of a face will be always defined as the bounded zone of the plane defined by the Jordan curve of the border of the face. We assume that given a spanning tree $T$ rooted in $r$, the root will not be inside any face. See \Cref{sec:technicaldetails} for details.

\begin{figure}[h!] 
\centering

    \scalebox{0.6}{\input{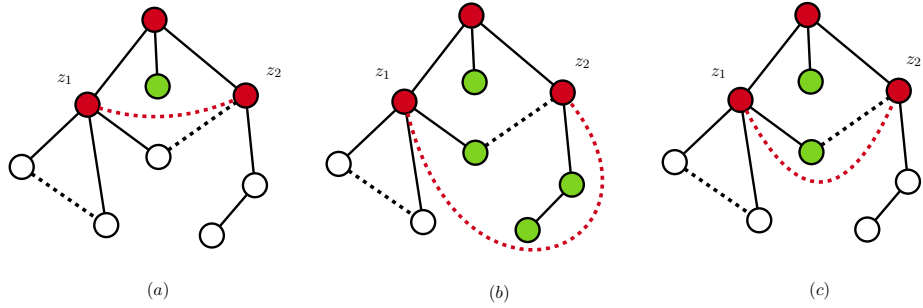}}
 
	\caption{Given embedding $\mathcal{E}$, black solid edges are $T$-edges, dotted black edges are real $T$-fundamental edges. Nodes $z_1$ and $z_2$ are $\mathcal{E}$-compatible, but there are multiple ways to insert the edge $e=\{z_1,z_2\}$.
 The subgraph induced by red and green nodes corresponds to $F_e$. Red nodes are the border $C_e$, green nodes are the node inside $F_e$ (i.e., in $\inside{F}_e$) and white nodes are the nodes outside $F_e$. The dotted red edges represent three ways to insert the virtual fundamental edge in $t_{z_1}$ and $t_{z_1}$, each insertion generates different faces. The nodes in each face are marked with red. }	  \label{fig:fundamentalcycles}
\end{figure}
Given a planar combinatorial embedding \(\mathcal{E}\) and a virtual fundamental edge \(e = \{uv\}\) of \(T\), the edge \(e\) is defined as \(\mathcal{E}\)-\emph{compatible} if a planar combinatorial embedding \(\mathcal{E}' = \{t'_w\}_{w \in V}\) of \(G' = (V, E \cup \{uv\})\) can be defined such that: \(t'_w = t_w\) for every \(w \notin \{u, v\}\); in \(t'_u\), the nodes \(w \in N_G(u)\) retain the same relative order as in \(t_u\); and the same condition holds for \(t'_v\) with respect to \(t_v\). Analogous to the definition given for real fundamental edges, the set \(\mathcal{F}_e\) is defined as the set of subgraphs of \(G\) induced by the nodes in the \(T\)-path between two nodes \(u, v\) of a compatible virtual fundamental edge \(e = uv\), plus the nodes inside the face defined by the virtual edge \(e\) inserted in \(\mathcal{E}\), satisfying the definition of \(\mathcal{E}\)-compatibility. Each subgraph in \(\mathcal{F}_e\) is called a \emph{virtual fundamental face} of \(T\). In other words, \(\mathcal{F}_e\) is the set of all possible virtual fundamental faces that can be generated by an \(\mathcal{E}\)-compatible virtual fundamental edge, depending on how \(e\) is validly inserted into \(\mathcal{E}\). 

\section{High-level description of our main results.}
\label{sec:highlevelideas}

We describe the main ideas and techniques developed to compute cycle separators and DFS trees. 

\subsection{Computation of cycle separators.}\label{separatorhighlevel}

The \separatorproblem\ is defined as follows: given a planar graph \( G = (V, E) \) and a partition \( \mathcal{P} = \{ P_1, \dots, P_k \} \) of \( V \), where each $G[P_i]$ is connected, the goal is to compute a family of sets \( S_1, \dots, S_k \) such that: (1) \( S_i \) is a cycle separator of \( G[P_i] \), and (2) each node in \( P_i \) knows whether it belongs to \( S_i \). As outlined in \Cref{subsec:shortcuts}, a planar embedding of \( G \) can be obtained in \( \biglo{D} \) rounds in the \CONGEST\ model, so we assume that a planar embedding \( \mathcal{E} \) of \( G \) is available. Additionally, each node is aware of its partition  within \( \mathcal{P} \).

The techniques developed to solve the \separatorproblem\ constitute key technical contributions of our work. Detailed explanations of the solution can be found in \Cref{sec:technicaldetails} and \Cref{sec:congestimplementation}. In this section, we discuss the primary technical challenges we encountered and explain how we addressed them to obtain our deterministic algorithm.

As we explained in the introduction, our approach follows the framework of Tarjan and Lipton \cite{lipton:1979}. We consider separator sets defined by a path in a spanning tree \( T \), along with a fundamental edge \( e \) (which may be real or virtual). The Tarjan and Lipton approach involves two tasks that are challenging to implement in the \CONGEST\ model. 

The first task involves counting the number of nodes within the fundamental face \( F_e \) defined by a non-tree edge \( e \). If that count represents a fraction between \( 1/3 \) and \( 2/3 \) of the total number of nodes, then the boundary \( F_e \) defines a separator cycle. Such an edge \( e \) is called a \emph{good edge}. The second task is finding a good edge. Unfortunately, it may happen that no non-tree edge is a good edge. However, by computing triangulations of certain fundamental faces, it is always possible to find a (potentially virtual non-$G$) good edge.

While both tasks are relatively straightforward in a centralized setting, they pose significant challenges in a distributed model. For example, if we wish to simulate communication between two nodes that are not directly connected (as might be required for a triangulation), we risk having to route messages through numerous intermediary nodes. If done in parallel, this creates the potential for severe congestion.

To address both problems, we redefine the two tasks mentioned above to make them more compatible with a distributed algorithm. First, we relax the counting of nodes in a face, and instead consider the \emph{weight of a face}, which essentially still measures the number of nodes inside the face, but allows for a margin of error, as we allow the count to include some of the nodes on the border of the face. Second, to locate the \emph{virtual good edges}, we compute an \emph{augmentation} of a face instead of a \emph{triangulation}. Roughly, this means that we simulate some of the {\it virtual} edges of the face, but not all of those required to create a triangulation. More precisely, if \( e = uv \) is a non-tree edge, in our augmentation we add all edges such that (1) one of their endpoints is \( u \) (it does not matter if we choose \( v \) instead), (2) the other endpoint is a node \( z \) in \( F_e \) reachable from \( u \) (i.e., the edge \( f = uz \) is compatible with the embedding), and (3) all descendants of \( z \) in \( T \) remain within the virtual face \( F_f \).

In the following sections we show not only that our notions of face weight and face augmentations are sufficient to identify a good edge and define a separator cycle, but also that their computation can be achieved in \( \tilde{\mathcal{O}}(D) \) rounds in the \CONGEST\ model.

\subsubsection{DFS orders of a spanning tree.} 

We first define a formula that allows us to deterministically compute, or at least estimate, the number of nodes within each cycle of \( G \). To achieve this, the first key ingredient we require is the notion of a depth-first search (DFS) order of an arbitrary spanning tree.

Consider a tree $T =(V,E)$ rooted in a node $r\in V$, and an embedding  $\mathcal{E}$ of $T$.  The \dfsorder\ of $T$ is a DFS enumeration of the nodes of $V$ that visits the neighbors of a node respecting the order given by $\mathcal{E}$. In fact, we distinguish two types of orderings. In the \leftorder\ (respectively the \rightorder), each node $v\in V$ will select first the unexplored children with the greater (resp. smaller) position in $t_v$\footnote{Recall that $t_v$ is a clockwise ordering of the neighbors of $v$ and we assume that $t_v(e)=1$ for the edge $e$ between $v$ and its parent in $T$}. Hence, the \rightorder\ picks the neighbors in the clockwise ordering, while the \leftorder\ in the counterclockwise order. See \Cref{fig:ordering} for an example of the \leftorder\ and the \rightorder.

 \begin{figure}[h]
     \centering
     \begin{subfigure}{0.4\textwidth}
        \centering
        \tikzset{every picture/.style={line width=0.75pt}} 

\begin{tikzpicture}[x=0.75pt,y=0.75pt,yscale=-1,xscale=1,scale=0.5]

\draw  [color={rgb, 255:red, 0; green, 0; blue, 0 }  ,draw opacity=1 ][line width=1.5]   (295.5, 82.83) circle [x radius= 10.12, y radius= 10.12]   ;
\draw (289,80.73) node [anchor=north west][inner sep=0.75pt]  [font=\tiny]  {$$};
\draw  [color={rgb, 255:red, 0; green, 0; blue, 0 }  ,draw opacity=1 ][line width=1.5]   (371.83, 146.17) circle [x radius= 10.12, y radius= 10.12]   ;
\draw (365.33,144.07) node [anchor=north west][inner sep=0.75pt]  [font=\tiny]  {$$};
\draw  [color={rgb, 255:red, 0; green, 0; blue, 0 }  ,draw opacity=1 ][line width=1.5]   (239.5, 153.83) circle [x radius= 10.12, y radius= 10.12]   ;
\draw (233,151.73) node [anchor=north west][inner sep=0.75pt]  [font=\tiny]  {$$};
\draw  [color={rgb, 255:red, 0; green, 0; blue, 0 }  ,draw opacity=1 ][line width=1.5]   (184.5, 205.83) circle [x radius= 10.12, y radius= 10.12]   ;
\draw (178,203.73) node [anchor=north west][inner sep=0.75pt]  [font=\tiny]  {$$};
\draw  [color={rgb, 255:red, 0; green, 0; blue, 0 }  ,draw opacity=1 ][line width=1.5]   (255.17, 254.5) circle [x radius= 10.12, y radius= 10.12]   ;
\draw (248.67,252.4) node [anchor=north west][inner sep=0.75pt]  [font=\tiny]  {$$};
\draw  [color={rgb, 255:red, 0; green, 0; blue, 0 }  ,draw opacity=1 ][line width=1.5]   (298.83, 197.5) circle [x radius= 10.12, y radius= 10.12]   ;
\draw (292.33,195.4) node [anchor=north west][inner sep=0.75pt]  [font=\tiny]  {$$};
\draw  [color={rgb, 255:red, 0; green, 0; blue, 0 }  ,draw opacity=1 ][line width=1.5]   (298.17, 138.5) circle [x radius= 10.12, y radius= 10.12]   ;
\draw (291.67,136.4) node [anchor=north west][inner sep=0.75pt]  [font=\tiny]  {$$};
\draw  [color={rgb, 255:red, 0; green, 0; blue, 0 }  ,draw opacity=1 ][line width=1.5]   (379.5, 236.17) circle [x radius= 10.12, y radius= 10.12]   ;
\draw (373,234.07) node [anchor=north west][inner sep=0.75pt]  [font=\tiny]  {$$};
\draw  [color={rgb, 255:red, 0; green, 0; blue, 0 }  ,draw opacity=1 ][line width=1.5]   (383.17, 300.17) circle [x radius= 10.12, y radius= 10.12]   ;
\draw (376.67,298.07) node [anchor=north west][inner sep=0.75pt]  [font=\tiny]  {$ $};
\draw (260,65) node [anchor=north west][inner sep=0.75pt]  [font=\Large,color={rgb, 255:red, 0; green, 0; blue, 0 }  ,opacity=1 ]  {$1$};
\draw (200 ,135) node [anchor=north west][inner sep=0.75pt]  [font=\Large,color={rgb, 255:red, 0; green, 0; blue, 0 }  ,opacity=1 ]  {$2$};
\draw (150,192.07) node [anchor=north west][inner sep=0.75pt]  [font=\Large,color={rgb, 255:red, 0; green, 0; blue, 0 }  ,opacity=1 ]  {$3$};
\draw (215,245) node [anchor=north west][inner sep=0.75pt]  [font=\Large,color={rgb, 255:red, 0; green, 0; blue, 0 }  ,opacity=1 ]  {$4$};
\draw (265,189.4) node [anchor=north west][inner sep=0.75pt]  [font=\Large,color={rgb, 255:red, 0; green, 0; blue, 0 }  ,opacity=1 ]  {$5$};
\draw (260,125.07) node [anchor=north west][inner sep=0.75pt]  [font=\Large,color={rgb, 255:red, 0; green, 0; blue, 0 }  ,opacity=1 ]  {$6$};
\draw (330,130) node [anchor=north west][inner sep=0.75pt]  [font=\Large,color={rgb, 255:red, 0; green, 0; blue, 0 }  ,opacity=1 ]  {$7$};
\draw (350,190) node [anchor=north west][inner sep=0.75pt]  [font=\Large,color={rgb, 255:red, 0; green, 0; blue, 0 }  ,opacity=1 ]  {$8$};
\draw (345,285) node [anchor=north west][inner sep=0.75pt]  [font=\Large,color={rgb, 255:red, 0; green, 0; blue, 0 }  ,opacity=1 ]  {$9$};
\draw [color={rgb, 255:red, 0; green, 0; blue, 0 }  ,draw opacity=1 ][line width=1.5]    (289.23,90.78) -- (245.77,145.88) ;
\draw [color={rgb, 255:red, 0; green, 0; blue, 0 }  ,draw opacity=1 ][line width=1.5]    (232.14,160.79) -- (191.86,198.88) ;
\draw [color={rgb, 255:red, 0; green, 0; blue, 0 }  ,draw opacity=1 ][line width=1.5]    (247.65,159.83) -- (290.68,191.5) ;
\draw [color={rgb, 255:red, 0; green, 0; blue, 0 }  ,draw opacity=1 ][line width=1.5]    (241.06,163.84) -- (253.61,244.49) ;
\draw [color={rgb, 255:red, 0; green, 0; blue, 0 }  ,draw opacity=1 ][line width=1.5]    (295.98,92.95) -- (297.68,128.39) ;
\draw [color={rgb, 255:red, 0; green, 0; blue, 0 }  ,draw opacity=1 ][line width=1.5]    (303.29,89.3) -- (364.04,139.7) ;
\draw [color={rgb, 255:red, 0; green, 0; blue, 0 }  ,draw opacity=1 ][line width=1.5]  [dash pattern={on 1.69pt off 2.76pt}]  (307.12,191.68) -- (363.55,151.99) ;
\draw [color={rgb, 255:red, 0; green, 0; blue, 0 }  ,draw opacity=1 ][line width=1.5]    (372.69,156.25) -- (378.64,226.08) ;
\draw [color={rgb, 255:red, 0; green, 0; blue, 0 }  ,draw opacity=1 ][line width=1.5]    (380.08,246.27) -- (382.59,290.06) ;
\draw [color={rgb, 255:red, 0; green, 0; blue, 0 }  ,draw opacity=1 ][line width=1.5]  [dash pattern={on 1.69pt off 2.76pt}]  (265.18,253.02) -- (369.48,237.64) ;
\draw [color={rgb, 255:red, 0; green, 0; blue, 0 }  ,draw opacity=1 ][line width=1.5]  [dash pattern={on 1.69pt off 2.76pt}]  (298.72,187.38) -- (298.28,148.62) ;
\draw [color={rgb, 255:red, 0; green, 0; blue, 0 }  ,draw opacity=1 ][line width=1.5]  [dash pattern={on 1.69pt off 2.76pt}]  (307.97,201.88) -- (370.37,231.79) ;
\draw [color={rgb, 255:red, 0; green, 0; blue, 0 }  ,draw opacity=1 ][line width=1.5]  [dash pattern={on 1.69pt off 2.76pt}]  (192.84,211.58) -- (246.83,248.76) ;

\end{tikzpicture}
         \subcaption{}
     \end{subfigure}
     \begin{subfigure}{0.4\textwidth}
        \centering
        \tikzset{every picture/.style={line width=0.75pt}} 

\begin{tikzpicture}[x=0.75pt,y=0.75pt,yscale=-1,xscale=1,scale=0.5]

\draw  [color={rgb, 255:red, 0; green, 0; blue, 0 }  ,draw opacity=1 ][line width=1.5]   (295.5, 82.83) circle [x radius= 10.12, y radius= 10.12]   ;
\draw (289,80.73) node [anchor=north west][inner sep=0.75pt]  [font=\tiny]  {$$};
\draw  [color={rgb, 255:red, 0; green, 0; blue, 0 }  ,draw opacity=1 ][line width=1.5]   (371.83, 146.17) circle [x radius= 10.12, y radius= 10.12]   ;
\draw (365.33,144.07) node [anchor=north west][inner sep=0.75pt]  [font=\tiny]  {$$};
\draw  [color={rgb, 255:red, 0; green, 0; blue, 0 }  ,draw opacity=1 ][line width=1.5]   (239.5, 153.83) circle [x radius= 10.12, y radius= 10.12]   ;
\draw (233,151.73) node [anchor=north west][inner sep=0.75pt]  [font=\tiny]  {$$};
\draw  [color={rgb, 255:red, 0; green, 0; blue, 0 }  ,draw opacity=1 ][line width=1.5]   (184.5, 205.83) circle [x radius= 10.12, y radius= 10.12]   ;
\draw (178,203.73) node [anchor=north west][inner sep=0.75pt]  [font=\tiny]  {$$};
\draw  [color={rgb, 255:red, 0; green, 0; blue, 0 }  ,draw opacity=1 ][line width=1.5]   (255.17, 254.5) circle [x radius= 10.12, y radius= 10.12]   ;
\draw (248.67,252.4) node [anchor=north west][inner sep=0.75pt]  [font=\tiny]  {$$};
\draw  [color={rgb, 255:red, 0; green, 0; blue, 0 }  ,draw opacity=1 ][line width=1.5]   (298.83, 197.5) circle [x radius= 10.12, y radius= 10.12]   ;
\draw (292.33,195.4) node [anchor=north west][inner sep=0.75pt]  [font=\tiny]  {$$};
\draw  [color={rgb, 255:red, 0; green, 0; blue, 0 }  ,draw opacity=1 ][line width=1.5]   (298.17, 138.5) circle [x radius= 10.12, y radius= 10.12]   ;
\draw (291.67,136.4) node [anchor=north west][inner sep=0.75pt]  [font=\tiny]  {$$};
\draw  [color={rgb, 255:red, 0; green, 0; blue, 0 }  ,draw opacity=1 ][line width=1.5]   (379.5, 236.17) circle [x radius= 10.12, y radius= 10.12]   ;
\draw (373,234.07) node [anchor=north west][inner sep=0.75pt]  [font=\tiny]  {$$};
\draw  [color={rgb, 255:red, 0; green, 0; blue, 0 }  ,draw opacity=1 ][line width=1.5]   (383.17, 300.17) circle [x radius= 10.12, y radius= 10.12]   ;
\draw (376.67,298.07) node [anchor=north west][inner sep=0.75pt]  [font=\tiny]  {$$};
\draw (316,75.4) node [anchor=north west][inner sep=0.75pt]  [font=\Large,color={rgb, 255:red, 0; green, 0; blue, 0 }  ,opacity=1 ]  {$1$};
\draw (390,135.4) node [anchor=north west][inner sep=0.75pt]  [font=\Large,color={rgb, 255:red, 0; green, 0; blue, 0 }  ,opacity=1 ]  {$2$};
\draw (397,228.4) node [anchor=north west][inner sep=0.75pt]  [font=\Large,color={rgb, 255:red, 0; green, 0; blue, 0 }  ,opacity=1 ]  {$3$};
\draw (400,289.4) node [anchor=north west][inner sep=0.75pt]  [font=\Large,color={rgb, 255:red, 0; green, 0; blue, 0 }  ,opacity=1 ]  {$4$};
\draw (314,126.4) node [anchor=north west][inner sep=0.75pt]  [font=\Large,color={rgb, 255:red, 0; green, 0; blue, 0 }  ,opacity=1 ]  {$5$};
\draw (257,139.4) node [anchor=north west][inner sep=0.75pt]  [font=\Large,color={rgb, 255:red, 0; green, 0; blue, 0 }  ,opacity=1 ]  {$6$};
\draw (319,183.4) node [anchor=north west][inner sep=0.75pt]  [font=\Large,color={rgb, 255:red, 0; green, 0; blue, 0 }  ,opacity=1 ]  {$7$};
\draw (269,257.4) node [anchor=north west][inner sep=0.75pt]  [font=\Large,color={rgb, 255:red, 0; green, 0; blue, 0 }  ,opacity=1 ]  {$8$};
\draw (202,193.4) node [anchor=north west][inner sep=0.75pt]  [font=\Large,color={rgb, 255:red, 0; green, 0; blue, 0 }  ,opacity=1 ]  {$9$};
\draw [color={rgb, 255:red, 0; green, 0; blue, 0 }  ,draw opacity=1 ][line width=1.5]    (289.23,90.78) -- (245.77,145.88) ;
\draw [color={rgb, 255:red, 0; green, 0; blue, 0 }  ,draw opacity=1 ][line width=1.5]    (232.14,160.79) -- (191.86,198.88) ;
\draw [color={rgb, 255:red, 0; green, 0; blue, 0 }  ,draw opacity=1 ][line width=1.5]    (247.65,159.83) -- (290.68,191.5) ;
\draw [color={rgb, 255:red, 0; green, 0; blue, 0 }  ,draw opacity=1 ][line width=1.5]    (241.06,163.84) -- (253.61,244.49) ;
\draw [color={rgb, 255:red, 0; green, 0; blue, 0 }  ,draw opacity=1 ][line width=1.5]    (295.98,92.95) -- (297.68,128.39) ;
\draw [color={rgb, 255:red, 0; green, 0; blue, 0 }  ,draw opacity=1 ][line width=1.5]    (303.29,89.3) -- (364.04,139.7) ;
\draw [color={rgb, 255:red, 0; green, 0; blue, 0 }  ,draw opacity=1 ][line width=1.5]  [dash pattern={on 1.69pt off 2.76pt}]  (307.12,191.68) -- (363.55,151.99) ;
\draw [color={rgb, 255:red, 0; green, 0; blue, 0 }  ,draw opacity=1 ][line width=1.5]    (372.69,156.25) -- (378.64,226.08) ;
\draw [color={rgb, 255:red, 0; green, 0; blue, 0 }  ,draw opacity=1 ][line width=1.5]    (380.08,246.27) -- (382.59,290.06) ;
\draw [color={rgb, 255:red, 0; green, 0; blue, 0 }  ,draw opacity=1 ][line width=1.5]  [dash pattern={on 1.69pt off 2.76pt}]  (265.18,253.02) -- (369.48,237.64) ;
\draw [color={rgb, 255:red, 0; green, 0; blue, 0 }  ,draw opacity=1 ][line width=1.5]  [dash pattern={on 1.69pt off 2.76pt}]  (298.72,187.38) -- (298.28,148.62) ;
\draw [color={rgb, 255:red, 0; green, 0; blue, 0 }  ,draw opacity=1 ][line width=1.5]  [dash pattern={on 1.69pt off 2.76pt}]  (307.97,201.88) -- (370.37,231.79) ;
\draw [color={rgb, 255:red, 0; green, 0; blue, 0 }  ,draw opacity=1 ][line width=1.5]  [dash pattern={on 1.69pt off 2.76pt}]  (192.84,211.58) -- (246.83,248.76) ;

\end{tikzpicture}
         \subcaption{}
     \end{subfigure}
     \caption{In (a) a \leftorder\ of $T$ and in (b) a \rightorder\ of $T$. The solid edges correspond to the edges of \(T\), while the non-tree edges 
     are dotted. The clockwise order \(t_v\) of the embedding is consistent with this figure.}
     \label{fig:ordering}
 \end{figure}
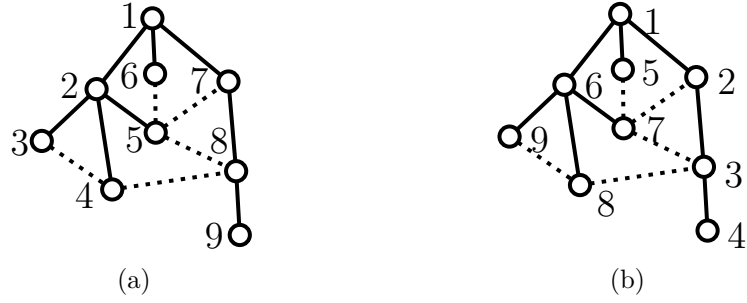

In the following, $\leftpi^T:V\to [n]$ and $\rightpi^T:V\to [n]$ are denoted as the \leftorder{} and \rightorder{} of a tree $T$ rooted at $r$ with respect to the embedding $\mathcal{E}$. If the context clearly indicates $T$, its dependence may be omitted. The \dfsorderproblem\ is then defined as the problem of computing the \leftorder\ and \rightorder.

This problem has been addressed in prior works (e.g., \cite{ghaffari2016distributed-2}); however, to the best of our knowledge, the spanning tree $T$ has always been a BFS tree, where its maximum depth is $D$ and therefore, it can be solved in $\bigo{D}$ rounds. In this work, an arbitrary spanning tree is considered, which can have a maximum depth of up to $\Theta(n)$. In \Cref{sec:congestimplementation} it is shown that these problems can also be solved in $\biglo{D}$ rounds in the \CONGEST\ model.


\subsubsection{Weights of fundamental faces.}

 Let $e = uv$ be a real fundamental edge of $T$, and $F_e$ its unique real fundamental face. A \emph{weight} is assigned to $F_e$, which represents the capacity of the \(T\)-path between \(u\) and \(v\) to serve as a separator set. 
 
 In \cite{ghaffari:2017}, the weight of a real fundamental face was defined as the number of nodes inside it, plus the number of nodes in its border. In their work, they can only compute an $(1+\epsilon)$-approximation of the weight of a face, by using a (random) {\it sketching} idea \cite{ghaffari2013distributed, ghaffari2016distributed-2}, which consists, in a high level ideal,  in $\biglo{\log_{1+\epsilon} n}$ phases of $\biglo{\log n}$ iterations each one, and in each iteration every node inside a face marks itself with certain probability. In every iteration, a selected leader of each face can learn the number of nodes marked inside the face, and it is proven that after these iterations, a $(1+\epsilon)$-approximation of the weight is computed. 

Since the $(1+\varepsilon)$-approximation technique for the weight of a face used in \cite{ghaffari:2017} is inherently randomized, we adopt a different notion of face \emph{weight}. Roughly speaking, our definition counts the number of nodes strictly inside the face, and includes only a subset of the nodes on its boundary (rather than necessarily counting all boundary nodes, as in \cite{ghaffari:2017}).

Using this new notion of weight for a face $F_e$ defined by a fundamental edge $e=uv$, we define $\omega(F_e)$ through an explicit formula that allows the endpoints $u$ and $v$ to compute $\omega(F_e)$ deterministically in at most $\biglo{D}$ rounds. The exact deterministic formula for $\omega(F_e)$ will depend on whether one of the two nodes of $e$ is an ancestor of the other, or if neither node is an ancestor of the other. To describe the formula in the latter case (when neither node is an ancestor of the other), we need to define the notion of $\mathcal{E}$-left and $\mathcal{E}$-right. Roughly speaking, a real fundamental edge $e=uv$ such that $u$ is an ancestor of $v$ is said to be a $\mathcal{E}$-left oriented edge if the edge $e$ is drawn 'to the left' in the embedding $\mathcal{E}$, and $e$ is said to be a $\mathcal{E}$-right oriented edge if the edge $e$ is drawn 'to the right' in the embedding $\mathcal{E}$.

\begin{definition}
    Given a planar configuration $(G,\mathcal{E},T)$, consider a $T$-real fundamental edge $e = uv$ such that $\leftpi(u) < \leftpi(v)$, where $u$ is an ancestor of $v$ in $T$. Let $z \in V$ be the first node in the $T$-path from $u$ to $v$. If $t_u(v) < t_u(z)$ then the real fundamental edge $e$ is said to be a $\mathcal{E}$-left oriented edge. Conversely, if $t_u(v) > t_u(z)$ $e$ is said to be a $\mathcal{E}$-right oriented edge. This definition extends to virtual fundamental edges $f=uv$ inserted in a specific order in 
    $t_u$ and $t_v$ of $\mathcal{E}$.
\end{definition}

For example, in \Cref{fig:weightancestors}, the red and dotted edge $e=uv$ is a $\mathcal{E}$-left oriented edge for that planar instance $(G,\mathcal{E},T)$. Now we can describe the explicit formula to compute the weight of a face in each case. See \Cref{fig:weightancestors2} and \Cref{fig:weightancestors} for a graphical description of the weight defined below.

\begin{definition}
\label{def:weights}
    
Given a planar configuration $(G,\mathcal{E},T)$ and a $T$-fundamental face $F_e$, which is either the unique face of a $T$-real fundamental edge or $F_e \in \mathcal{F}_e$ for some $\mathcal{E}$-compatible virtual fundamental edge $e$, such that $e = uv$ and $\leftpi(u) < \leftpi(v)$. The \emph{weight} of $F_e$, denoted by $\omega(F_e)$, is defined as follows. 
 
 \begin{enumerate}[itemsep=0pt, label*=\arabic*.]
     \item If $u$ is not an ancestor of $v$ in $T$, then
     \[\omega(F_e) = p_{F_e}(v)+p_{F_e}(u) + \leftpi(v)- (\leftpi(u) + n_T(u)) +1.\]
     \item If $u$ is an ancestor of $v$, let $z$ be the first node in the $T$-path from $u$ to $v$ (i.e., $uz\in E(T)$). We distinguish two cases:
     \begin{enumerate}[label*=\arabic*.]
        \item If $e$ is $\mathcal{E}$-left oriented, then 
         \[\omega(F_e) = p_{F_e}(v)+p_{F_e}(u) + \left(\leftpi(v)-\leftpi(z)\right) - (d_T(v)-d_T(z)).\]
         \item If $e$ is $\mathcal{E}$-right oriented, then
         \[\omega(F_e) = p_{F_e}(v)+p_{F_e}(u) + \left(\rightpi(v)-\rightpi(z)\right) - (d_T(v)-d_T(z)).\]
     \end{enumerate}

 \end{enumerate}

 where $p_{F_e}(u) = |V(F_e\cap (T_u\setminus \{u\}))|$, $p_{F_e}(v) = |V(F_e\cap (T_v\setminus \{v\}))|$, $d_T(v)$ is the depth of node $v$ in the spanning tree $T$ and $n_T(u)$ is the number of nodes in the subtree of $T$ rooted in $u$ (without considering $u$).

\end{definition}

In Section~\ref{sec:weights}, it is proven exactly which nodes are counted in each case of \Cref{def:weights}. In \Cref{subsec:subroutinescongest}, an algorithm is provided that allows the nodes of each real fundamental edge to compute the weights exactly as defined in \Cref{def:weights} in $\biglo{D}$ rounds in the \CONGEST\ model. Furthermore, it is shown that the same complexities hold if, for each subset of nodes $P_i$ from a partition $\mathcal{P} = \{P_1, \ldots, P_k\}$, the weights of each $T_i$-real fundamental edge are computed in parallel, with $T_i$ being a spanning tree of $G[P_i]$. 

\begin{figure}[h!]
    \centering
    \scalebox{0.7}{\tikzset{every picture/.style={line width=0.75pt}} 

\begin{tikzpicture}[x=0.75pt,y=0.75pt,yscale=-1,xscale=1]

\draw  [draw opacity=0][fill={rgb, 255:red, 184; green, 233; blue, 134 }  ,fill opacity=0.28 ][dash pattern={on 1.69pt off 2.76pt}][line width=1.5]  (306.58,16.1) .. controls (333.31,-8.43) and (486.36,61.63) .. (488.12,109.69) .. controls (489.88,157.74) and (489.88,269.87) .. (480.19,291.79) .. controls (470.5,313.71) and (398.23,285.05) .. (369.15,285.05) .. controls (340.07,285.05) and (265.74,269.12) .. (246.65,251.32) .. controls (227.56,233.53) and (233.75,230.86) .. (232.22,216.59) .. controls (230.7,202.33) and (231.67,182.28) .. (231.67,162.8) .. controls (231.67,143.32) and (279.85,40.64) .. (306.58,16.1) -- cycle ;
\draw  [draw opacity=0][fill={rgb, 255:red, 74; green, 144; blue, 226 }  ,fill opacity=0.5 ][line width=1.5]  (265.16,199.9) .. controls (274.85,209.17) and (271.33,208.33) .. (276.61,219.29) .. controls (281.9,230.25) and (290.72,243.74) .. (276.61,251.32) .. controls (262.51,258.91) and (259.87,246.26) .. (242.24,220.97) .. controls (224.62,195.68) and (255.46,190.62) .. (265.16,199.9) -- cycle ;
\draw  [draw opacity=0][fill={rgb, 255:red, 74; green, 144; blue, 226 }  ,fill opacity=0.5 ][line width=1.5]  (325.08,23.69) .. controls (357.69,38.02) and (422.03,51.51) .. (444.06,76.81) .. controls (466.09,102.1) and (479.31,93.67) .. (478.43,126.55) .. controls (477.55,159.43) and (477.55,283.36) .. (451.99,288.42) .. controls (441.41,299.38) and (445.71,263.65) .. (407.92,220.13) .. controls (372.67,218.44) and (317.55,254.58) .. (302.17,210.86) .. controls (294.48,188.99) and (282.32,96.66) .. (288.95,54.89) .. controls (295.58,13.11) and (321.01,21.9) .. (325.08,23.69) -- cycle ;
\draw  [draw opacity=0][fill={rgb, 255:red, 74; green, 144; blue, 226 }  ,fill opacity=0.5 ][line width=1.5]  (409.69,229.4) .. controls (422.91,236.99) and (414.97,239.52) .. (407.04,248.79) .. controls (399.11,258.07) and (388.54,255.54) .. (390.3,240.36) .. controls (392.06,225.19) and (396.47,221.82) .. (409.69,229.4) -- cycle ;
\draw [line width=1.5]    (143,237) .. controls (149.83,256.55) and (200.38,266.5) .. (239.98,239.18) ;
\draw [shift={(243,237)}, rotate = 143.13] [fill={rgb, 255:red, 0; green, 0; blue, 0 }  ][line width=0.08]  [draw opacity=0] (11.61,-5.58) -- (0,0) -- (11.61,5.58) -- cycle    ;
\draw [line width=1.5]    (290.43,304.39) .. controls (322.61,303.41) and (378.14,302.11) .. (394.8,258.44) ;
\draw [shift={(396,255)}, rotate = 107.68] [fill={rgb, 255:red, 0; green, 0; blue, 0 }  ][line width=0.08]  [draw opacity=0] (11.61,-5.58) -- (0,0) -- (11.61,5.58) -- cycle    ;
\draw [line width=1.5]    (571,167.05) .. controls (552.48,201.17) and (517.79,224.84) .. (480.85,193.42) ;
\draw [shift={(478,190.89)}, rotate = 42.77] [fill={rgb, 255:red, 0; green, 0; blue, 0 }  ][line width=0.08]  [draw opacity=0] (11.61,-5.58) -- (0,0) -- (11.61,5.58) -- cycle    ;

\draw  [line width=1.5]   (320.09, 33.04) circle [x radius= 10, y radius= 10]   ;
\draw (314.09,30.44) node [anchor=north west][inner sep=0.75pt]  [font=\tiny]  {$$};
\draw  [line width=1.5]   (468.73, 157.25) circle [x radius= 10, y radius= 10]   ;
\draw (462.73,154.65) node [anchor=north west][inner sep=0.75pt]  [font=\tiny]  {$$};
\draw  [line width=1.5]   (208.17, 267.13) circle [x radius= 10, y radius= 10]   ;
\draw (202.17,264.53) node [anchor=north west][inner sep=0.75pt]  [font=\tiny]  {$$};
\draw  [line width=1.5]   (374.14, 130.84) circle [x radius= 10, y radius= 10]   ;
\draw (368.14,128.24) node [anchor=north west][inner sep=0.75pt]  [font=\tiny]  {$$};
\draw (189.33,269.2) node [anchor=north west][inner sep=0.75pt]  [font=\Large]  {$u$};
\draw (451.61,279.3) node [anchor=north west][inner sep=0.75pt]  [font=\Large,rotate=-359.78]  {$v\ $};
\draw (317.62,268.94) node [anchor=north west][inner sep=0.75pt]  [font=\Large]  {$e$};
\draw  [line width=1.5]   (428.19, 80.81) circle [x radius= 10, y radius= 10]   ;
\draw (422.19,78.21) node [anchor=north west][inner sep=0.75pt]  [font=\tiny]  {$$};
\draw  [line width=1.5]   (450.52, 269.38) circle [x radius= 10, y radius= 10]   ;
\draw (444.52,266.78) node [anchor=north west][inner sep=0.75pt]  [font=\tiny]  {$$};
\draw  [line width=1.5]   (209.34, 193.51) circle [x radius= 10, y radius= 10]   ;
\draw (203.34,190.91) node [anchor=north west][inner sep=0.75pt]  [font=\tiny]  {$$};
\draw  [line width=1.5]   (224.03, 104.14) circle [x radius= 10, y radius= 10]   ;
\draw (218.03,101.54) node [anchor=north west][inner sep=0.75pt]  [font=\tiny]  {$$};
\draw  [line width=1.5]   (257.81, 146.85) circle [x radius= 10, y radius= 10]   ;
\draw (251.81,144.25) node [anchor=north west][inner sep=0.75pt]  [font=\tiny]  {$$};
\draw  [line width=1.5]   (312.45, 105.21) circle [x radius= 10, y radius= 10]   ;
\draw (306.45,102.61) node [anchor=north west][inner sep=0.75pt]  [font=\tiny]  {$$};
\draw  [line width=1.5]   (326.55, 166.87) circle [x radius= 10, y radius= 10]   ;
\draw (320.55,164.27) node [anchor=north west][inner sep=0.75pt]  [font=\tiny]  {$$};
\draw  [line width=1.5]   (365.33, 178.34) circle [x radius= 10, y radius= 10]   ;
\draw (359.33,175.74) node [anchor=north west][inner sep=0.75pt]  [font=\tiny]  {$$};
\draw  [line width=1.5]   (421.08, 151.8) circle [x radius= 10, y radius= 10]   ;
\draw (415.08,149.2) node [anchor=north west][inner sep=0.75pt]  [font=\tiny]  {$$};
\draw  [line width=1.5]   (419.97, 214.3) circle [x radius= 10, y radius= 10]   ;
\draw (413.97,211.7) node [anchor=north west][inner sep=0.75pt]  [font=\tiny]  {$$};
\draw  [line width=1.5]   (321.27, 211.55) circle [x radius= 10, y radius= 10]   ;
\draw (315.27,208.95) node [anchor=north west][inner sep=0.75pt]  [font=\tiny]  {$$};
\draw  [line width=1.5]   (254, 208.96) circle [x radius= 10, y radius= 10]   ;
\draw (248,206.36) node [anchor=north west][inner sep=0.75pt]  [font=\tiny]  {$$};
\draw  [line width=1.5]   (268.1, 236.78) circle [x radius= 10, y radius= 10]   ;
\draw (262.1,234.18) node [anchor=north west][inner sep=0.75pt]  [font=\tiny]  {$$};
\draw  [line width=1.5]   (402.05, 237.35) circle [x radius= 10, y radius= 10]   ;
\draw (396.05,234.75) node [anchor=north west][inner sep=0.75pt]  [font=\tiny]  {$$};
\draw (119,211) node [anchor=north west][inner sep=0.75pt]   [align=left] {$\displaystyle p_{F_{e}}( u)$};
\draw (240,295) node [anchor=north west][inner sep=0.75pt]   [align=left] {$\displaystyle p_{F}{}_{e}( v)$};
\draw (500,146) node [anchor=north west][inner sep=0.75pt]   [align=left] {$\displaystyle \pi _{\ell }( v) -\ ( \pi _{\ell }( u) \ +\ n_{T}( u)) \ +1$};
\draw [line width=1.5]    (329.24,37.08) -- (419.05,76.77) ;
\draw [line width=1.5]    (432.88,89.65) -- (464.05,148.42) ;
\draw [line width=1.5]    (467.13,167.12) -- (452.12,259.51) ;
\draw [line width=1.5]    (208.33,257.14) -- (209.18,203.5) ;
\draw [line width=1.5]    (210.97,183.64) -- (222.41,114.01) ;
\draw [line width=1.5]    (232.07,98.19) -- (312.05,38.99) ;
\draw [line width=1.5]    (420.85,87.61) -- (381.48,124.04) ;
\draw [line width=1.5]  [dash pattern={on 1.69pt off 2.76pt}]  (218.17,267.23) -- (440.52,269.29) ;
\draw [line width=1.5]    (216.55,186.57) -- (250.61,153.79) ;
\draw [line width=1.5]    (364.9,127) -- (321.69,109.05) ;
\draw [line width=1.5]    (426.47,206.7) -- (462.23,164.85) ;
\draw [line width=1.5]    (383.28,134.92) -- (411.95,147.72) ;
\draw [line width=1.5]  [dash pattern={on 1.69pt off 2.76pt}]  (378.96,139.6) -- (415.16,205.53) ;
\draw [line width=1.5]    (372.32,140.67) -- (367.15,168.51) ;
\draw [line width=1.5]    (366.17,136.87) -- (334.53,160.83) ;
\draw [line width=1.5]  [dash pattern={on 1.69pt off 2.76pt}]  (267.72,145.49) -- (364.23,132.2) ;
\draw [line width=1.5]  [dash pattern={on 1.69pt off 2.76pt}]  (316.8,169.09) -- (219.1,191.29) ;
\draw [line width=1.5]  [dash pattern={on 1.69pt off 2.76pt}]  (331.26,211.83) -- (409.97,214.02) ;
\draw [line width=1.5]    (322.44,201.62) -- (325.38,176.8) ;
\draw [line width=1.5]    (214.36,259.28) -- (247.81,216.82) ;
\draw [line width=1.5]    (217.09,262.61) -- (259.17,241.3) ;
\draw [line width=1.5]    (442.18,263.87) -- (410.39,242.86) ;

\end{tikzpicture}}
    \caption{How we count the number of nodes in $F_e$ when $e=uv$,  $u$ is not an ancestor of $v$. Nodes in the light green area are the nodes counted in $\omega(F_e)$}
    \label{fig:weightancestors2}
\end{figure}
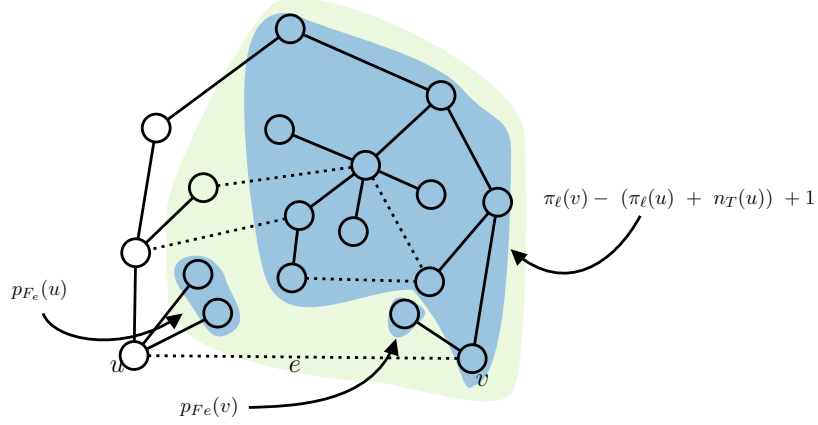

\begin{figure}[h!]
    \centering
    \scalebox{0.9}{
    \input{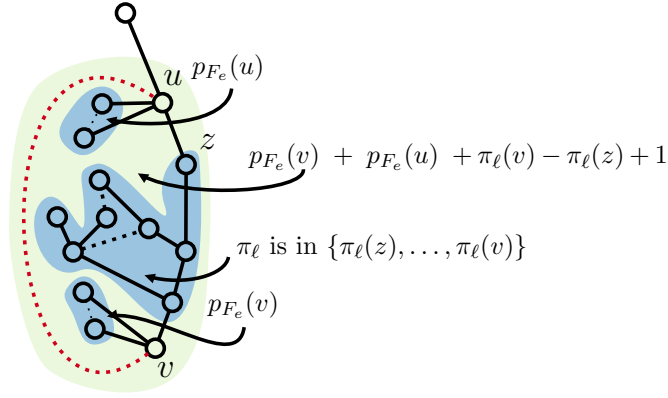}}
    \vspace{-1cm}
    \caption{How we count the number of nodes in $F_e$ when $e=uv$ when node $u$ is an ancestor of $v$ and $t_u(v)>t_u(z)$. Nodes in the light green area are the nodes counted in $\omega(F_e)$  }
    \label{fig:weightancestors}
\end{figure}

\subsubsection{Tree compatible edges and augmentations.}
\label{subsec:highlevelaugmentation}

As already explained, if there exists a real fundamental edge $e$ such that $\omega(F_e)$ has weight in $[n/3,2n/3]$, then the $T$-path between the endpoints of $e$ is a cycle separator. In the bad scenario in which no real fundamental edge has weight in this interval, we would like to find an \(\mathcal{E}\)-compatible fundamental edge \(e\) (recall that these edges are not $G$ edge, but there is a way to connect the endpoint in $\mathcal E$ such that planarity is not broken) such that the weight of some face in $\mathcal F_e$ is in the desired interval.

Given an \(\mathcal{E}\)-compatible fundamental edge \(e\), if a fundamental face \(F_e \in \mathcal{F}_e\) exists such that the number of nodes inside \(F_e\) is in the range \([|P_i|/3, 2|P_i|/3]\), it can be ensured that the \(T\)-path between \(u\) and \(v\) serves as a separator set. This is because, given that this face \(F_e\) is defined by a valid insertion of \(e\) into \(\mathcal{E}\), there cannot be any edge between a node inside \(F_e\) and a node outside \(F_e\). This is formally proven in \Cref{teo:balancedface}.

Therefore, in the case when no real fundamental face has a weight in the range \(\left[|P_i|/3, 2|P_i|/3\right]\), a naive approach would compute the weight of each virtual fundamental face \(F_e \in \mathcal{F}_e\) for all \(\mathcal{E}\)-compatible edges \(e \in \binom{V}{2} - E(G)\)\footnote{Recall that different faces \(F_e, F_e' \in \mathcal{F}_e\) indicate that \(e\) is inserted in different orders in \(\mathcal{E}\), and therefore their weights can differ.}. The problem is that computing the weights of all possible faces for all \(\mathcal{E}\)-compatible edges could cause congestion in the edges. Therefore, we define a restricted notion of virtual fundamental face for an \(\mathcal{E}\)-compatible edge.

\begin{definition}\label{def:lazyrechable}
Let \( (G, \mathcal{E}, T) \) be a planar configuration, \( e = uv \) be a fundamental edge, and \( z_1, z_2 \in F_{e} \) be a pair of non-adjacent nodes. The virtual fundamental edge \( f = z_1z_2 \) is called a \emph{\((T, F_{uv})\)-compatible edge} if:
(1) \( f = z_1z_2 \) can be inserted in \(\mathcal{E}\) such that \( f \) is contained in \( F_{e} \) and does not cross any other edge. (2) All nodes in \( V(T_{z_1}) \cup V(T_{z_2}) \) that are contained in \( F_e \) are also contained in \( V(F_f) \).

 Given a \((T, F_e)\)-compatible edge \( f \), a face in \(\mathcal{F}_f\) satisfying the definition of \((T, F_e)\)-compatibility is denoted as \( F_{f}^\ell \). Given a  node $z_1 \in F_e$, other node node $z_2\in F_e$ is called  \((T, F_e)\)-compatible with $z_1$ if $z_1z_2$ is a \((T, F_e)\)-compatible edge. See \Cref{fig:triangulation} for an example.

\end{definition}

\begin{figure}[h!]
\centering
\resizebox{\textwidth}{!}{\input{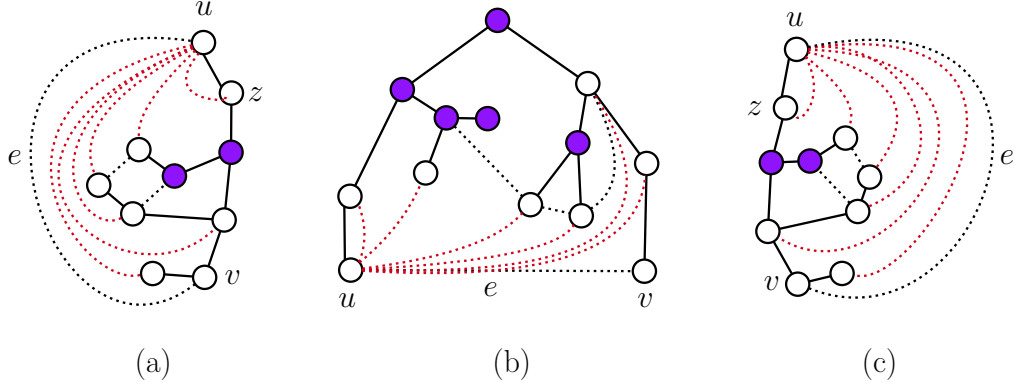}}
\caption{An example of the $(T,F_e)$-compatible nodes with the node $u$. Black edges are $T$-edges, black dashed edges represent real fundamental edges, red dashed edges are the virtual fundamental edges added in a $(T,F_e)$-compatible way and purple nodes are not $(T,F_e)$-compatible with node~$u$.}
    \label{fig:triangulation}
    
\end{figure}

By using the weights defined in \Cref{def:weights} together with the notion of augmentation of a $(T,F_e)$-compatible edge, the following lemma becomes the core technical result (proved in \Cref{subsec:proofteo}) that enables us to compute our cycle separator. Indeed, the algorithm proving \Cref{teo:separatorcongest} can be viewed essentially as a distributed implementation of the proof of the lemma below.

\begin{lemma}\label{teo:fundamentalseparator}
Given a planar configuration $(G,\mathcal{E}, T)$ such that $G$ has $n$ nodes. At least one of the following conditions holds:
\begin{enumerate}
    \item There exists a real fundamental edge $e$ satisfying $\omega(F_e)\in [n/3,2n/3]$. 
    \item There exists a real fundamental edge $e$ and a $(T,F_e)$-compatible edge $f$ such that $\omega(F_f) \in [n/3,2n/3]$. 
    \item There exists a real fundamental edge $e$, or a $(T,F_{e})$-compatible edge $f$ such that the $T$-path between the extremes of $e$ or $f$, respectively, has more than $n/3$ nodes.\footnote{This is a particular and easy case.}

\end{enumerate}
Furthermore, if a (real or virtual) \(T\)-fundamental edge \(e\) satisfies any of the conditions above, then the \(T\)-path between the endpoints of \(e\) induces a cycle separator of \(G\).

\end{lemma}

We now provide the high-level ideas to prove \Cref{teo:fundamentalseparator}. Given a planar configuration \((G, \mathcal{E}, T)\), if there exists a \( T \)-real fundamental face \( F_e \) such that \( \omega(F_e) \in [n/3, 2n/3] \), then the nodes in the \( T \)-path between the endpoints of \( e \) form a separator set of \( G \). On the other hand, if no \( T \)-real fundamental face has a weight in the range \([n/3, 2n/3]\), there are three possible scenarios: (1) \( G \) has no \( T \)-real fundamental edges, (2) \( G \) has \( T \)-real fundamental edges, and at least one of them has a weight greater than \( 2n/3 \), or (3) \( G \) has \( T \)-real fundamental edges, and all of them have a weight smaller than \( n/3 \).

If case (1) holds, then \( G \) is a tree and must have a centroid node \( v_0 \) such that the subtree \( T_{v_0} \) rooted at \( v_0 \) contains at least \( n/3 \) nodes and fewer than \( 2n/3 \) nodes. 

Case (2) is the most complicated case. Let \( e \) be a real fundamental edge such that \( \omega(F_e) > 2n/3 \) with \( e = uv \). In this case, we show that there exists a node \( z \in \inside{F_e} \) that is \((T, F_e)\)-compatible with \( u \) such that \( \omega(F_{uz}^\ell) \in [n/3, 2n/3] \). Then, the \( T \)-path between \( u \) and \( z \) is a cycle separator of \( G \).

Finally, if case (3) holds, let \( e \) be a \( T \)-real fundamental edge that is not contained in any other \( T \)-real fundamental face. It can be proven that if there are at most \( 2n/3 \) nodes outside \( F_e \), then the nodes in the \( T \)-path between the endpoints of \( e \) form a cycle separator of \( G \). If there are more than \( 2n/3 \) nodes outside \( F_e \), the case reduces to one analogous to case (2), proving that there exists an \( \mathcal{E} \)-compatible virtual fundamental edge \( f \) that admits a fundamental face \( F_f \in \mathcal{F}_f \) with \( \omega(F_f) > 2n/3 \). See \Cref{lemma:balancedincritical}, \Cref{lemma:verysmall}, and \Cref{subsec:proofteo} for details.

\Cref{teo:fundamentalseparator} implies that the search for a cycle separator reduces to finding a (real or virtual) fundamental edge \( e \) such that \( \omega(F_e) \in [|P_i|/3, 2|P_i|/3] \). In \Cref{subsec:separatoralgcongest}, we show that this search can be carried out through the application of a series of relatively simple subroutines, consisting of performing up-casts and down-casts over \( T \) of local information  located in the nodes (e.g. presence of edges, number of nodes in certain subtrees, etc.). In fact, we show that most of our subroutines can be represented as \emph{part-wise aggregation problems}. These problems were defined by Haeupler and Ghaffari as an algorithmic primitive that can be solved in \( \tilde{\mathcal{O}}(D) \) rounds using the low-congestion shortcut technique \cite{haeupler2018round}. All these details are described in \Cref{subsec:shortcuts}.

\subsection{Depth first search trees.}
\label{subsec:dfshighlevelideas}

We now give a high-level description of the deterministic algorithm that computes a DFS in planar graphs. The deterministic DFS algorithm will be referred to as the \emph{main algorithm}. In each recursive call, referred to as a \emph{phase}, a partial DFS-tree \(T_{d}\) is already constructed. In the first phase, \(T_{d}\) consists only of the node \(r\) without any edge. Then, in each phase, the following steps are computed in parallel for each connected component \(C_i\) of \(G - T_d\).

\begin{enumerate}[align= left, label=\textbf{Step \arabic*.}, ref=\textbf{Step \arabic*}.]
	\item  A cycle separator \(S_i\) of \(C_i\) is computed. 
	\item  The nodes in \(S_i\) are added to the partial DFS-tree \(T_{d}\) following the \dfsrule.\
\end{enumerate}

\noindent
The \dfsrule\ is defined as follows. Let \(T_d\) be a partial DFS tree. To expand \(T_d\), pick a node \(w\in V\setminus V(T_d)\) that lies in some connected component \(C\) of \(G-T_d\). Let \(r_C\in C\) be a node that has a neighbor in \(T_d\) of maximum depth (ties broken arbitrarily). To preserve the DFS structure, we add to \(T_d\) the entire path in \(C\) from \(r_C\) to \(w\), and connect \(r_C\) to \(T_d\) via its deepest neighbor in \(T_d\). Any node added to \(T_d\) keeps its parent \id\ and its depth in \(T_d\); these values remain fixed even if \(T_d\) is extended further. This procedure is the \dfsrule\ (see \Cref{fig:virtualface}(a)). Throughout, a \emph{partial DFS tree} is a rooted subgraph \(T_d\subseteq G\) (rooted at \(r\)) whose vertices were introduced by repeated applications of the \dfsrule.
\medskip

The high-level approach is the same as in the randomized algorithm of \cite{ghaffari:2017}. However, the way we implement Steps~1 and~2 deterministically is fundamentally different. For Step~1, at each recursion level we work on the subgraphs induced by the connected components of $G-T_d$. These vertex sets form a node-disjoint collection. Therefore, Step~1 follows directly from \Cref{teo:separatorcongest}: we can compute a cycle separator in each induced subgraph within $\biglo{D}$ rounds. Indeed, \Cref{teo:separatorcongest} extends straightforwardly from a partition of $V(G)$ to any node-disjoint collection of vertex subsets.

It remains to explain how to implement Step~2 deterministically.

\subsubsection{Joining each separator set to the partial DFS tree.}

In each phase \( j \) of the main algorithm, once Step 1 concludes, a cycle separator for each connected component \( C_i \) of \( G - T_d \) has been computed. With this, the formal problem to be solved in Step 2, which its called \joinproblem, is the following: Given a partial DFS tree $T_d$ of $G$, the connected components $\mathcal{C}=\{C_1,...,C_k\}$ of $G-T_d$, and a set of marked nodes $S_i$ in each $C_i$ that is a cycle separator, to add all the nodes of each $S_i$ to the DFS tree $T_d$. 

\begin{lemma}\label{teo:joinkpaths}

There exists a deterministic algorithm in the \CONGEST\ model which solves \joinproblem\ in $\biglo{D}$ rounds.
    
\end{lemma}
 While the complete proof of \Cref{teo:joinkpaths} is given in Section \ref{subsec:dfscongestalg}, the main idea of the algorithm is explained here. Let $S_i=u_ix_i^1....x_i^kv_i$ be the cycle separator of $C_i$ computed in Step 1. and $T_i$ an spanning tree of $C_i$ rooted in a node $r_i$. Without loss of generality, we can assume the root $r_i$ has the deepest neighbor in $T_d$ (the current partial DFS computed).
 
In order to build \(\tilde{T}_d\) from the output of Step 1, one can grow \(T_d\) by including the \(T_i\)-path from \(r_{i}\) to the furthest node in \(T_i\) between \(u_i\) and \(v_i\). Since \(r_{i}\) is a node with the deepest neighbor in \(T_d\), this path is added to \(T_d\) following the \dfsrule. However, this approach might only partially include the set \(S_i\), as \(S_i\) may not be entirely contained in any \(T_i\)-path starting from \(r_i\). By doing this, it is ensured that at least half of the nodes of \(S_i\) are included in \(T_d\) (see \Cref{fig:virtualface} (b)).

By adding the path from $r_{i}$ to $u_i$ or $v_i$, it is meant that all nodes in this path acknowledge the \id\ of their parent and their depth in $T_d$. Although this information can be easily transmitted through $T_i$, this path can contain up to $\Omega(n)$ nodes. Therefore, to ensure all subroutines execute in $\biglo{D}$ rounds, this path must be added in an efficient manner. In \Cref{sec:congestimplementation} we show that this can done in  $\biglo{D}$ rounds in the \CONGEST\ model.

After the first recursion ends, at least half of the nodes of $S_i$ have been added to $T_d$. Then, recursively, following the same steps but with the set $S_i'$ of the nodes of $S_i$ that have not yet been added to $T_d$, again, half of these remaining nodes can be added to $T_d$, each time using $\biglo{D}$ rounds. As the number of nodes in each $S_i$ not added to the partial DFS tree decreases by at least half in each recursion, the process terminates after $\bigo{\log n}$ iterations.

 \begin{figure}[h] 
\centering
  \scalebox{0.6}{\input{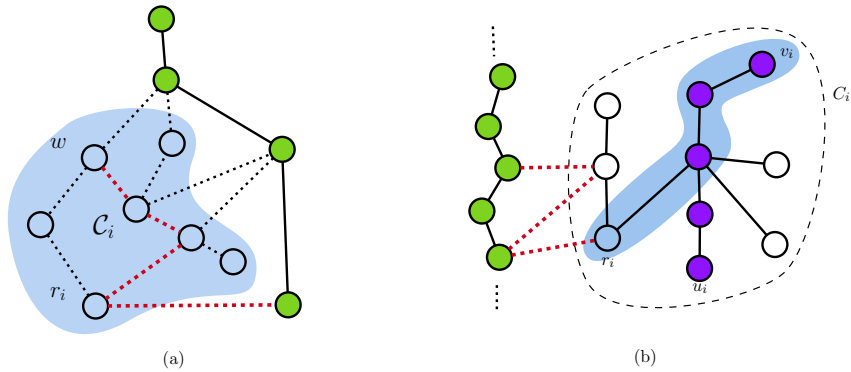}}
	\caption{ (a) Green nodes are nodes already in $T_d$. The nodes in the blue cloud represent a connected component \(\mathcal{C}_i\) of \(G-T_d\). To join node \(w\) to this tree, the red depicted arbitrary \(w, r_i\)-path in \(\mathcal{C}_i\) follows the \dfsrule. (b) Green nodes are in $T_d$. Nodes highlighted with purple are the cycle separator of $\mathcal{C}_i$. The nodes in the blue cloud represent the path from $r_{i}$ to $v_i$, the deepest node in the cycle separator. This path is added to $T_d$.}
	\label{fig:virtualface}
\end{figure}

\section{Fundamental results in planar graph theory.}\label{sec:technicaldetails}

Before describing the algorithm to compute a separator set, we prove formally some result regarding the weights and the augmentation of a face in order to proof \Cref{teo:fundamentalseparator}, which is the core of our distributed algorithm.

First, as explained before, in order to avoid ambiguity in the definition of the inside zone of a face, given planar configuration $(G,\mathcal{E},T)$ and a real fundamental face \(F_e\), the cycle defined by the \(T\)-path composed of the nodes on the border \(C_e\) plus the edge \(e\) defines a Jordan curve that splits the plane into two sets: the \emph{outside} (unbounded) and the \emph{inside} (bounded). The interior of a face \(F_e\) is always considered as the inside portion of the plane defined by the Jordan curve \(C_e\). To formally define this set, consider a graph \(G = (V, E)\) and a spanning tree \(T\) rooted at \(r\), with a virtual root \(r_0\) connected only to the root \(r\), and $r_0$ lies in the external face of $\mathcal E$. The augmented graph is \(\Tilde{G} = (V \cup \{r_0\}, E \cup \{\{r, r_0\}\})\). Then, the real fundamental face defined by a real fundamental edge \(e\) is the face of \(T + e\) that does not contain the virtual root \(r_0\). 

\subsection{Weighting the fundamental faces of a given spanning tree.} \label{sec:weights}

Suppose that \( e = uv \) is a real fundamental edge such that \(\leftpi(u) < \leftpi(v)\) and \( u \) is not an ancestor of \( v \). Let \( w \) denote the lowest common ancestor (LCA) of \( u \) and \( v \) in \( T \). The set \(\tilde{F}_e\) is defined as the set of all nodes in \(\inside{F_e}\) plus the nodes in the \( T \)-path from \( w \) to \( v \). While Figures \ref{fig:weightancestors} and \ref{fig:weightancestors2} illustrate which nodes are counted by $\omega(F_e)$ for each fundamental real face $F_e$, the following lemmas (Lemmas \ref{lem:weighting} and \ref{lem:weighting2}) formally prove that the weight of a fundamental edge provides an approximation of the number of nodes contained within the corresponding fundamental face, and that these nodes correspond exactly to those shown in Figures \ref{fig:weightancestors} and \ref{fig:weightancestors2}.

\begin{lemma}\label{lem:weighting}
Given a planar configuration $(G,\mathcal{E},T)$ and a real fundamental edge $e=uv$, such that $\leftpi(u)<\leftpi(v)$ and $u$ is not an ancestor of $v$. Then $\omega(F_e)$  equals the number of nodes in $\tilde{F}_e$.
     
\end{lemma}

\begin{proof}
Let $w$ to be the LCA of $u$ and $v$ in $T$. Let us call $P_u: w = u_0, u_1, \dots, u_{k+1} = u$ and $P_v : w = v_0, v_1, \dots, u_{k'+1} = v$ the $T$-path from $w$ to $u$ and $v$, respectively, for appropriate values of $k\geq0$ and $k'\geq0$. In the following, we assume that the planar combinatorial embedding $\mathcal{E}$ satisfies that if $p(z)$ is the parent of $z$ in $T$, then  $t_z(p(z)) = 0$ (if this is not the case, the nodes can shift the local order of their neighbors clockwise to reach this condition without any communication).

\begin{claim}\label{claim:1} Let $z\in N(C_e)$. We have that $z\in \inside{F_e}$ if and only if one of the following conditions holds:
\begin{enumerate}[label=(\roman*)]
\item $z\in N(w)$ and  $t_w(v_1)<t_w(z)<t_w(u_1)$.
\item $z \in N(u)$ and $t_u(z)< t_u(v)$.
\item $z \in N(v)$ and $t_v(u) < t_v(z)$.
\item There is a $i\in  \{1,\dots, k\}$ such that $z\in N(u_i)$ and $t_{u_i}(z) < t_{u_i}(u_{i+1})$.
\item There is a $i\in  \{1,\dots, k'\}$ such that $z\in N(v_i)$ and $t_{v_i}(v_{i+1})<t_{v_i}(z)$.
\end{enumerate}
\end{claim}

\begin{proof}[Proof of \Cref{claim:1}] If $z\in N(\omega)$, condition (i) is obtained because we always assume that  $e$ is an edge that goes \emph{below} $T$, i.e., $F_e$ does not contain the (virtual) parent of $w$. Consider now the orientation of $C_e$ that, starting from $w$, goes over the nodes in $P_u$ in increasing order (with respect to their indices), and then goes over the nodes in $P_v$ in decreasing order.  Condition (i) implies that the interior of the Jordan curve defined by $C_e$ is found on the left of the defined orientation. More precisely, the face $F_e$ is found in the side of the orientation of $C_e$ that covers in a clockwise ordering from the predecessor to the successor of a node in the cycle. From this fact are deduced conditions (ii) and (iii) if $z\in N(u)$ or $z\in N(v)$, respectively. If $z\notin N(\omega)\cup N(u)\cup N(v)$, as $z\in N(C_e)$, there exists an $i\in[k]$ such that $z\in N(u_i)$ or a $i'\in [k']$ such that $z\in N(v_i)$, in any case by the same argument that cases (ii) and (iii) the conditions (iv) and (v). holds. 
\end{proof}

\begin{claim}\label{claim:2}~
\begin{enumerate}[label=(\roman*)]
\item For every $i\in \{1, \dots, k\}$,  $\leftpi(u_i) < \leftpi(u_{i+1})$;
\item  For every $j\in\{1,\dots, k'\}$, $\leftpi(v_j)< \leftpi(v_{j+1})$;
\item  For every node $z$ of $T_{u_1}$, $\leftpi(z) < \leftpi(v_1)$
\item  $\leftpi(u) + n_T(u) < \leftpi(v_1)$
\end{enumerate}
\end{claim}

\begin{proof}[Proof of \Cref{claim:2}]
 (i) is given because $P_u$ is the only $T$-path connecting $w$ and $u$. (ii) holds by the same argument for $v$ and the nodes in $P_v$. (iii) is obtained because $\leftpi(u)<\leftpi(v)$ and then, in a  $\leftorder$, all nodes in $T_{u_1}$ are visited before $v_1$. (iv) is obtained because, in particular, all nodes in $T_u$ are visited before $v_1$, hence $\leftpi(v_1) > \leftpi(u)+n_T(u)$. 
\end{proof}

\begin{claim}\label{claim:3} Every node $z\in V$ satisfies that
  $z \in \tilde{F_e} \setminus (T_u \cup T_v \cup \{w\})$ iff $\leftpi(z) \in [\leftpi(u)+n_T(u)+1, \leftpi(v)-1]$.
\end{claim}

\begin{proof}[Proof of \Cref{claim:3}] Note first that if a node $z$ is in $\inside{F_e}$, then all the nodes in $T_z$ are also in $\inside{F_e}$. Then, a node $\alpha \in V$ belongs to $\inside{F_e}$ if and only if $\alpha \in T_z$  for some $z\in N(C_e)\cap \inside{F_e}$. Hence, from \Cref{claim:1}, we have that $\alpha$ belongs to $\inside{F_e}$ if and only there is a node $z$ satisfying one of the conditions (i)-(v) of the claim such that $\alpha \in T_z$. Since we are considering nodes in $\tilde{F}_e \setminus  (T_u \cup T_v \cup \{w\})$ we only analyze cases (i), (iv) and (v):
\begin{itemize}
\item (Case (i)) $z\in N(w)$ and  $t_w(u_1)<t_w(z)<t_w(v_1)$. Then, we have that in the \leftorder\  the nodes in $T_z$ are visited before after all the nodes in $T_{u_1}$ and before $v_1$. In particular, $\alpha$ is visited after all the nodes in $T_u$ and, by \Cref{claim:2}, before $v$. Hence,  $\leftpi(u) + n_T(u) < \leftpi(\alpha) < \leftpi(v_1) < \leftpi(v)$. 
\item (Case (iv)) There is a $i\in  \{1,\dots, k\}$ such that $z\in N(u_i)$ and $t_{u_i}(z) < t_{u_i}(u_{i+1})$. In this case all the nodes in $T_z$ are visited in the \leftorder\ after the nodes in $T_{u_{i+1}}$  and before $v_1$. Hence,  $\leftpi(u) + n_T(u) < \leftpi(\alpha) < \leftpi(v)$. 
\item  (Case (v)) There is a $i\in  \{1,\dots, k'\}$ such that $z\in N(v_i)$ and $t_{v_i}(v_{i+1})<t_{u_i}(z)$. In this case all the nodes in $T_z$ are visited  in the \leftorder\ after $v_{i}$ and before $v_{i+1}$. By \Cref{claim:2} we deduce that  $\leftpi(u) + n_T(u) < \leftpi(v_i) < \leftpi(\alpha) < \leftpi(v_{i+1}) \leq \leftpi(v)$. 
\end{itemize}
We deduce that a node $\alpha$ belongs to $\inside{F_e} \setminus  (T_u \cup T_v)$ then $\leftpi(\alpha)\in [\leftpi(u)+n_T(u)+1, \leftpi(v)-1]$. Observe that from \Cref{claim:2} all the nodes $\alpha \in P_v \setminus \{w, v\}$ also satisfy that $\leftpi(\alpha)\in [\leftpi(u)+n_T(u)+1, \leftpi(v)-1]$. We deduce that $\alpha$ belongs to $\tilde{F}_e \setminus  (T_u \cup T_v \cup \{w\})$ then $\leftpi(\alpha)\in [\leftpi(u)+n_T(u)+1, \leftpi(v)-1]$.

Now let us show that if $\alpha \in V$ satisfies that $\leftpi(\alpha) < \leftpi(u) + n_T(u)$, then $\alpha$ cannot belong to $\tilde{F_e} \setminus (T_u \cup T_v)$. First, from \Cref{claim:2} we know that the nodes $z$ in $P_v\setminus\{w\}$ satisfy that $\leftpi(z) > \leftpi(u)+ n_T(u)$, so $\alpha$ is not in $P_v\setminus \{w\}$. Then, let us show that $\alpha$ is not contained in $\inside{F_e}$. If $\leftpi(\alpha)<\leftpi(w)$ then $\alpha \notin T_w$. Since we  assume that  $e$ is an edge that goes \emph{below} $T$, this implies that $\alpha \notin \inside{F_e}$.  If $\leftpi(\alpha) \in [\leftpi(u), \leftpi(u)+n_T(u)]$ then  $\alpha \in T_u$, which implies that $\alpha \notin \tilde{F_e} \setminus (T_u \cup T_v)$. It remains the case where $\leftpi(\alpha) \in [\leftpi(w)+1, \leftpi(u)-1]$. Let us pick a node $z\in T_w$ such that $\alpha\in T_z$ and $\leftpi(z) \in [\leftpi(w)+1, \leftpi(u)-1]$ is minimum. Observe that necessarily $z$ is a neighbor of one of the nodes in $P_u$. Let $i\in \{0, \dots, k\}$ be such that $z\in N(u_i)$. Observe that $\leftpi(z) < \leftpi(u_{i+1})$, as otherwise $\leftpi(z) > \leftpi(\beta)$ for all $\beta \in T_{u_{i+1}}$, which in particular would imply that $\leftpi(z) > \leftpi(u) + n_T(u)$. Hence $t_{u_i}(z) > t_{u_i}(u_{i+1})$. From \Cref{claim:1} ((i) and (iv)) we deduce that $z\notin F_e$. Moreover, every node in $T_z$ is outside $F_e$, in particular $\alpha \notin \inside{F_e}$. 

Finally, let us assume that $\leftpi(\alpha) > \leftpi(v)$ and show that $\alpha$ cannot belong to $\tilde{F_e} \setminus (T_u \cup T_v)$. First, from \Cref{claim:2} we know that the nodes $z$ in $P_v\setminus\{w\}$ satisfy that $\leftpi(z)\leq \leftpi(v)$. Then necessarily $\alpha \notin P_v$. Now let us show that $\alpha \notin \inside{F_e}$. If $\leftpi(\alpha)>\leftpi(w)+n_T(w)$ then again $\alpha\notin T_w$, which implies that $\alpha \notin F_e$. It remains the case when $\leftpi(\alpha) \in [\leftpi(v)+1, \leftpi(w) + n_T(w)]$. We pick the node $z\in T_w$ such that $\alpha \in T_z$, $z\in  [\leftpi(v)+1, \leftpi(w) + n_T(w)]$ and $\leftpi(z)$ is minimum. We have that $z$ is a neighbor of one of the nodes in $P_v$. Let $i\in \{0, \dots, k'\}$ be such that $z\in N(v_i)$. We have that $\leftpi(z)>\leftpi(v_{i+1})$ as otherwise, by \Cref{claim:2}, $\leftpi(z)$ would be smaller than $\leftpi(v)$. Therefore $t_{v_{i}}(z)< t_{v_{i}}(v_{i+1})$. From \Cref{claim:1} ((i) and (v)) we deduce that $z\notin \inside{F_e}$, which implies that $\alpha \notin \inside{F_e}$. 
\end{proof}

\Cref{claim:3} implies that the number of nodes in $\tilde{F}_e \setminus (T_u \cup T_v \cup \{w\})$ equals $$(\leftpi(v)-1)-(\leftpi(u)+n_T(u)+1) + 1 = \leftpi(v)-(\leftpi(u)+n_T(u)) -1.$$ To obtain the number of nodes in $\tilde{F}_e$ we have to add to this quantity the number of nodes in $\tilde{F}_e\cap (T_u \cup T_v)$, plus one for $w$.   The number of nodes in $\tilde{F}_e \cap T_u$ equals $P_{F_e}(u)$, and the number of nodes in $\tilde{F}_e \cap T_v$ equals $P_{F_e}(v) + 1$ (as $v$ is included in $\tilde{F}_e$). Therefore, the total number of nodes in $\tilde{F}_e$ is $\omega\left(F_e\right) = P_{F_e}(u) + P_{F_e}(v) + \leftpi(v)-(\leftpi(u)+n_T(u)) +1$.
\end{proof}

\begin{lemma}\label{lem:weighting2}
    Given a graph configuration $(G,\mathcal{E},T)$ and a real fundamental edge $e=uv$, such $u$ is an ancestor of $v$. Then $\omega\left(F_e\right)$ equals the number of nodes in  $\inside{F_e}$.
\end{lemma}

\begin{proof}
Let us call $P_{uv} \colon u = w_0, w_1, \dots, w_{k+1} = v$ the path in $T$ from $u$ to $v$, for appropriate $k\geq 0$. Observe that since $u$ is an ancestor of $v$, then, for each $i \in \{0, \dots, k\}$, $\leftpi(w_i)<\leftpi(w_{i+1})$. Notice also that  the set of nodes in $P_{uv}$ corresponds to the border of $F_e$. 

In the following we assume that $t_u(v) > t_u(w_1)$. Therefore, we aim to prove that the number of nodes in $\inside{F_e}$ equals $\omega\left(F_e\right) = P_{F_e}(v) + P_{F_e}(u) + (\leftpi(v)-\leftpi(w_1)) - (d_T(v)-d_T(w_1))$. The case $t_u(v) < t_u(w_1)$ is completely analogous, switching the orders of some inequalities in the appropriate places. 

\begin{claim}\label{claim:p1}
Let $z\in N(P_{uv})$. We have that $z \in \inside{F_e}$ if and only if one of the following conditions hold:
\begin{enumerate}[label=(\roman*)]
\item $z\in N(u)$ and  $t_w(w_1)<t_w(z)<t_w(v_1)$.
\item $z \in N(v)$ and $t_v(u) < t_v(z)$.
\item There is a $i\in  \{1,\dots, k\}$ such that $z\in N(w_i)$ and $t_{w_i}(w_{i+1})<t_{w_i}(z)$.
\end{enumerate}
\end{claim}

The proof of this claim is completely analogous to the proof of \Cref{claim:1}, and can be obtained by contracting the path $P_u$ defined on the proof of \Cref{lem:weighting} into a single node.

\begin{claim}\label{claim1p2}
        Let \(z\in N(P_{uv})\). We have that \(z\in (F_e \cap T_{w_1}) \setminus T_v\) if and only if \[(\leftpi(z)\in [\leftpi(w_1),\leftpi(v)-1].\]
\end{claim}

\begin{proof}[Proof of \Cref{claim1p2}]

Let us pick an arbitrary \(z\in (F_e \cap T_{w_1}) \setminus T_v\). If $z$ is in $P_{uv}\setminus\{u,v\}$ we directly have that $ \leftpi(z)\in [\leftpi(w_1),\leftpi(v)]$, since for each $i \in \{0, \dots, k\}$, $\leftpi(w_1) \leq \leftpi(w_i)<\leftpi(w_{i+1}) \leq \leftpi(v)$. Then, let us suppose that $z$ is in $z\in (\inside{F_e} \cap T_{w_1}) \setminus T_v\). In this case, there must exist a node $\alpha$ in $N(P_{uv}\setminus \{u,v\}) \cap \inside{F_e}$ such that $z\in T_\alpha$. Let us pick $i\in \{1, \dots, k\}$ such that $\alpha \in N(w_i)$. By \Cref{claim:p1},  $t_{w_i}(\alpha) > t_{w_i}(w_{i+1})$. Hence, in a \leftorder\ the nodes in $T_{\alpha}$ are visited after $w_{i}$ and before $w_{i+1}$. We deduce that $\leftpi(w_i) < \leftpi(\alpha) < \leftpi(z) < \leftpi(w_{i+1})$. Since for every $i \in \{1, \dots, k\}$ we have that $\leftpi(w_1)\leq \leftpi(w_i)$ and $\leftpi(w_{i+1}) \leq \leftpi(v)$, we conclude that $\leftpi(z)\in [\leftpi(w_1), \leftpi(v)]$.

Conversely, let us suppose that $\leftpi(z) \notin [\leftpi(w_1), \leftpi(v)-1]$ and let  us show that $z \notin (F_e \cap T_{w_1}) \setminus T_v$. Since each node in $P_{uv}\setminus{u}$ is visited after $w_1$ and before $v$, we have that $z$ cannot belong to the border. Let us show that $z$ cannot be  $(\inside{F_e} \cap T_{w_1}) \setminus T_v$ either. If $\leftpi(z) < \leftpi(w_1)$ or $\leftpi(z) > \leftpi(w_1) + n_T(w_1)$  we directly have that $z$ is not in $T_{w_1}$ and in particular $z \notin (\inside{F_e}  \cap T_{w_1}) \setminus T_v$. If $z \in [\leftpi(v), \leftpi(v) + n_T(v)]$, then $z$ must be contained in $T_v$, which implies that $z\notin  (\inside{F_e} \cap T_{w_1}) \setminus T_v$. It remains the case where $z \in [\leftpi(v) + n_T(v) + 1, \leftpi(w_1) + n_T(w_1)]$. In this case, there must exist and index $i\in \{1,\dots, k\}$ and a node $\alpha \in N(w_i)$ such that  $z\in T_\alpha$. If $\leftpi(\alpha) < \leftpi(w_{i+1})$ then we would have that in a \leftorder\ the nodes of $T_\alpha$ would be visited before $w_{i+1}$. In particular,  $z$ would be visited before $v$, contradicting the fact that $\leftpi(z) > \leftpi(v)$. We deduce that $t_{w_{i}}(\alpha) < t_{w_i}(w_{i+1})$. By \Cref{claim:p1} this implies that $\alpha \notin \inside{F_e}$. Moreover, by planarity no vertex of $T_{\alpha}$ is included in $\inside{F_e}$. In particular $z\notin (\inside{F_e}  \cap T_{w_1}) \setminus T_v$.
\end{proof}

    From \Cref{claim1p2} we deduce that the number of nodes in $(F_e \cap T_{w_1}) \setminus T_v$ equals $\leftpi(v) - \leftpi(w_1)$. From here we can obtain the number of nodes in $F_e$ by adding the number of nodes in $(T_u \cup T_v) \cap \inside{F_e}$, which equals $P_{F_e}(u) + P_{F_e}(v)$, plus $2$ as we need to include $u$ and $v$. Hence, the number of nodes in $F_e$ is $P_{F_e}(u) + P_{F_e}(v) + (\pi(v) - \pi(w_1)) + 2$ Finally, to obtain the number of nodes in $\inside{F_e}$ we have to subtract the nodes that are in the border of $F_e$, which equals to the number of nodes in the path $P_{uv}$. Since $v$ is a descendant of $v$, we have that this number corresponds to $d_T(v) - d_T(u) + 1 = d_T(v) - d_T(w_1) + 2$. We deduce that the number of vertices in $\inside{F_e}$ is $\omega\left(F_e\right)$.
\end{proof}
See \Cref{fig:weightancestors} and \Cref{fig:weightancestors2} to see a graphical example of \Cref{lem:weighting} and \Cref{lem:weighting2}, respectively. It is direct from definition that if a fundamental face $F_e$ is contained in other fundamental face $F_f$, then $\omega(F_e)\leq \omega(F_f)$. We say that $\omega$ is an increasing function for contained faces.

Finally, if a weight of any fundamental face is between the range $[n/3,2n/3]$, the $T$-path between the extremes of $e$ forms a separator set.

\begin{lemma}\label{teo:balancedface}
    Let $(G,\mathcal{E}, T)$ by a planar configuration and $F_e$ a fundamental face of $T$. If $\omega(F_e)\in [n/3, 2n/3]$ with $n=|V|$, then nodes in the $T$-path $P_e$ between the extremes of $e$ are a separator set of $G$. This path $P_e$ is a cycle separator of $G$.
\end{lemma}

\begin{proof}
    Let $F_e$ be a real fundamental face of $T$, $e=uv$ and $P_e$ be the $T$-path between $u$ and $v$. Then, given the combinatorial embedding $\mathcal{E}$, $G-P_e$ is divided in two disconnected (between them) zones  $\inside{F}_e^c$ and $\inside{F}_e$ as depicted in \Cref{fig:zonesinG}.

\begin{figure}[h!]
    \centering
    \scalebox{0.75}{\input{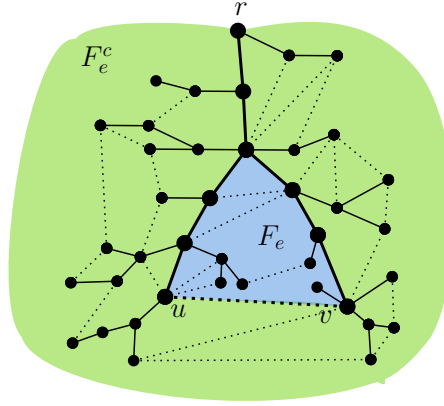}}
    \caption{$\inside{F}_e$ corresponds to the nodes inside $F_e$ in blue zone and $F_e^c$ to the nodes outside $F_e$ in green zone.}
    \label{fig:zonesinG}
\end{figure}

Each connected component of $G-P_e$ is either $\inside{F}_e$ or $\inside{F}_e^c$. As $\omega(F_e)$ upper bound the number of nodes in $\inside{F}_e$ (\Cref{lem:weighting,lem:weighting2}) and $\omega(F_e)\leq 2n/3$, then each connected component in $\inside{F}_e$ has at most $2n/3$ nodes. Since no node of $\inside{F}_e^c$ is counted in $\omega(F_e)$ and $\omega(F_e)\geq n/3$, then $\inside{F}_e^c\leq n -\omega(F_e)\leq 2n/3$, and then each connected component in $\inside{F}_e^c$ has at most $2n/3$ nodes. Then, each connected component of $G-P_e$ has at most $2n/3$ nodes and then the nodes of $P_e$ forms a separator set of $G$.

If $F_e\in \mathcal{F}_e$ for some $\mathcal{E}$-compatible virtual fundamental edge $e=uv$, then if we consider the embedding $\mathcal{E}'$ given by $\mathcal{E}$ plus the edge $e$ inserted in $\mathcal{E}$ such that $F_e$ is formed in $G'=(V,E+e)$ in the embedding $\mathcal{E}'$, then by analogous argument, $G'$ is divided in two zones in $\mathcal{E}'$, the nodes inside $F_e$ and the nodes outside $F_e$. As $\omega(F_e)\in [n/3,2n/3]$, by the same above argument each connected component of $G'- P_e$ has size at most $2n/3$, and finally, each connected component in $G-P_e$ is contained in a connected component of $G'-P_e$. Then, we conclude that the size of each connected component of $G-P_e$ has also at most $2n/3$ nodes and therefore, the nodes in $P_e$ are also a separator set of $G$.
\end{proof}

\subsection{Properties and characterization of real fundamental faces.}\label{prop:insidenotchildren}

It is possible to characterize the nodes inside an specific real fundamental face through the left and right DFS orders of the trees.

\begin{remark}\label{prop:caracinside}
    Given a planar configuration $(G,\mathcal{E},T)$, the left $\leftpi^T$ and right $\rightpi^T$ DFS order of $T$ and a $T$-real fundamental face $F_e$, defined by a $T$-real fundamental edge $e=uv$ such that $\leftpi(u)<\leftpi(v)$, a node $z\notin V(T_u)\cap V(T_v)$ is inside $F_e$ if and only if one of the following conditions holds
    \begin{enumerate}
        \item Node $u$ is not ancestor of $v$ and $\leftpi(u)<\leftpi(z)<\leftpi(v)$.
        \item Node $u$ is ancestor of $v$, $e$ is a $\mathcal{E}$-left oriented edge and $\leftpi(u)<\leftpi(z)<\leftpi(v)$.
        \item Node $u$ is ancestor of $v$, $e$ is a $\mathcal{E}$-right oriented edge and $\rightpi(u)<\rightpi(z)<\rightpi(v)$.
    \end{enumerate}
\end{remark}

\subsection{Properties of augmented faces.}

In the distributed algorithm designed to compute cycle separators, it is necessary at some point to verify whether certain nodes are \((T, F_e)\)-compatible for a given fundamental face \(F_e\). The nodes that need to be verified for \((T, F_e)\)-compatibility are the leaves of \(T\) inside $F_e$. Therefore, we characterize these nodes, which are \((T, F_e)\)-compatible with one endpoint of \(e\).

\begin{definition}\label{def:contained}
    Given a planar configuration $(G,\mathcal{E},T)$ and a fundamental face $F_e$, we say that a  node $z\in \inside{F}_e$  is {\it hidden} in $F_e$ if there exist other real edge $f=z_1z_2$ contained in $F_e$ such that $z$ is also inside $F_f$, i.e., $z\in \inside{F}_f$, and one of the following conditions holds
    \begin{enumerate}
        \item $z_1,z_2\neq u$
        \item $z_1 = u$ or $z_2 = u$ and $V(T_u)\cap V(F_e)\not \subseteq V(F_f)$.
    \end{enumerate}

    In this case we say that the real fundamental edge $f$ \emph{hides} $z$. 
\end{definition}

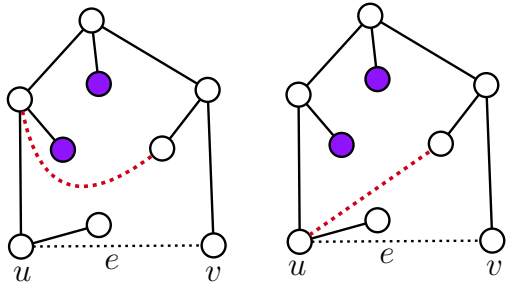
\begin{figure}
    \centering
    \scalebox{0.6}{\tikzset{every picture/.style={line width=0.75pt}} 

\begin{tikzpicture}[x=0.75pt,y=0.75pt,yscale=-1,xscale=1]

\draw  [color={rgb, 255:red, 0; green, 0; blue, 0 }  ,draw opacity=1 ][line width=1.5]   (456, 39) circle [x radius= 10, y radius= 10]   ;
\draw (450,36.4) node [anchor=north west][inner sep=0.75pt]  [font=\tiny]  {$$};
\draw  [color={rgb, 255:red, 0; green, 0; blue, 0 }  ,draw opacity=1 ][line width=1.5]   (553, 97) circle [x radius= 10, y radius= 10]   ;
\draw (547,94.4) node [anchor=north west][inner sep=0.75pt]  [font=\tiny]  {$$};
\draw  [color={rgb, 255:red, 0; green, 0; blue, 0 }  ,draw opacity=1 ][line width=1.5]   (396, 105) circle [x radius= 10, y radius= 10]   ;
\draw (390,102.4) node [anchor=north west][inner sep=0.75pt]  [font=\tiny]  {$$};
\draw  [color={rgb, 255:red, 0; green, 0; blue, 0 }  ,draw opacity=1 ][line width=1.5]   (397, 228) circle [x radius= 10, y radius= 10]   ;
\draw (391,225.4) node [anchor=north west][inner sep=0.75pt]  [font=\tiny]  {$$};
\draw  [color={rgb, 255:red, 0; green, 0; blue, 0 }  ,draw opacity=1 ][fill={rgb, 255:red, 144; green, 19; blue, 254 }  ,fill opacity=1 ][line width=1.5]   (432, 147) circle [x radius= 10, y radius= 10]   ;
\draw (426,144.4) node [anchor=north west][inner sep=0.75pt]  [font=\tiny]  {$$};
\draw  [color={rgb, 255:red, 0; green, 0; blue, 0 }  ,draw opacity=1 ][fill={rgb, 255:red, 144; green, 19; blue, 254 }  ,fill opacity=1 ][line width=1.5]   (462, 92) circle [x radius= 10, y radius= 10]   ;
\draw (456,89.4) node [anchor=north west][inner sep=0.75pt]  [font=\tiny]  {$$};
\draw  [color={rgb, 255:red, 0; green, 0; blue, 0 }  ,draw opacity=1 ][line width=1.5]   (558, 227) circle [x radius= 10, y radius= 10]   ;
\draw (552,224.4) node [anchor=north west][inner sep=0.75pt]  [font=\tiny]  {$$};
\draw  [color={rgb, 255:red, 0; green, 0; blue, 0 }  ,draw opacity=1 ][line width=1.5]   (515, 146) circle [x radius= 10, y radius= 10]   ;
\draw (509,143.4) node [anchor=north west][inner sep=0.75pt]  [font=\tiny]  {$$};
\draw  [color={rgb, 255:red, 0; green, 0; blue, 0 }  ,draw opacity=1 ][line width=1.5]   (462, 210) circle [x radius= 10, y radius= 10]   ;
\draw (456,207.4) node [anchor=north west][inner sep=0.75pt]  [font=\tiny]  {$$};
\draw (395,250) node {\huge $u$};
\draw (560,250) node {\huge $v$};
\draw (470,240) node  {\huge $e$};
\draw  [color={rgb, 255:red, 0; green, 0; blue, 0 }  ,draw opacity=1 ][line width=1.5]   (223, 43) circle [x radius= 10, y radius= 10]   ;
\draw (217,40.4) node [anchor=north west][inner sep=0.75pt]  [font=\tiny]  {$$};
\draw  [color={rgb, 255:red, 0; green, 0; blue, 0 }  ,draw opacity=1 ][line width=1.5]   (320, 101) circle [x radius= 10, y radius= 10]   ;
\draw (314,98.4) node [anchor=north west][inner sep=0.75pt]  [font=\tiny]  {$$};
\draw  [color={rgb, 255:red, 0; green, 0; blue, 0 }  ,draw opacity=1 ][line width=1.5]   (163, 109) circle [x radius= 10, y radius= 10]   ;
\draw (157,106.4) node [anchor=north west][inner sep=0.75pt]  [font=\tiny]  {$$};
\draw  [color={rgb, 255:red, 0; green, 0; blue, 0 }  ,draw opacity=1 ][line width=1.5]   (164, 232) circle [x radius= 10, y radius= 10]   ;
\draw (158,229.4) node [anchor=north west][inner sep=0.75pt]  [font=\tiny]  {$$};
\draw  [color={rgb, 255:red, 0; green, 0; blue, 0 }  ,draw opacity=1 ][fill={rgb, 255:red, 144; green, 19; blue, 254 }  ,fill opacity=1 ][line width=1.5]   (199, 151) circle [x radius= 10, y radius= 10]   ;
\draw (193,148.4) node [anchor=north west][inner sep=0.75pt]  [font=\tiny]  {$$};
\draw  [color={rgb, 255:red, 0; green, 0; blue, 0 }  ,draw opacity=1 ][fill={rgb, 255:red, 144; green, 19; blue, 254 }  ,fill opacity=1 ][line width=1.5]   (229, 96) circle [x radius= 10, y radius= 10]   ;
\draw (223,93.4) node [anchor=north west][inner sep=0.75pt]  [font=\tiny]  {$$};
\draw  [color={rgb, 255:red, 0; green, 0; blue, 0 }  ,draw opacity=1 ][line width=1.5]   (325, 231) circle [x radius= 10, y radius= 10]   ;
\draw (319,228.4) node [anchor=north west][inner sep=0.75pt]  [font=\tiny]  {$$};
\draw  [color={rgb, 255:red, 0; green, 0; blue, 0 }  ,draw opacity=1 ][line width=1.5]   (282, 150) circle [x radius= 10, y radius= 10]   ;
\draw (276,147.4) node [anchor=north west][inner sep=0.75pt]  [font=\tiny]  {$$};
\draw  [color={rgb, 255:red, 0; green, 0; blue, 0 }  ,draw opacity=1 ][line width=1.5]   (229, 214) circle [x radius= 10, y radius= 10]   ;
\draw (223,211.4) node [anchor=north west][inner sep=0.75pt]  [font=\tiny]  {$$};
\draw (165,255) node  {\huge $u$};
\draw (325,255) node  {\huge $v$};
\draw (240,245) node {\huge $e$};
\draw [color={rgb, 255:red, 0; green, 0; blue, 0 }  ,draw opacity=1 ][line width=1.5]    (449.27,46.4) -- (402.73,97.6) ;
\draw [color={rgb, 255:red, 0; green, 0; blue, 0 }  ,draw opacity=1 ][line width=1.5]    (402.51,112.59) -- (425.49,139.41) ;
\draw [color={rgb, 255:red, 0; green, 0; blue, 0 }  ,draw opacity=1 ][line width=1.5]    (396.08,115) -- (396.92,218) ;
\draw [color={rgb, 255:red, 0; green, 0; blue, 0 }  ,draw opacity=1 ][line width=1.5]    (457.12,48.94) -- (460.88,82.06) ;
\draw [color={rgb, 255:red, 0; green, 0; blue, 0 }  ,draw opacity=1 ][line width=1.5]    (464.58,44.13) -- (544.42,91.87) ;
\draw [color={rgb, 255:red, 0; green, 0; blue, 0 }  ,draw opacity=1 ][line width=1.5]    (546.87,104.9) -- (521.13,138.1) ;
\draw [color={rgb, 255:red, 0; green, 0; blue, 0 }  ,draw opacity=1 ][line width=1.5]    (553.38,106.99) -- (557.62,217.01) ;
\draw [color={rgb, 255:red, 0; green, 0; blue, 0 }  ,draw opacity=1 ][line width=1.5]  [dash pattern={on 1.69pt off 2.76pt}]  (407,227.94) -- (548,227.06) ;
\draw [color={rgb, 255:red, 208; green, 2; blue, 27 }  ,draw opacity=1 ][line width=2.25]  [dash pattern={on 2.53pt off 3.02pt}]  (405.21,222.29) -- (506.79,151.71) ;
\draw [color={rgb, 255:red, 0; green, 0; blue, 0 }  ,draw opacity=1 ][line width=1.5]    (452.36,212.67) -- (406.64,225.33) ;
\draw [color={rgb, 255:red, 0; green, 0; blue, 0 }  ,draw opacity=1 ][line width=1.5]    (216.27,50.4) -- (169.73,101.6) ;
\draw [color={rgb, 255:red, 0; green, 0; blue, 0 }  ,draw opacity=1 ][line width=1.5]    (169.51,116.59) -- (192.49,143.41) ;
\draw [color={rgb, 255:red, 0; green, 0; blue, 0 }  ,draw opacity=1 ][line width=1.5]    (163.08,119) -- (163.92,222) ;
\draw [color={rgb, 255:red, 0; green, 0; blue, 0 }  ,draw opacity=1 ][line width=1.5]    (224.12,52.94) -- (227.88,86.06) ;
\draw [color={rgb, 255:red, 0; green, 0; blue, 0 }  ,draw opacity=1 ][line width=1.5]    (231.58,48.13) -- (311.42,95.87) ;
\draw [color={rgb, 255:red, 0; green, 0; blue, 0 }  ,draw opacity=1 ][line width=1.5]    (313.87,108.9) -- (288.13,142.1) ;
\draw [color={rgb, 255:red, 0; green, 0; blue, 0 }  ,draw opacity=1 ][line width=1.5]    (320.38,110.99) -- (324.62,221.01) ;
\draw [color={rgb, 255:red, 0; green, 0; blue, 0 }  ,draw opacity=1 ][line width=1.5]  [dash pattern={on 1.69pt off 2.76pt}]  (174,231.94) -- (315,231.06) ;
\draw [color={rgb, 255:red, 0; green, 0; blue, 0 }  ,draw opacity=1 ][line width=1.5]    (219.36,216.67) -- (173.64,229.33) ;
\draw [color={rgb, 255:red, 208; green, 2; blue, 27 }  ,draw opacity=1 ][line width=2.25]  [dash pattern={on 2.53pt off 3.02pt}]  (165.21,118.76) .. controls (179.06,187.59) and (215.29,199.95) .. (273.89,155.84) ;

\end{tikzpicture}}
    \caption{In both images, purple nodes are covered in $F_e$. To the left the dotted red edge hides the purple nodes and satisfy condition 1. of \Cref{def:contained}. To the right the dotted red edge hides the purple nodes and satisfy condition 2.}
    \label{fig:enter-label}
\end{figure}

Intuitively, let \(F_e\) be the face associated with a \(T\)-fundamental edge \(e=uv\). Any node hidden inside \(F_e\) according to the above definition cannot be connected to \(u\) by an edge without violating planarity. Since our algorithm later simulates edges between the endpoints of a fundamental edge and nodes inside the corresponding face, we must characterize the nodes in \(F_e\) that are \emph{not} hidden; namely, those that can be connected to \(u\) without breaking planarity. In the following lemma, it is shown that the notions of being (not) hidden in a fundamental face $F_e$ for a leaf of $T$ characterize the condition of being $(T,F_e)$-compatible inside a face.

\begin{lemma}\label{lemma:deltacharac}

Given a planar configuration $(G,\mathcal{E}, T)$ and a fundamental face $F_e$ with $e=uv$ and $\leftpi(u)<\leftpi(v)$. A node $z\in \inside{F}_e$ and is a leaf of $T$ is $(T,F_e)$-compatible with node $u$  if and only if it is not hidden in~$F_e$.
\end{lemma}

\begin{proof}

If \(z\) is a leaf of \(T\), is inside \(F_e\), and is hidden, then there exists a fundamental edge \(f\) contained in \(F_e\) such that \(z\) is inside \(F_f\) and satisfies condition (1) or (2) of \Cref{def:contained}. If condition (1) is met, nodes \(z\) and \(u\) are not \(\mathcal{E}\)-compatible, as any way to insert \(f' = uz\) into \(\mathcal{E}\) will cross edge \(f\), breaking planarity. If condition (2) is met, there exists a node \(u' \in V(T_u)\) that is not \(\mathcal{E}\)-compatible with \(z\). In this case, it is not possible to insert \(f' = uz\) such that all the \(T\)-children of \(u\) inside \(F_e\) are also inside \(F_f\), because this edge crosses \(f\) due to condition (2).

On the other hand, if \(z\) and \(u\) are not \((T, F_e)\)-compatible, then when inserting \(f' = uz\) such that all nodes in \(V(T_u) \cap V(\inside{F}_e)\) are within \(F_{f'}\), a fundamental edge \(f'' = u''v''\) contained in \(F_e\) is crossed. Since \(z\) is a leaf in \(T\), \(z\) must be inside \(F_{f''}\). If \(u''\) and \(v'' \neq u\), then condition (1) of \Cref{def:contained} is met. If \(u'' = u\) or \(v'' = u\), then it must be that \(V(T_u) \cap V(F_e) \not\subseteq V(F_{f''})\), because otherwise \(f'\) would not cross \(f''\), thereby satisfying condition (2).
\end{proof}

Given a planar configuration \((G, \mathcal{E}, T)\) and a fundamental face \(F_e\) of an edge \(e = uv\), a \textit{full augmentation from \(u\)} corresponds to adding all the edges to \(\mathcal{E}\) from \(u\) to the rest of the nodes \(z\) in \(F_e\) such that all the edges are contained in $F_e$ and all the \(T\)-children of both nodes \(u\) and \(z\) that are in \(F_e\), are also in the face defined by this insertion of $f=uz$. For nodes \((T, F_e)\)-compatible with \(u\), this augmentation creates the same face \(F_{uz}^\ell\) as defined by \((T, F_e)\)-compatibility. For nodes \(z\) that are not \((T, F_e)\)-compatible with \(u\), we abuse notation and also refer to the face \(F_{uz}\) created by this full augmentation as \(F_{uz}^\ell\). See \Cref{fig:fullaug}.

\begin{figure}[h!]
    \centering
    \resizebox{\textwidth}{!}{\input{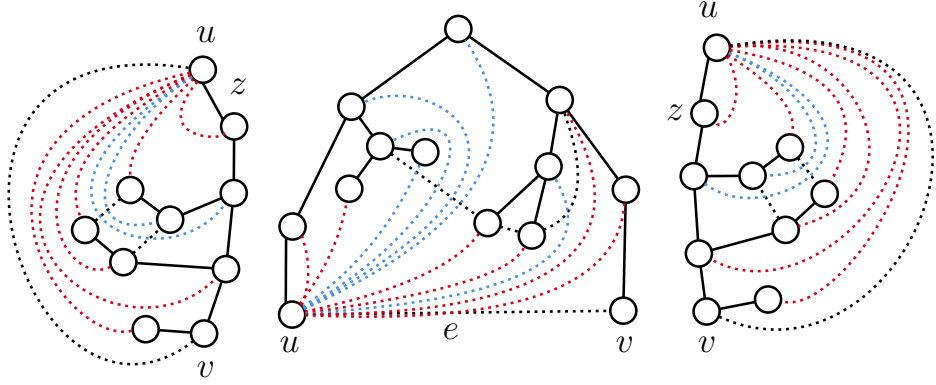}}
    \caption{Full augmentation from a node $u$. Nodes connected with red dotted deges are $(T,F_e)$-compatible, while the ones connected with blue dotted edges  are not.}
    \label{fig:fullaug}
\end{figure}

\begin{remark}\label{lemma:propaug}
    Given a planar configuration $(G,\mathcal{E},T)$ and a $T$-real fundamental face $F_e$ of a $T$-real fundamental edge $e=uv$ such that $\leftpi^T(u)<\leftpi^T(v)$, the following properties holds.

    \begin{enumerate}

        \item Given two nodes $z_1,z_2\in \inside{F}_e$ such that $z_1$ neither is ancestor of $z_2$ nor vice versa, then if $e$ is such that $u$ is not ancestor of $v$, or $e$ is $\mathcal{E}$-left oriented edge, then $\omega\left(F_{uz_1}^\ell\right)\leq  \omega\left(F_{uz_2}^\ell\right)$ iff $\leftpi^T(z_1)\leq \leftpi^T(z_2)$.
        \item Given two nodes $z_1,z_2\in \inside{F}_e$ such that $z_1$ neither is ancestor of $z_2$ nor vice versa, then if $e$ is $\mathcal{E}$-right oriented edge, then $\omega\left(F_{uz_1}\right)\leq  \omega\left(F_{uz_2}\right)$ iff $\rightpi^T(z_1)\leq \rightpi^T(z_2)$.

        \item If $e$ is such that $u$ is not ancestor of $v$, or $e$ is $\mathcal{E}$-left oriented edge, then given a node $z\in \inside{F}_e$, the $T_i$-leaf $t$ that is a descendant of $z$ in $T$ and has the higher position in $\leftpi^T$ satisfy that $\omega\left(F_{uz}\right)=  \omega\left(F_{ut}\right)$
        \item If $e$ is $\mathcal{E}$-right oriented edge, then given a node $z\in \inside{F}_e$, the $T_i$-leaf $t$ that is a descendant of $z$ in $T$ and has the higher position in $\rightpi^T$ satisfy that $\omega\left(F_{uz}\right)=  \omega\left(F_{ut}\right)$.
        
    \end{enumerate}
\end{remark}

\subsection{Proof of \Cref{teo:fundamentalseparator}.}\label{subsec:proofteo}

Recall that \Cref{teo:fundamentalseparator} states that in any planar configuration $(G,\mathcal E,T)$ there exists a (real or virtual) fundamental edge $e$ that defines a cycle separator, and satisfies one of the properties stated in the Lemma.
\Cref{lem:weighting,lem:weighting2} state that our weighting function does not always compute the exact number of nodes inside $F_e$, but instead the number of nodes in $\tilde{F}_e$ when $u$ and $v$ are not descendants of each other. Nevertheless, this is enough to prove \Cref{teo:fundamentalseparator}, but first we prove two lemmas that will be useful.

\begin{lemma}
    \label{lemma:balancedincritical} Given a planar configuration $(G,\mathcal{E},T)$ such that $\omega(F_e)\notin[n/3,2n/3]$ for every real  fundamental face of \(T\), there exists  a $T$-real fundamental edge $e=uv$ such that $\omega(F_e)>2n/3$ and $\omega(F_f)<n/3$ for all real fundamental edges \(f\) contained in $F_e$. There exists a cycle separator of $G$ that is a path of $T$ (plus a non-$T$ edge between the extreme nodes of the path). 
\end{lemma}

\begin{proof}
   Let \(e = uv\) be a real fundamental edge as described in the Lemma such that \(\leftpi(u) < \leftpi(v)\). We split the proof into two cases.

Suppose first that there exists a node \(t\) inside \(F_e\) such that \(\omega(F_{ut}^\ell) \in [n/3, 2n/3]\). Consider \(t\) such that it has the highest position in \(\leftpi\) and \(\omega(F_{ut}^\ell) \in [n/3, 2n/3]\). By \Cref{lemma:propaug}, \(t\) is a leaf of \(T\) (if \(t\) is not a leaf, then rightmost child of \(t\) in $T_t$ has the same weight and a greater left position in \(\leftpi\)). In this case, if \(t\) is \((T, F_e)\)-compatible with \(u\), then by \Cref{teo:balancedface} we conclude that the \(T\)-path between \(u\) and \(t\) is a cycle separator of \(G\).

Now, we should argue how to proceed when all the leaves \(t\) such that \(\omega(F_{ut}^\ell) \in [n/3, 2n/3]\) are not \((T, F_e)\)-compatible with \(u\). In fact, we assume a stronger condition: Assume that all the leaves \(t\) such that \(\omega(F_{ut}^\ell) \geq n/3\) are not \((T, F_e)\)-compatible with \(u\).

Let \(t\) be such a leaf with leftmost position in \(\leftpi\) (i.e., lowest position in $\leftpi$). As we assume that that node $t$ is not  \((T, F_e)\)-compatible with \(u\), then by \Cref{lemma:deltacharac}, there exists a $T$-real fundamental edge \(f\) such that \(t\) is inside \(F_f\) and \(f\) is contained in \(F_e\). Take \(f = z_1z_2\) as the edge that satisfies this condition and $f$ is not contained in any other real fundamental edge \(f'\) (inside $F_e$) that also hides \(t\). See \Cref{fig:4pathseparator}.

    \begin{claim}
        Let $P_{uz_2}$ the $T$-path between $u$ and $z_2$. The nodes in $P_{uz_2}$ are a cycle separator of $G$.
    \end{claim}

    \begin{figure}[h!]
        \centering
        \scalebox{0.7}{\tikzset{every picture/.style={line width=0.75pt}} 

\begin{tikzpicture}[x=0.75pt,y=0.75pt,yscale=-1,xscale=1]

\draw [color={rgb, 255:red, 208; green, 2; blue, 27 }  ,draw opacity=1 ][line width=1.5]    (426.27,58.54) .. controls (487.24,48.4) and (556.22,151.39) .. (539.88,234.09) ;
\draw [shift={(539.88,234.09)}, rotate = 101.18] [color={rgb, 255:red, 208; green, 2; blue, 27 }  ,draw opacity=1 ][fill={rgb, 255:red, 208; green, 2; blue, 27 }  ,fill opacity=1 ][line width=1.5]      (0, 0) circle [x radius= 4.36, y radius= 4.36]   ;
\draw [shift={(426.27,58.54)}, rotate = 350.56] [color={rgb, 255:red, 208; green, 2; blue, 27 }  ,draw opacity=1 ][fill={rgb, 255:red, 208; green, 2; blue, 27 }  ,fill opacity=1 ][line width=1.5]      (0, 0) circle [x radius= 4.36, y radius= 4.36]   ;
\draw [color={rgb, 255:red, 208; green, 2; blue, 27 }  ,draw opacity=1 ][line width=1.5]    (426.27,58.54) .. controls (371.06,70.25) and (309.33,165.43) .. (309.33,235.65) ;
\draw [shift={(309.33,235.65)}, rotate = 90] [color={rgb, 255:red, 208; green, 2; blue, 27 }  ,draw opacity=1 ][fill={rgb, 255:red, 208; green, 2; blue, 27 }  ,fill opacity=1 ][line width=1.5]      (0, 0) circle [x radius= 4.36, y radius= 4.36]   ;
\draw [color={rgb, 255:red, 208; green, 2; blue, 27 }  ,draw opacity=1 ][line width=1.5]    (490.87,84.29) .. controls (448.57,109.61) and (397.85,134.49) .. (399.67,150.09) ;
\draw [shift={(399.67,150.09)}, rotate = 83.36] [color={rgb, 255:red, 208; green, 2; blue, 27 }  ,draw opacity=1 ][fill={rgb, 255:red, 208; green, 2; blue, 27 }  ,fill opacity=1 ][line width=1.5]      (0, 0) circle [x radius= 4.36, y radius= 4.36]   ;
\draw [shift={(490.87,84.29)}, rotate = 149.1] [color={rgb, 255:red, 208; green, 2; blue, 27 }  ,draw opacity=1 ][fill={rgb, 255:red, 208; green, 2; blue, 27 }  ,fill opacity=1 ][line width=1.5]      (0, 0) circle [x radius= 4.36, y radius= 4.36]   ;
\draw [color={rgb, 255:red, 208; green, 2; blue, 27 }  ,draw opacity=1 ][line width=1.5]    (482.47,159.4) .. controls (467.95,145.36) and (471.63,108.29) .. (444.4,111.41) ;
\draw [shift={(444.4,111.41)}, rotate = 173.46] [color={rgb, 255:red, 208; green, 2; blue, 27 }  ,draw opacity=1 ][fill={rgb, 255:red, 208; green, 2; blue, 27 }  ,fill opacity=1 ][line width=1.5]      (0, 0) circle [x radius= 4.36, y radius= 4.36]   ;
\draw [shift={(482.47,159.4)}, rotate = 224.04] [color={rgb, 255:red, 208; green, 2; blue, 27 }  ,draw opacity=1 ][fill={rgb, 255:red, 208; green, 2; blue, 27 }  ,fill opacity=1 ][line width=1.5]      (0, 0) circle [x radius= 4.36, y radius= 4.36]   ;
\draw [color={rgb, 255:red, 208; green, 2; blue, 27 }  ,draw opacity=1 ] [dash pattern={on 0.84pt off 2.51pt}]  (539.88,234.09) .. controls (484.67,245.8) and (347.46,237.21) .. (309.33,235.65) ;
\draw [shift={(309.33,235.65)}, rotate = 182.34] [color={rgb, 255:red, 208; green, 2; blue, 27 }  ,draw opacity=1 ][fill={rgb, 255:red, 208; green, 2; blue, 27 }  ,fill opacity=1 ][line width=0.75]      (0, 0) circle [x radius= 3.35, y radius= 3.35]   ;
\draw [color={rgb, 255:red, 208; green, 2; blue, 27 }  ,draw opacity=1 ][line width=1.5]  [dash pattern={on 1.69pt off 2.76pt}]  (482.86,157.79) .. controls (427.65,169.5) and (434.86,170.79) .. (399.67,150.09) ;
\draw [shift={(399.67,150.09)}, rotate = 210.45] [color={rgb, 255:red, 208; green, 2; blue, 27 }  ,draw opacity=1 ][fill={rgb, 255:red, 208; green, 2; blue, 27 }  ,fill opacity=1 ][line width=1.5]      (0, 0) circle [x radius= 4.36, y radius= 4.36]   ;
\draw  [color={rgb, 255:red, 208; green, 2; blue, 27 }  ,draw opacity=1 ][fill={rgb, 255:red, 184; green, 233; blue, 134 }  ,fill opacity=1 ] (399.47,151.4) -- (416.71,194.09) -- (382.23,194.09) -- cycle ;
\draw  [color={rgb, 255:red, 208; green, 2; blue, 27 }  ,draw opacity=1 ][fill={rgb, 255:red, 184; green, 233; blue, 134 }  ,fill opacity=1 ] (482.47,159.4) -- (499.71,202.09) -- (465.23,202.09) -- cycle ;
\draw [color={rgb, 255:red, 74; green, 144; blue, 226 }  ,draw opacity=1 ][line width=1.5]  [dash pattern={on 1.69pt off 2.76pt}]  (315.33,233.09) .. controls (377.67,213.54) and (470.67,231.09) .. (399.47,151.4) ;
\draw [color={rgb, 255:red, 74; green, 144; blue, 226 }  ,draw opacity=1 ][line width=1.5]  [dash pattern={on 1.69pt off 2.76pt}]  (315.33,233.09) .. controls (393.33,235.09) and (467.43,204.09) .. (478.33,163.09) ;
\draw [color={rgb, 255:red, 208; green, 2; blue, 27 }  ,draw opacity=1 ]   (444.4,111.41) .. controls (436.33,124.09) and (461.33,132.09) .. (436.67,143.09) ;
\draw [shift={(436.67,143.09)}, rotate = 155.97] [color={rgb, 255:red, 208; green, 2; blue, 27 }  ,draw opacity=1 ][fill={rgb, 255:red, 208; green, 2; blue, 27 }  ,fill opacity=1 ][line width=0.75]      (0, 0) circle [x radius= 3.35, y radius= 3.35]   ;

\draw (420.54,25.37) node [anchor=north west][inner sep=0.75pt]  [font=\normalsize]  {$\omega _{1}$};
\draw (314.41,240.93) node [anchor=north west][inner sep=0.75pt]  [font=\normalsize]  {$u$};
\draw (544.96,239.37) node [anchor=north west][inner sep=0.75pt]  [font=\normalsize]  {$v$};
\draw (376.67,133.44) node [anchor=north west][inner sep=0.75pt]  [font=\normalsize]  {$z_{1}$};
\draw (488.96,141.2) node [anchor=north west][inner sep=0.75pt]  [font=\normalsize]  {$z_{2}$};
\draw (390,172) node [anchor=north west][inner sep=0.75pt]   [align=left] {$\displaystyle T_{z_{1}}$};
\draw (473,180) node [anchor=north west][inner sep=0.75pt]   [align=left] {$\displaystyle T_{z_{2}}$};
\draw (426,133) node [anchor=north west][inner sep=0.75pt]   [align=left] {$\displaystyle t$};
\draw (442,166) node [anchor=north west][inner sep=0.75pt]   [align=left] {$\displaystyle f$};
\draw (430,238) node [anchor=north west][inner sep=0.75pt]   [align=left] {$\displaystyle e$};

\end{tikzpicture}}
        \caption{ The two possible cycle separators of Claim 6. Either the $T$-path between $u$ and $z_1$ or the $T$-path between $u$ and $t_f$. }
        \label{fig:4pathseparator}
    \end{figure}

    \begin{proof}
    Define $\hat{\omega}\left(F_{uz_1}\right) = \omega\left(F_{uz_1}\right) - |F_{u,z_1}\cap F_f| $, i.e., $\hat{\omega}\left(F_{u,z_1}\right)$ counts the number of nodes in $F_{uz_1}$ that are not part of $F_f$. 

    First, suppose that $\hat{\omega}\left(F_{u,z_1}\right)\leq n/3$. We prove that in this case $P_{uz_2}$ is a cycle separator of $G$. Assuming that $\hat{\omega}\left(F_{uz_1}\right)\leq n/3$, notice that $\omega\left(F_{ut}\right)\leq \omega\left(F_{uz_2}\right)$, and therefore $\omega\left(F_{uz_2}\right)\geq n/3$.

By definition of $f =z_1z_2$ \footnote{$f$ is not contained in any other real fundamental face that hides $t$}, the virtual edge $uz_2$ is $\mathcal{E}$-compatible if it is inserted in $\mathcal{E}$ such that all the nodes in $T_{z_2}\cap \overline{F_f}$ are outside $F_{uz_2}$ (blue dashed edge between $u$ and $z_2$ in \Cref{fig:4pathseparator}). Therefore, \(G - P_{uz_2}\) is split into two distinct disconnected parts: \(P_1\) and \(P_2\). Here, \(P_1\) corresponds to the nodes outside \(F_{uz_2}\), when $uz_2$ is inserted in $\mathcal{E}$ as explained before, and \(P_2\) to the nodes inside \(F_{uz_2}\). 

Again, by election of edge $f$, there cannot be any $T$-real fundamental edge between a node in $P_2$ and a node in $P_1$.  We prove that the number of nodes in each zone is at most \(2n/3\), to conclude that \(T_{u,z_2}\) is a cycle separator of \(G\).

First, we note that since no node in \(P_1\) is counted in \(\omega(F_{uz_2})\) (according to \Cref{lem:weighting}) and \(\omega\left(F_{uz_2}\right) \geq n/3\), then \(|P_1| \leq n - n/3 \leq 2n/3\). Second, we note that the nodes in \(P_2\) can be divided into two disconnected subset of nodes, \(P_2^1\) and \(P_2^2\). In \(P_2^1\) are the nodes of \(P_2\) inside $F_f$, and in \(P_2^2\) are the nodes of $P_2$ that are outside $F_f$. On one hand, as $P_2^1$ is inside $F_f$ and $\omega\left(F_f\right)\leq n/3$, then $|P_2^1|\leq n/3$. On the other hand, all the nodes in $P_2^2$ are  counted in $\hat{\omega}\left(F_{uz_1}\right)$, therefore $|P_2^2|\leq n/3$. We conclude that $|P_2| = |P_2^1| + |P_2^2|\leq 2n/3$ and then $T_{uz_2}$ is a cycle separator of $G$, with $uz_2$ simulated as explained above.

To conclude the proof, we show that it must necessarily hold that $\hat{\omega}\left(F_{uz_1}\right)< n/3$.

Indeed, assuming for contradiction that $\hat{\omega}\left(F_{uz_1}\right)\geq n/3$, and defining $c_\ell$ as the rightmost leaf in $T_{z_1}\cap V(F_f)^c$ such that $\leftpi(c_\ell)<\leftpi(t)$, by construction we have $\omega\left(F_{uc_\ell}\right) = \hat{\omega}\left(F_{uz_1}\right)$, and hence $\omega\left(F_{uc_\ell}\right)\geq n/3$, contradicting the choice of $t$.

Thus, it must necessarily hold that $\hat{\omega}\left(F_{uz_1}\right)< n/3$, and therefore the nodes of $P_{uz_2}$ form a cycle separator of $G$.
\end{proof}
    
In this point we have prove that if exists a leaf $t$ inside $F_e$ such that $\omega\left(F_{ut}\right)\geq n/3$, there exists a cycle separator of $G$, regardless $t$ is $(T,\mathcal{E})$-compatible with node $u$.

Now analyzing the other case, suppose that \(\omega(F_{ut}^\ell) < n/3\) for all leaves \(t\) inside \(F_e\) in the full augmentation of \(F_e\). This case is split into two subcases.\\

{\it Subcase 1.} Suppose that there is at most one node inside \(F_e\), i.e., \(|\inside{F}_e| \leq 1\). Then, in this case, as \(\omega(F_e) > 2n/3\) and \(e = uv\), \(u\) cannot be an ancestor of \(v\), because in that case, \(\omega(F_e)\) counts exactly the nodes inside \(F_e\), so \(\omega(F_e) \leq 1\), which is a contradiction. Therefore, \(u\) and \(v\) are not descendants of each other, and since \(|\inside{F}_e| \leq 1\), \(\omega(F_e)\) counts the number of nodes in the \(T\)-path between \(LCA(u, v)\) and \(v\), plus at most the only node inside \(F_e\). Since \(\omega(F_e) > 2n/3\), this \(T\)-path between \(LCA(u, v)\) and \(v\) has at least \(2n/3 - 1 > n/3\) nodes, making it a cycle separator of \(G\).\\

For the final subcase, we define \(z_1, \ldots, z_k\) as the nodes on the boundary of \(F_e\)\footnote{These nodes are exactly the nodes in the $T$-path between $u$ and $v$}. For each of these nodes, we define the set \(C(z_j)\) as the set of nodes $s$ inside \(F_e\), such that deepest ancestor of $v$ (in $T$) that is on the boundary of \(F_e\) is \(z_j\).

\begin{remark}\label{remark:partition}
    Some of the sets $C(z_j)$ might be empty. The set of non-empty sets $C(z_j)$ induces a partition of $\inside{F_e}$.
\end{remark}

    {\it Subcase 2.} For all the leaves $t$ inside  $F_e$, $\omega(F_{uz}^\ell) < n/3$ and there exists at least two leaves inside $F_e$.
    
In this case, \( u \) cannot be an ancestor of \( v \) or vice versa, as \(\omega(F_e)\) counts exactly the nodes inside \( F_e \), and since \(\omega(F_e) > \frac{2n}{3}\), the number of nodes inside \( F_e \) is at least \(\frac{2n}{3}\). However, if \(\omega(F_{uz}^\ell) < \frac{n}{3}\) for all nodes \( z \) inside \( F_e \), then the node \( v_r \) with higher position in $\rightpi$ and the node \( v_\ell\) with higher position in $\leftpi$ that are a leaf in \( T \) and inside \( F_e \), fulfill that \(\omega(F_{uv_r}^\ell)\) and \(\omega(F_{uv_\ell}^\ell)\) count all the nodes inside \( F_e \), if \( e \) is an \(\mathcal{E}\)-right oriented edge or an \(\mathcal{E}\)-left oriented edge, respectively. Since \(\omega(F_{uv_r}^\ell), \omega(F_{uv_\ell}^\ell) < \frac{n}{3}\), it follows that the number of nodes inside \( F_e \) is less than \(n/3\), which is a contradiction.

     Therefore, in this case $u$ is not ancestor of $v$. Now let $\overline{z}$ be the leaf of $T$ inside $F_e$ with highest position in $\leftpi$. By \Cref{remark:partition}, there exists a node $z_j$ of the boundary such that $\overline{z}\in C(z_j)$. By \Cref{lemma:propaug}, $\omega(F_{u\overline{z}}^\ell) = \omega(F_{uz_j}^\ell)$. The following claim holds.

     \begin{claim}\label{claim:ns1}
         All nodes inside $F_e$ are in $\tilde{F_{uz_j}^\ell}$. Therefore, all nodes inside $F_e$ are counted in $\omega(F_{uz_j}^\ell)$.
     \end{claim}

     \begin{proof}
         Suppose there exists a node $v'$ inside $F_e$ but outside $F_{u\overline{z}}^\ell$. By \Cref{lemma:propaug}, all the nodes $u'$ inside $F_e$ such that $\leftpi(u')<\leftpi(\overline{z})$ or that are descendant of $\overline{z}$ in $T$ are included in $F_{u\overline{z}}^\ell$. Therefore $\leftpi(u') >\leftpi(\overline{z})$. As $u'$ is inside $F_e$, if we take $u''$ has the rightmost children of $u'$ in $T$, this node $u''$ is a leaf in $T$, inside $F_e$ and with higher position in $\leftpi$ than $\overline{z}$, contradicting the election of $\overline{z}$. Finally, by definition all nodes inside $F_e$ and in $F_{uz_j}^\ell$ cannot be in the $T$-path $\overline{P}$ between $u$ and $LCA(u,z_j)$, because in that case, as $z_j$ is inside $F_e$, the nodes in $\overline{P}$ are in the border of $F_e$, a contradiction. Therefore, all nodes inside $F_e$ are in $\tilde{F_{uz_j}^\ell}$. Then, by \Cref{lem:weighting}, all nodes inside $F_e$ are counted in $\omega\left(F_{uz_j}^\ell\right)$.
    \end{proof}

    Therefore, all nodes inside \(F_e\) are counted in \(\omega(F_{uz_j}^\ell)\). Thus, the nodes counted in \(\omega(F_e)\) but not in \(\omega(F_{uz_j}^\ell)\) correspond to the nodes in the \(T\)-path between \(z_j\) and \(v\), except \(z_j\) which is counted in \(\omega(F_{uz_j}^\ell)\). Given that \(\omega(F_{uz_j}^\ell) < n/3\) and \(\omega(F_e) > 2n/3\), the \(T\)-path $P_{z_jv}$ between \(z_j\) and \(v\) has at least \(2n/3\) nodes, and thus, it forms a path separator of \(G\).

At this point, it is unclear whether the edge $vz_j$ can be added to $G$ while maintaining planarity. To avoid complications, in this case, we know that the $T$-path $P_{uv}$ between $u$ and $v$ contains $P_{z_jv}$, and since it is also a path separator. Moreover, since $uv\in E(G)$, in this case, the nodes of $P_{uv}$ form a cycle separator of $G$.
\end{proof}

\begin{lemma}
    \label{lemma:verysmall} Let $G=(V,E)$ be an $n$-node graph, $T$ an arbitrary spanning tree of $G$ such that $\omega\left(F_e\right)< n/3$ for all $T$-real fundamental edges \(e\) . Then, there exists a path $P$ of $T$ that is a cycle separator of $G$.
    
\end{lemma}

\begin{proof}

    Let $e=uv$ be a real fundamental edge such that is not contained in any other real fundamental face. As $G$ is a finite graph and planar, this edge always exists. Assume that $\leftpi(u)< \leftpi(v)$, and as the spanning tree $T$ is fixed, we may omit the subscript in some values to avoid heavy notation.

    For the given embedding \(\mathcal{E}\), if \( u \) is not an ancestor of \( v \), \( F_\ell^e \) is defined as the set of nodes such that (1) are outside \( F_e \) and have a left position in \(\leftpi\) smaller than \(\leftpi(u)\) or  (2) are outside \( F_e \) and in \( V(T_u) \). If \( u \) is an ancestor of \( v \), taking \( z \) as the first node in the \( T \)-path from \( u \) to \( v \) (i.e., \( uz \in E(T) \)), \( F_\ell^e \) is defined as the set of nodes \( v_0 \) outside \( F_e \) that have a left position \(\leftpi\) smaller than \(\leftpi(u)\) or are outside \( F_e \), in \( V(T_u) \), and \( t_u(v_0) > t_u(z) \).

Analogously, \( F_r^e \) is defined as the set of nodes outside \( F_e \) that have a left position \(\leftpi\) greater than \(\leftpi(v)\), regardless of whether \( u \) is an ancestor of \( v \) or not. See \Cref{fig:splitted}.
    
\begin{figure}[h!]
    \centering
    \input{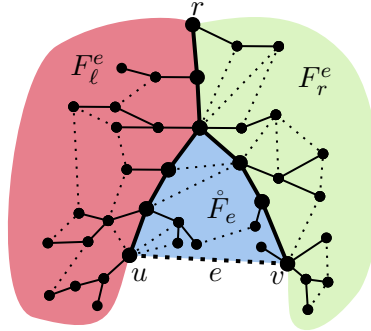}
    \caption{The sets $F_\ell^e$, $F_r^e$ and $\inside{F}_e$, for a fundamental edge $e=uv$. }
    \label{fig:splitted}
\end{figure}

    Given the choice of \( e \), there cannot be any fundamental edge between a node in \( F_\ell^e \) and a node in \( F_r^e \), as such an edge would contain \( e \). Additionally, it is assumed that no real fundamental edge contains the root \( r \) within it.

First, if \( |F_\ell^e|,|F_r^e| \leq n/3 \), then it is clear that that $T$-path $P_{uv}$ between $u$ and $v$ is a cycle separator of $G$.

Now, if $|F_r^e|,|F_\ell^e|\leq 2n/3$ and $|F_r^e|\geq n/3$ (the case when $|F_\ell^e|\geq n/3$ is analogous), then $|\inside{F_e}\cup F_\ell^e|\leq 2n/3$, and therefore the $T$-path $P_{rv}$ between the root $r$ of $T$ and $v$ is a separator set of $G$. Therefore, the $T$-path between $u$ and $v$ is a separator cycle of $G$, because includes a $P_{rv}$ and $uv\in E$.

Finally,  if $|F_\ell^e| > 2n/3$ (the case when $|F_r^e|\geq 2n/3$ is analogous, and both cases cannot happen simultaneously), take $u'$ as the leaf of $T$ with highest position in $\leftpi$ and $u'\in F_\ell^e\cap V(T_u)$. If such a leaf does not exists, take $u'=u$.

    The root $r$ and $u'$ are $\mathcal{E}$-compatible, because $u$ cannot be hidden, and by election of $u'$, if an edge $f$ hides $u'$ then $F_f$ contains $e$, contradicting the election of $e$. Then, we can add $f=ru'$ to $\mathcal{E}$ constructing the combinatorial planar embedding $\mathcal{E}'$ of $G'=(V,E+ru')$, such that $ru'$ is $\mathcal{E'}$-left oriented and $F_{ru}$ is a real fundamental face of $\mathcal{E}'$. Then, $\omega(F_{ru'})\geq 2n/3$, because by election of $u'$, $\omega(F_{ru'}) = |F_\ell^e|$, and $F_{ru'}$ is a real fundamental face of the planar configuration $(G',\mathcal{E}',T)$ and any other real fundamental face contained in $F_{ru'}^\ell$ has weight smaller than $n/3$, because all other real fundamental edges of $(G',\mathcal{E}',T)$ is a real fundamental edge of $T$ in $G$, and by hypothesis all these faces have weight at most $n/3$. Then, we can replicate the prof of \Cref{lemma:balancedincritical} to conclude that there exists a cycle separator $P$, inside $F_{ru'}^\ell$,  of $G'$. As any connected component of $G-S$ is a subgraph of a connected component of $G'-S$, then $S$ is also a cycle separator of $G$.
    
Analogously, if \(|F_r^e| > 2n/3\), then take \(v'\) as the leaf of \(T\) with the lowest position in \(\leftpi\) such that \(v' \in F_r^e \cap V(T_v)\). If such a leaf does not exist, take \(v' = v\). By the same argument as before, \(r\) and \(v'\) are \(\mathcal{E}\)-compatible. Then, we can add \(f = rv'\) to \(\mathcal{E}\), forming the combinatorial planar embedding \(\mathcal{E}'\) of \(G' = (V, E + rv')\), such that \(rv'\) is \(\mathcal{E}'\)-right oriented and \(F_{rv'}\) is a real fundamental face of \(\mathcal{E}'\). Consequently, \(\omega(F_{rv'}) \geq 2n/3\), because by the selection of \(v'\), \(\omega(F_{rv'}) = |F_r^e|\). Since \(F_{rv'}\) is a real fundamental face of the planar configuration \((G', \mathcal{E}', T)\) and any other real fundamental face contained in \(F_{rv'}^\ell\) has a weight smaller than \(n/3\), we can replicate the proof of \Cref{lemma:balancedincritical} to conclude that there exists a \(T\)-path \(P\) (inside \(F_{rv'}\)) that is a cycle separator of \(G'\). Since any connected component of \(G - P_k\) is a subgraph of a connected component of \(G' - P_k\), \(P_k\) is also a cycle separator of \(G\).
\end{proof}

Using \Cref{lemma:balancedincritical,lemma:verysmall} it is direct then to prove \Cref{teo:fundamentalseparator}.

\begin{proof}[Proof of \Cref{teo:fundamentalseparator}]

Let \((G, \mathcal{E}, T)\) be a planar configuration. First, if \(G\) does not have a \(T\)-real fundamental edge, i.e., \(G\) is a tree, then it is known that in any tree, there exists a node \(v\) such that the size (in terms of the number of nodes) of the subtree of \(T\) rooted at \(v\) is within the range \([n/3, 2n/3]\). Therefore, the \(T\)-path between the root of \(T\) and this node \(v\) is a cycle separator of~\(G\)\footnote{As $G$ is a tree, there is not any real fundamental edge and therefore the virtual edge $vr$ can be simulated without breaking planarity.}.

On the other hand, if \(G\) has at least one \(T\)-real fundamental edge, let \(\omega(F_e)\) be the weight of each \(T\)-real fundamental face \(F_e\). If there exists a real fundamental edge \(F_e\) such that \(\omega(F_e) \in [n/3, 2n/3]\), then by \Cref{teo:balancedface}, the \(T\)-path between the nodes of \(e\) are a cycle separator of \(G\). If no real fundamental face has a weight in the range \([n/3, 2n/3]\), we split into two cases.

First, if there exists a real fundamental edge such that \(\omega(F_e) > 2n/3\), then we can choose, among all the real fundamental edges such that \(\omega(F_e) > 2n/3\), the edge \(e\) such that no edge contained in \(e\) has a weight greater than \(2n/3\) (due to planarity and the finiteness of \(G\), such an edge always exists). The planar configuration \((G, \mathcal{E}, T)\) and this edge \(e\) satisfy the hypotheses of \Cref{lemma:balancedincritical}, so in this case, there also exists a cycle separator of \(G\), by \Cref{lemma:balancedincritical}.

Finally, if every fundamental edge has a weight less than \(n/3\), we can choose a fundamental edge \(e\) that is not contained in any other real fundamental face defined by another real fundamental edge. The planar configuration \((G, \mathcal{E}, T)\) and \(e\) satisfy the hypotheses of \Cref{lemma:verysmall}, so there also exists a cycle separator of \(G\).
\end{proof}

\section{Separator cycle computation.}\label{sec:congestimplementation}

To describe in detail the computation of separator cycles in the \CONGEST model, we first outline the necessary subroutines to demonstrate \Cref{teo:separatorcongest}. In \Cref{subsec:shortcuts}, we state known problems and variants specific to this work that can be solved using low-congestion shortcuts in $\biglo{D}$ rounds. Subsequently, in \Cref{subsec:subroutinescongest}, we formalize and demonstrate how to solve various problems in $\biglo{D}$ rounds. These problems are necessary for computing separator cycles.

\subsection{Low-congestion shortcuts and planar embeddings.}\label{subsec:shortcuts}

It is known that a planar combinatorial embedding of a graph can be computed in the \CONGEST\ model.

\begin{proposition}[\cite{ghaffari2016distributed-1}]\label{embeddingcongest}
    There exists a deterministic algorithm in the \CONGEST\ model, which given a planar graph $G=(V,E)$ with diameter $d$, computes a planar combinatorial embedding $\mathcal{E}$ of $G$ in $\biglo{D}$ rounds.
\end{proposition}

Notice that the ordering \(t_v\) given by a planar embedding $\mathcal{E}$ can also be defined based on the set of incident edges at \(v\), so that for an incident edge \(e=vu\) at \(v\), \(t_v(e) = t_v(u)\) according to the above definition. Depending on the context, we will use \(t_v\) interchangeably as an ordering of the neighbors or the incident edges of \(v\).

Given a spanning tree \(T\) and a planar combinatorial embedding \(\mathcal{E}\) of \(G\), we can assume without loss of generality that the edge \(e\) connecting node \(v\) to its parent in \(T\) is assigned the position \(t_v(e) = 1\). We assume this every time we use a planar combinatorial embedding and spanning trees. 

In recent years, the study of low-congestion shortcuts has made significant strides. Various works have explored the construction and application of low-congestion shortcuts in different classes of graphs, including planar graphs, bounded-genus graphs, and minor-free graphs.

To make this work self-contained for readers unfamiliar with the state-of-the-art in low-congestion shortcuts, a brief summary of the definition and known results regarding low-congestion shortcuts is included. For a detailed analysis, the reader is referred to the papers cited in each result.
\begin{definition}[\cite{ghaffari2016distributed-2}, Definition 1]
    Given a graph $G=(V,E)$ with diameter $D$ and a partition $\mathcal{P} = \{V_1,...,V_k\}$ of the nodes, such that for all $V_i\in\mathcal P$ the induced subgraph $G[V_i]$ is connected, a collection of subset of edges $H_1,...,H_k$ is said to be a $(c,d)$-low-congestion shortcut for the partition $\mathcal{P}$ if
    \begin{enumerate}
        \item[(i)] The diameter of $G[V_i]+ H_i$ is at most $c$.
        \item[(ii)] For all edges $e\in E$, $|\{i\in [k]\colon e\in H_i\}|\leq d$
    \end{enumerate}
\end{definition}

In the same work \cite{ghaffari2016distributed-2}, the authors prove that for any partition \(\mathcal{P}\) of the nodes of a planar graph \((V,E)\), there exists a \(\left(\biglo{D},\biglo{D}\right)\)-low congestion shortcut.

Since this seminal work by Ghaffari and Haupler, significant advances have been made in the study and applications of low congestion shortcuts. Below, we summarize the most relevant results that will be used in this work.

\begin{proposition}[\cite{haeupler2018round}, simplified]
\label{teo:deterministicshortcut}
There exists a deterministic distributed algorithm in the \CONGEST\ model such that, for any $n$-node planar graph $G=(V,E)$ with diameter $D$, and a given partition $\mathcal P$ of the nodes into connected subgraphs, the algorithm computes in  $\biglo{D}$ rounds, a  $(\biglo{D},\biglo{D})$-low congestion shortcut of $G$ respect the partition $\mathcal P$.
    
\end{proposition}

\begin{proposition}[\cite{haeupler2018round}]
\label{teo:mstcongest}
There exists a deterministic algorithm in the \CONGEST\ model such any, for any planar graph $G=(V,E)$ with diameter $D$, and a weight function $\omega\colon E\to \R_+$, computes a MST of $G$ (respect to $\omega$) in $\biglo{D}$ rounds.
    
\end{proposition}

A direct application of \Cref{teo:mstcongest} allows to compute spanning trees, concurrently, for a collection of subgraphs.

\begin{lemma}\label{teo:kmstcongest}
    There exists a deterministic algorithm in the \CONGEST\ model which, given planar graph $G=(V,E)$ with diameter $D$, and a partition $\mathcal{P}=\{P_1,...,P_k\}$ of $V$, such that $G[P_i]$ is connected for each $i\in [k]$, computes a collection of trees $T_1,...,T_k\leq G$, such that $T_i$ is a spanning tree of $G[P_i]$ for each $i\in [k]$.
\end{lemma}

\begin{proof}

The algorithm in \Cref{teo:mstcongest} simulates the well-known Boruvka's algorithm, which in $\bigo{\log n}$ iterations, merges subtrees of $G$. In the first iteration, each node $v \in V$ forms its own subtree $(\{v\},\emptyset)$. In each iteration $j$, each subtree selects its Minimum Outgoing Edge (MOE), the edge with the smallest weight connecting a node in the subtree to a neighbor in another subtree, and these components are merged. \Cref{teo:mstcongest} proves that each iteration can be completed in $\biglo{D}$ rounds.

    Then, our algorithm simulates the algorithm of \Cref{teo:mstcongest} using the weight function 
    \begin{equation*}
        \omega\left(e\right) = \begin{cases}
            0 & e=uv\text{ and }u,v\text{ are in same subset }P_i\\
            1 & \text{ otherwise}
        \end{cases}
    \end{equation*}

And stop the simulation of \Cref{teo:mstcongest} if the minimum outgoing edge has weight $1$. Then, after $\biglo{D}$ rounds, the algorithm stops with $T_1, \ldots, T_k$ as different trees, each corresponding to a spanning tree of its respective $G[P_i]$, which we now prove.

Let $T_i$ be a subtree of $G$ that contains a node $v \in P_i$ computed at the end of the algorithm described above. If $T_i$ did not merge with another subtree, it was because in a certain iteration $j$ of the simulation of \Cref{teo:mstcongest}, the minimum outgoing edge of $T_i$ to any node in $V - V(T_i)$ had a weight of $1$, causing it to stop. Hence, $P_i \subseteq V(T_i)$, because if not, due to the connectivity of $G[P_i]$, there would exist two nodes $v_1, v_2 \in P_i$ such that $v_1v_2 \in E$, $v_1 \in V(T_i)$, and $v_2 \notin V(T_i)$. In this case, $v_1v_2$ would be an outgoing edge with weight $0$, which is a contradiction.

Moreover, it must hold that $V(T_i) \subseteq P_i$. Indeed, let $u \in V(T_i) \cap P_i^c$. Let $k \leq j$ be the first iteration such that $u$ and a node $v$ of $P_i$ are in the same subtree. Initially, $u$ is its own subtree, so $k > 1$. In this iteration $k$, the subtree containing $u$ merges with a subtree $T_v$ containing $v$. By the rule of our algorithm, the MOE must have weight $0$, meaning it corresponds to an edge between two nodes of $P_j$. Let $u_2$ be the node in $T_v$ whose edge is the MOE. Repeating the argument now for $u_2$, there must exist an iteration $1 \leq k_2 < k$ such that $u_2$ and $v$ are in the same subtree. Continuing this iteration leads to the conclusion that in iteration $k' = 1$, there is a subtree with a node $u_{k'}$ from $P_j$ and a node $v_{k'}$ from $P_i$, which is a contradiction because, in the first iteration, each node is its own subtree. 
\end{proof}

{\it Partwise Aggregation Problems.} Using \Cref{teo:deterministicshortcut} directly, various problems can be resolved in $\biglo{D}$. However, most of these correspond to variations of broadcast problems, where the objective is to send a message to all other nodes or compute certain local functions that depend on the input from the rest of the nodes.

\begin{definition}[ Part-wise aggregation problem]
    Consider a network graph $G = (V, E)$ and suppose that the vertices are partitioned into disjoint parts $P_1 , P_2 , . . . , P_k$ such that the subgraph induced by each part connected. In the part-wise aggregation (PA) problem, the input is a value $x_v$ for each node $v \in V$ . The output is that each node $u \in P_i$ should learn an aggregate function of the values held by vertices in its part $P_i$, e.g., $\Sigma_{ v\in P_i} x_v$, $\min_{v\in P_i} x_v$, or $\max_{v\in P_i} x_v$. Alternatively, exactly one node in each part has a message and it should be delivered to all nodes of the part.
\end{definition}

The main result explained in \cite{haeupler2018round} states (for planar graphs) the following

\begin{proposition}[\cite{haeupler2018round}]\label{broadcasts}
    There exists a deterministic algorithm in the \CONGEST\ model which, given an specific partwise aggregation problem and a planar graph $G=(V,E)$ with diameter $D$, solves the PA problem over $G$ in $\biglo{D}$ rounds.
\end{proposition}

{\it Broadcasts over a spanning tree $T$}. Many applications and variations of broadcasts or part-wise aggregation problems has been studied. In particular, in \cite{Ghaffari2022}, the authors prove the following theorems

\begin{proposition}[\cite{Ghaffari2022}, rephrased]\label{lemma:hltrees}

There exist two deterministic algorithms $\mathcal{A}_1$ and $\mathcal{A}_2$ in the \CONGEST\ model. Given a planar graph \( G=(V,E) \), a spanning tree \( T \) of \( G \) rooted at \( r \in V \), a function \(\bigoplus\) known by all the nodes, and each node \( v \in V \) have a $\bigo{\log n}$-bit private input \( x_v \), these algorithms solve the following problems:

\begin{itemize}
    \item  $\mathcal{A}_1$ solves the \ancestorsumproblem: Each node \( v \) learns \( p_v \colon = \bigoplus_{\omega \in \text{anc}(v)} x_\omega \), where anc$(v)$ denotes the set of ancestors of \( v \) with respect to \( r \).
    \item $\mathcal{A}_2$ solves the \descendantsumproblem: Each node \( v \) learns \( p_v \colon = \bigoplus_{\omega \in \text{desc}(v)} x_\omega \), where desc$(v)$ denotes the set of descendants of \( v \) with respect to \( r \).
\end{itemize}

\end{proposition}

An important application of part-wise aggregation functions, besides computing the value of the function itself, is enabling nodes to determine the \id\ of a specific node. For example, in $\min$ or $\max$ functions, nodes can identify the \id\ of the node with the input $x_v$ that minimizes/maximizes the function. We summarize some results in the following lemma. To our knowledge, these problems are not explicitly stated in previous works, although some of them are used implicitly and follow a smilar spirit deterministic algorithms presented in \cite{Ghaffari2022}. By completeness, we include the proofs of the results below.
\begin{lemma}\label{lemma:somebroadcasts}
    Given a planar graph $G=(V,E)$ with diameter $D$, a spanning tree $T$ of $G$, a partition $\mathcal{P} = \{P_1,....,P_k\}$ of $V$, $T_1,...T_k$ spanning trees with root $r_i\in P_i$ of each $P_i\in\mathcal{P}$, such that $T_i = T_{r_i}$ for all $i\in [k]$, and a $\log n$-input string $x_v$ for each $v\in V$. There exists deterministic algorithms in the \CONGEST\ model that solves, in parallel, in each $P_i$, the following problems in $\biglo{D}$ rounds.

    \begin{enumerate}[label={\bf \arabic*.}]
        \item \minproblem\ and \maxproblem: all the nodes in $P_i$ acknowledge the \id\ of a node $v$ such that $x_v = arg\min_{v\in P_i} x_v$ and $x_v = arg\max_{v\in P_i} x_v$, respectively.
        \item \sumsubsetproblem: All the nodes in $P_i$ knows the value $n_i = \left|P_i\right|$.
        \item \sumtreeproblem: All nodes $v\in V(T_i)$ acknowledge the number of nodes $|V((T_i)_v)|$, where $(T_i)_v$ is the subtree of $T_i$ rooted in $v$.
        \item \rangeproblem: If the input $x_v$ represents real values, and $R_i =[m_i,M_i]\subseteq \Q$ is a range known by the nodes, all the nodes in $P_i$ acknowledge the \id\ of an arbitrary node $v\in P_i$ such that $x_v\in R$.
        \item \detectancestor\ and \detectdescendant$\colon$ If in each $P_i$ there is a unique node $v_0\in P_i$ known by all nodes in $P_i$, all nodes in $P_i$ can acknowledge if $v_0$ is ancestor or descendant of them in $T_i$, respectively.
    \end{enumerate}
\end{lemma}

\begin{proof}
    \begin{enumerate}
        \item We explain how to solve \minproblem, as \maxproblem\ is analogous. First, using $\min$ as the partwise aggregation function and the spanning tree $T$, using \Cref{lemma:hltrees}, in $\biglo{D}$ rounds the roots $r_1,...r_k$ can learn $p_{r_i} = \min_{\omega\in anc_T(r_i)}x_\omega$, where $anc_T(v)$ is respect to $T$, but as $T_{r_i} = T_i$, then $anc_{T}(r_i) = anc_{T_i}(r_i)$, and then each root can learn the minimum input value of all the nodes in $T_i$.
        Second, in $\biglo{D}$ rounds thanks to $\Cref{broadcasts}$, each root can spreed the value $p_{r_i}$ in its subset $P_i$.

        Finally, in each component $P_i$, each node $v\in P_i$ can define its input value

        \begin{equation*}
            \overline{x}_v = \begin{cases}
                    \id(v) & x_v = p_{r_i} \\
                    0 &\text{ otherwise}
            \end{cases}
        \end{equation*}

        Then, through a \maxproblem, each root $r_i$ can learn the $\id$ of a node $\omega\in P_i$ such that $\overline{x}_\omega = \id(\omega)$. Using $\Cref{broadcasts}$ this \id\ can be learned by all the nodes in $P_i$. All these can be done in $\biglo{D}$ rounds.

        \item Using \Cref{lemma:hltrees} using $\overline{x}_v = 1$ for each $v\in V$ and addition as the partwise aggregation function, in $\biglo{D}$ rounds each node $v\in V$ can learn $p_v = |T_v|$, but as $T_{r_i} = T_i$, then $T_v = (T_i)_{v}$.
        
        \item Direct application of $\Cref{broadcasts}$ using $\overline{x}_v = 1$ for each node $v\in V$ and addition as the partwise aggregation function.
        \item In each subset $P_i$, each node $v\in P_i$ defines

        \begin{equation*}
            \overline{x}_v = \begin{cases}
                \id(v) & x_v\in R_i\\
                0 &\text{ otherwise}
            \end{cases} 
        \end{equation*}
    \end{enumerate}
        Then, similar to the proof of $1.$, through \maxproblem, each root $r_i$ can learn the \id\ of a node $\omega$ such that $\overline{x}_\omega = \id(\omega)$ and $x_\omega\in R_i$, and through a partwise aggregation problem each root spread this \id\ to the rest of the nodes in $P_i$. All these can be done in $\biglo{D}$ rounds thanks to \Cref{broadcasts} and \Cref{lemma:hltrees}.

        \item Each node can define $\overline{x}_v = 0$ if $v\neq v_0$ and $\overline{x}_v = 1$ if $v= v_0$. Using $\bigoplus$ as addition $+$, using \Cref{lemma:hltrees} all nodes v in $P_i$ can learn $\Sigma_{\omega\in desc(v)}x_v$ and $\Sigma_{\omega\in anc(v)}x_v$, if these additions are equal to one, then $v$ is ancestor or descendant or $v_0$, respectively.

        \item Each node defines $\overline{x}_v = 0$ if $v\neq v_0$ and $\overline{x}_{v_0} = 1$. Then using the addition function and \Cref{lemma:hltrees} each node can know if its ancestor of descendant of $v_0$, because in such cases $\Sigma_{\omega \in \text{desc}(v)} x_\omega = 1$ or $\Sigma_{\omega \in \text{anc}(v)} \overline{x}_\omega = 1$, respectively. 
\end{proof}

\subsection{Deterministic subroutines.}\label{subsec:subroutinescongest}

In this Section we describe several deterministic algorithms that will be used later to obtain our deterministic cycle separator and DFS algorithms. Even some of these problems can by implied by prior work (for which we include the reference), we include a proof for completeness and self-content of our work.

\subsubsection{DFS-orders and computing the weights.}

To describe the deterministic algorithm for computing cycle separators, we first demonstrate that some of the subproblems defined in \Cref{sec:highlevelideas} can be resolved in $\biglo{D}$ rounds.

First, the \dfsorderproblem\ problem is defined as follows.

\begin{algorithmbox}{\dfsorderproblem}
		
		\textbf{Input}: A planar configuration $(G,\mathcal{E},T)$, a partition of the nodes $\mathcal{P}=\{P_1,...,P_k\}$ and a spanning tree $T_i$ of each induced subgraph $G[P_i]$. \\
		
		\textbf{Output:} For each $P_i\in\mathcal{P}$, each node $v\in P_i$ acknowledge $\leftpi(v)$ and $\rightpi(v)$ with respect to $T_i$ and $\mathcal{E}$\footnote{Each induced subgraph $G[P_i]$ uses the induced combinatorial planar embedding given by $\mathcal{E}$ restricted to $G[P_i]$}.

	\end{algorithmbox}

In the following lemma, we show that the \dfsorderproblem\ can be solved in $\biglo{D}$ rounds. 

 \begin{lemma}
 	
 \label{lemma:dfsordercongest}
     
     There exists a deterministic algorithm in the \CONGEST\ model, which solves \dfsorderproblem\ in $\biglo{D}$ rounds.
 \end{lemma}

    \begin{proof}

        We explain how the algorithm computes $\leftpi$. The algorithm for computing $\rightpi$ is symmetric.

First, using \Cref{lemma:somebroadcasts}, within $\biglo{D}$ rounds, each node $v$ in each induced subgraph $G[P_i]$ can determine the number of nodes $n_{T_i}(v)$ in its subtree $(T_i)_v$. The algorithm then proceeds through $\bigo{\log n}$ phases, each consisting of $\biglo{D}$ rounds. We will describe how the algorithm operates within a specific subgraph $G[P_i]$ and subsequently demonstrate how this can be executed in parallel across all induced subgraphs by a set $\mathcal{P}$.

Recall that in a distributed representation of a spanning tree, we assume each node knows its depth and the \id\ of its parent. In each phase $j$, we have a collection of disjoint subtrees $T_j(v_1), \ldots, T_j(v_k)$ of $T_i$. Each subtree $T_j(v_h)$ is rooted at $v_h \in P_i$ and has a depth $d(T_j(v_h))$. The nodes in each subtree $T_j(v_h)$ are aware of the (correct) \leftorder\ with respect to $(T_i)_{v_h}$, the subtree of $T_i$ rooted at $v_h$ that contains all the descendants of $v_h$ in $T_i$.\footnote{The subtrees $T_j(v_h)$ may not be the entire subtree $(T_i)_{v_h}$ of $T$ with root $v_h$, but rather a subtree of it.
}

In the first phase, each node starts as its own subtree, i.e., $T_1(v) = ({v}, \emptyset)$ for all $v \in P_i$. The depth $d(T_1(v))$ in the first phase is initialized as the depth of $v$ in $T_i$. The \id\ of $T_1(v)$ is the \id\ of $v$, and each node sets $\leftpi^1(v) = 1$.
In phase $j$, each subtree $T_j(v_h)$ at an odd depth of $2s+1$ joins the subtree $T_j(v_t)$ at an even depth of $2s$, such that the parent $p_i(v_h)$ of $v_h$ in $T_i$ is in $T_j(v_t)$. Let $T_j(v_{h_1}), \ldots, T_j(v_{h_l})$ be all the subtrees at depth $2s+1$ that will join $T_j(v_t)$ in phase $j$. Joining $T_j(v_h)$ to $T_j(v_t)$ means that $v_h$ updates its parent as its parent in $T_i$, and each node in $T_j(v_h)$ needs to update its \leftorder\ and its depth, now according to the subtree $T_{j+1}(v_t) = (V(T_j(v_t)) \cup V(T_j(v_{h_1})) \cup \ldots \cup V(T_j(v_{h_l})), E(T_j(v_t)) \cup E(T_j(v_{h_1})) \cup \ldots \cup E(T_j(v_{h_l})) \cup \{(v_h, p_i(v_h))\})$.

To update the left order position of all nodes in $T_j(v_{h_1}), \ldots, T_j(v_{h_l})$, recall that each node $z$ in $T_j(v_t)$ has the correct position in $\leftpi$ with respect to $(T_i)_{v_t}$. This position considers all nodes in $(T_i)_{v_t}$ even if they are not currently in $T_j(v_t)$. In phase $j$, if $z$ is a leaf in $T_j(v_t)$, then all $T_i$-children $z'$ of $z$, which are roots of a subtree $T_j(z')$, join $(T_i)_{v_t}$. Consequently, $z$ can determine in one round the number of nodes in the subtree of each child in $T_i$. After this, $z$ transmits in one round to each $T_i$ child its correct position in $\leftpi$ in $T_{j+1}(v_t)$ according to $\mathcal{E}$ and the number of nodes that each branch of its subtree has in $T_i$. Specifically, if $z$ has $q$ $T_i$-children that will join it in phase $j$, let $v_1, \ldots, v_q$ be the counterclockwise order of the $q$ nodes in $t_z$ of $\mathcal{E}$. Then, for each node $v_x$ with $x \in [q]$, the left order position $\leftpi(v_x) = \leftpi(z) + 1 + \sum_{y<x} n_{T_i}(v_y)$ is assigned, and the range $[\leftpi(z) + 1 + \sum_{y<x} n_{T_i}(v_y), \leftpi(z) + 1 + \sum_{y<x} n_{T_i}(v_y) + n_{T_i}(v_x)]$ is designated for all the left order positions of the nodes in $(T_i)_{v_x}$. Additionally, each node $v_x$ updates its depth to be the depth of $z$ in $T_j(v_t)$ plus one.

Once each root of all subtrees $T_j(v_{h_1}), \ldots, T_j(v_{h_l})$ that will be joined to $T_j(v_t)$ knows their new depth and left order position in $T_{j+1}(v_t)$, according to \Cref{broadcasts}, within $\biglo{D}$ rounds, they can transmit their new depth $d_{j+1}(v_{h_x})$ and left order position $\leftpi(v_{h_x})$ to the rest of the nodes in their respective subtree $T_j(v_{h_x})$. With this new information, each node $v$ in $T_j(v_{h_x})$ updates its depth and left order position in $T_{j+1}(v_t)$ as $d_j(v) + d_{j+1}(v_x)$ and $\leftpi^{j+1}(v) = \leftpi^j(v) + \leftpi^{j+1}(v_x)$, respectively. All nodes in $T_j(v_t)$ retain the same information they already have (the same parent, depth, and left order position). Also, using \Cref{broadcasts}, each node of a subtree $T_j(v_i)$ joined to a subtree $T_j(v_t)$ acknowledges \id($v_t$) as the \id\ of the subtree, and all the nodes in the new subtree $T_{j+1}(v_t)$ update the depth of the new subtree to $\lfloor \text{depth}_{T_i}(v_j) / 2^j \rfloor$.

Each phase thus requires $\biglo{D}$ rounds to be completed, and the maximum depth of a subtree in each phase $j$ decreases by half. Consequently, after $\bigo{\log n}$ rounds, in the final round $j$ of the algorithm, each subgraph $G[P_i]$ contains a unique subtree $T_j(r_{T_i})$ such that $T_j(r_{T_i}) = (T_i)_{r} = T_i$. Therefore, after $\biglo{D}$ rounds, each subgraph $G[P_i]$ has computed the left order of $T_i$.

    \end{proof}

Using \Cref{lemma:dfsordercongest}, we can compute the weights of any real fundamental edge in $\biglo{D}$ rounds.

\begin{algorithmbox}{\weightproblem}

\textbf{Input}: A planar configuration $(G,\mathcal{E},T)$. A partition $\mathcal{P}=\{P_1,...,P_k\}$ of $G$ and a spanning tree $T_i$ of each $G[P_i]$.\\

\textbf{Output:} In each induced subgraph $G[P_i]$, the endpoints of every $T_i$-real fundamental edge $e$ know the value of $\omega(F_e)$ according to \Cref{def:weights} of its unique real fundamental face $F_e$. 

\end{algorithmbox}

\begin{lemma}\label{teo:weightscongest}
    There exists a deterministic algorithm in the \CONGEST\ model, which solves \weightproblem\ in $\biglo{D}$ rounds.
\end{lemma}

\begin{proof}
    First, in $\biglo{D}$ rounds the \leftorder\ and \rightorder\ of each $T_i$ can be computed in $\biglo{D}$.
    Let $P_i\in\mathcal{P}$ and $e=uv$ be a $T_i$-real fundamental edge. In order to compute the weights according to \Cref{def:weights}, $u$ and $v$ need to first determine if one is an ancestor of the other. At the end of the algorithm to compute \leftorder\ and \rightorder, each node $z$ not only know its position in $\leftpi$ and $\rightpi$, respectively, but also know the range $R_z\subseteq [n]$ of left order positions that include all nodes in the subtree of $T_i$ rooted in $z$. Thus, in one round of interaction, $u$ and $v$ exchange their left and right order positions, and hence $u$ (similarly, $v$) can verify if $\leftpi(v)\in R_u$ (similarly, $\leftpi(u)\in R_v$). Moreover, once they determine if one node is an ancestor of the other, given the embedding $\mathcal{E}$, the ancestor node can verify whether condition 2.1 or 2.2 of \Cref{def:weights} is satisfied. Thus, after $\bigo{1}$ rounds, all endpoints of each real fundamental edge can determine which weight function they should compute according to \Cref{def:weights}.

Next, to compute the respective weight function, assuming $\leftpi(u)<\leftpi(v)$, they must calculate and share some of the following values: $p_{F_e}(u), p_{F_e}(v), \leftpi(v), \leftpi(u), \rightpi(v), \rightpi(u), d_{T_i}(u), d_{T_i}(v)$, and $n_{T_i}(u)$. Given that $\leftpi$ and $\rightpi$ are already computed, and in the distributed representation of a spanning tree, each node knows its depth, the only values that still need to be computed are $p_{F_e}(u), p_{F_e}(v)$, and $n_{T_i}(u)$. On the one hand, $p_{F_e}(u)$ and $p_{F_e}(v)$ are easily computed by $u$ and $v$ locally, without communication, since they know the order of their edges according to $t_u$ and $t_v$ of $\mathcal{E}$. On the other hand, the value of $n_{T_i}(u)$, the number of nodes in the subtree of $T_i$ rooted at $u$, is already computed, as nodes calculate it (using \Cref{lemma:hltrees}) during the computation of the left and right orders.

Thus, after computing the left and right orders in $\biglo{D}$ rounds, the endpoints of each $T_i$-real fundamental edge of each spanning $T_i$ for $i\in [k]$ can compute the weight that corresponds to them according to \Cref{def:weights}.

\end{proof}

\subsubsection{Marking a path.} Given a planar configuration $(G,\mathcal{E},T)$, a recurrent problem in our deterministic algorithms will be to mark a path of $T$ between two nodes $u,v$\footnote{Even not strictly needed, we assume $u,v$ are known by all the nodes, since they can transmit its \id\ through PA problems.}. Specifically, at the end of the algorithm, each node should be aware if they are part of the path being marked. Since the path has a maximum length of $n$, this can be done trivially in $\bigo{n}$ rounds, but we will show how it can be resolved in $\biglo{D}$ rounds. This problem relates to deterministic subset-sum tree procedures described in \cite{Ghaffari2022}.

The formal problem is the following.

\begin{algorithmbox}{\markpathproblem}
      \textbf{Input}: A planar configuration $(G,\mathcal{E},T)$ with diameter $D$, a partition of $\mathcal{P}$ of $V$, a spanning tree $T_i$ of each $G[P_i]$, and two nodes $v_i,v_i\in P_i$ in each $P_i\in \mathcal{P}$ known by all the nodes of $P_i$. \\

       \textbf{Output:} All the nodes in each $P_i$ acknowledge if they are in $T$-path between $u_i$ and $v_i$.
\end{algorithmbox}

\begin{lemma}\label{teo:markpathcongest}
    There exists a deterministic algorithm in the \CONGEST\ model, which solves \markpathproblem\ in $\biglo{D}$ rounds.
\end{lemma}

\begin{proof}
    The algorithm proceeds through $\bigo{\log n}$ phases, and each phase consists of $\bigo{\log n}$ iterations. The algorithm is explained within a specific induced subtree $G[P_i]$ and then it is argued how this can be done in parallel in all connected components.

    In the first phase, the goal is to mark the edge $e_1$ in the 'middle' of the $T_i$-path between $u_i$ and $v_i$. To achieve this, in each iteration $j$ of the first phase, a partition $T_1^j, \ldots, T_k^j$ of $T$ into $k$ sub-trees is maintained. The nodes in each sub-tree $T_h^j$ know the \id\ of the sub-tree and its depth. The \id\ of each sub-tree $T_i$ is either $\id(u)$ or $\id(v)$, if $u$ or $v$ are in $T_i$, and if not, the \id\ of $T_i$ is the \id\ of the root of $T_i$. If the root of $T_i$ has depth $h$ in $T$, then the depth of $T_i$ is $\lfloor h/2^{j-1}\rfloor$. In each iteration $j$, the sub-trees $T_{i_1}^j$ with even depth $2d$ are merged into the sub-tree $T_{i_2}^j$ that has depth $2d-1$ and where there exists an edge in $T$ between a node of $T_{i_1}^j$ and a node in $T_{i_2}^j$. By merging sub-trees, it is meant that if $T_{i_1}^j$ has depth $2d$ and will merge with a sub-tree $T_{i_2}^j$ of depth $2d-1$, then all nodes in $T_{i_1}^j$ will update the \id\ of their sub-tree to \id$(T_{i_2}^j)$, and all nodes in $T_j$ plus the nodes in sub-trees merged to $T_j$ in round $i$ will update the depth of their new merged sub-tree as $\left\lfloor \text{depth}(T_j)/2 \right\rfloor$.

    The first phase continues until two sub-trees $T_{u'}$ and $T_{u'}$ merge, such that $u_i \in T_{u'}$ and $v_i \in T_{v'}$. At this point, the edge $e_1 = z_1z_2 \in T$ is marked, where $z_1 \in T_{u'}$ and $z_2 \in T_{v'}$ and $e_1$ is the edge that merge $T_{u'}$ with ${v'}$. By construction, this edge $e_1$ is in the $T_i$-path between $u_i$ and $v_i$, and the size of the $T_i$-paths between $u_i$ to $z_1$ and $z_2$ to $v_i$ are half of the size of the $T_i$-path between $u_i$ and $v_i$. In the second phase, the search for the edge in the middle of the $T_i$-path between $u_i$ and $z_1$ and the edge in the middle of the $T$-path between $z_2$ and $v_i$ continues in parallel. The rest of the phases operate analogously.

    In each iteration of any phase, the maximum depth of each sub-tree decreases by half, and initially, the maximum depth is $n$, so after $\log n$ iterations, each phase ends. In phase $t$, the algorithm in parallel searches for $2^{t-1}$ edges to mark, and since the $T_i$-path between $u_i$ and $v_i$ has a length of at most $n$, after $\bigo{\log n}$ phases, the algorithm has marked each edge in each $T_i$-path between $u_i$ and $v_i$.    Finally, in each iteration $j$ of any phase, every sub-tree $T_i^j$ at odd depth merges with the sub-tree $T_h^j$ such that the parent of the root of $T_i^j$ is in $T_h^j$. Then, the root of $T_i^j$ will know in $1$ round the \id\ and depth of $T_h^j$. To complete the merge between $T_i^j$ and $T_h^j$, the root needs to communicate this information to all other nodes in $T_i^j$. This is a classic broadcast problem and can be solved in $\biglo{D}$ rounds using \Cref{broadcasts}.

Then, each iteration can be computed in $\biglo{D}$ rounds, thereby achieving the claimed round complexity of $\biglo{D}$ rounds in total.
\end{proof}

\subsubsection{Detecting LCA between two nodes.} To generalize the algorithm from \Cref{teo:markpathcongest}, it is first necessary to introduce the detection of the lowest common ancestor between two nodes.

\begin{algorithmbox}{\lcaproblem}
      \textbf{Input}: A planar configuration $(G,\mathcal{E},T)$ with diameter $D$, a partition of $\mathcal{P}$ of $V$, a spanning tree $T_i$ of each $G[P_i]$, and two nodes $u_i,v_i\in P_i$ in each $P_i\in \mathcal{P}$ known by all the nodes of $P_i$. \\

       \textbf{Output:} In each $P_i$ the lowest common ancestor node between $u_i$ and $v_i$ in $T$ is marked.
\end{algorithmbox}

\begin{lemma}\label{teo:lcacongest}
    There exists a deterministic algorithm in the \CONGEST\, which solves, in $\biglo{D}$ rounds, the \lcaproblem.
\end{lemma}

\begin{proof}
    First, in each induced subgraph $G[P_i]$, the left and right DFS order of each $T_i$ is computed in $\biglo{D}$. Once the orders have been computed, all nodes $v\in P_i$ also know, for each $T_i$-child $z$, the set of values $R_z^\ell, R_z^r\subseteq [n]$ that the nodes of each subtree $T_z$ take in $\leftpi$ and $\rightpi$, respectively. Therefore, since in $\biglo{D}$ rounds the nodes $u_i$ and $v_i$ can share their left order positions $\leftpi(u_i), \leftpi(v_i)$ in $\leftpi$, each node $v\in P_i$ can determine if it is on the $T_i$-path between $u_i$ and the root of $T_i$, and also if it is on the $T_i$-path between $v_i$ and the root of $T_i$. Given this, each node $v\in P_i$ defines its input $x_v$ as

    \begin{equation*}
        x_v = \begin{cases}
            d_{T_i}(v) + 1 &v \text{ is in both }T_i\text{-paths}\\
            0 &\text{~}
        \end{cases}
    \end{equation*}

    The LCA between $u_i$ and $v_i$ in $T_i$ satisfy that is the deepest node $v\in P_i$ such that $x_v >0$, then, through a \maxproblem, this node can be detected in all subgraphs $G[P_i]$ in $\biglo{D}$ rounds using \Cref{lemma:somebroadcasts}. Then, in each subgraph $G[P_i]$ the LCA between $u_i$ and $v_i$ has been detected using constant number of subroutines, each one with a a round complexity of $\biglo{D}$ rounds.
\end{proof}

\subsubsection{Subroutines in a real fundamental face.}
In the algorithm designed for computing cycle separators, it is necessary at a certain point to solve, deterministically, specific problems inside a specific face $F_e$. The first of these problems consists in each node to detect whether it is part of a real fundamental face or not.
\begin{algorithmbox}{\detectfaceproblem}
     \textbf{Input}: A planar configuration $(G,\mathcal{E},T)$ with diameter $D$, a partition of $\mathcal{P}$ of $V$, a spanning tree $T_i$ of each $G[P_i]$ and the \id of the two nodes $u_i,v_i$ of a $T_i$ real fundamental edge $e=u_iv_i$ known by all the nodes of $P_i$. \\

       \textbf{Output:} All the nodes in each $P_i$ acknowledge if they are in $F_e$ or not.
\end{algorithmbox}

\begin{lemma}\label{teo:detectfacecongest}
    There exists a deterministic algorithm in the \CONGEST\ model, which solves \detectfaceproblem\ in $\biglo{D}$ rounds.
\end{lemma}

\begin{proof}
    Let $P_i \in \mathcal{P}$ be a subset of nodes, $T_i$ its spanning tree, and $u_i, v_i$ the endpoints of the real fundamental edge $e = u_i v_i$ given in the input. The nodes in $F_e$ can be either nodes on the border of the $T_i$-path between $u_i$ and $v_i$, or nodes that are inside $F_e$, in $\inside{F}_e$.

First, in $\biglo{D}$ rounds, the nodes can compute the \leftorder\ and \rightorder\ in $T_i$ (\Cref{lemma:dfsordercongest}). Additionally, in $\biglo{D}$ rounds, all the nodes on the $T_i$-path between $u_i$ and $v_i$ can be marked using \Cref{teo:markpathcongest}. To mark the nodes inside $F_e$, these nodes can either be descendants of $u_i$ or $v_i$ (if they are nodes in the subtrees $T_{u_i}$ or $T_{v_i}$ inside $F_e$), or they may not be descendants of $F_e$. For the nodes that are descendants of $u_i$ or $v_i$, since both nodes know their combinatorial planar embedding $t_{u_i}$ and $t_{v_i}$, respectively, they can determine which of their $T_i$-children are inside $F_e$. Moreover, by computing the left order $\leftpi^{T_i}$, nodes $u_i$ and $v_i$ can determine the range of values $R_\omega \subseteq [n]$ taken by the nodes $T_\omega$ in $\leftpi^{T_i}$ for all their $T_i$-children $\omega$. Thus, $u_i$ and $v_i$ can calculate the range of values $I(u_i), I(v_i) \subseteq [n]$ taken by all their descendants inside $F_e$. In $\biglo{D}$ rounds (two calls to \Cref{broadcasts}, one for $u_i$ and one for $v_i$), $u_i$ and $v_i$ can inform all other nodes in $P_i$ of these intervals $I(u_i), I(v_i)$, so that each node $z$ such that $\leftpi^{T_i}(z) \in I(u_i)$ or $\leftpi^{T_i}(z) \in I(v_i)$ is marked as being inside $F_e$, as they are descendants of $u_i$ or $v_i$ inside $F_e$.

Finally, for the remaining nodes inside $F_e$, by \Cref{prop:caracinside}, if all nodes in $P_i$ know whether $u_i$ is an ancestor of $v_i$, or if $u_i v_i$ is $\mathcal{E}$-left-oriented or $\mathcal{E}$-right-oriented, each node also know the left and right positions of $u_i$ and $v_i$, they can locally decide whether they are inside $F_e$ according to each case in \Cref{prop:caracinside}. All this information from $u_i$ and $v_i$ can be transmitted in $\biglo{D}$ rounds using \Cref{broadcasts} and \Cref{lemma:somebroadcasts}.
\end{proof}

The last sub-problem involves detecting whether a particular node is covered or hidden in a real fundamental face.

\begin{algorithmbox}{\hiddenproblem}
     \textbf{Input}: A planar configuration $(G,\mathcal{E},T)$ with diameter $D$, a partition of $\mathcal{P}$ of $V$, a spanning tree $T_i$ of each $G[P_i]$, the \id of two nodes $u_i,v_i$ of a $T_i$ real fundamental edge $e=u_iv_i$ known by all the nodes of $P_i$, and the \id\ of a $T_i$-leaf $z_i\in \inside{F}_e$. \\

       \textbf{Output:} All the nodes in each $P_i$ acknowledge all the adjacent real fundamental edges (if any exists) that hiddens $z_i$ in $F_{e_i}$.
\end{algorithmbox}

\begin{lemma}
    There exists a deterministic algorithm in the \CONGEST\ model, which solves \hiddenproblem\ in $\biglo{D}$ rounds.
\end{lemma}

\begin{proof}
    First, by using \Cref{teo:detectfacecongest}, in all subgraphs $G[P_i]$, the nodes can determine if they are in $F_{e_i}$ or not. All nodes that are not in $F_{e_i}$ know that they cannot have an edge adjacent to them that hides $z_i$, as by definition, only nodes in $F_{e_i}$ can do so.

In $\biglo{D}$ rounds, the nodes can compute $\leftpi^{T_i}$, and the node $z_i$ can share its position in $\leftpi^{T_i}$ with all other nodes in $P_i$. Additionally, by computing \detectancestor\ and \detectdescendant\ for the node $z_i$, all nodes can also determine if they are ancestors or descendants of $z_i$. With this information, each node can decide if the real fundamental edge incident to them hides $z_i$ by communicating with their neighboring nodes in $G[P_i]$ in one round.

Effectively, node $u$ can verify if a fundamental edge $f = uv'$ incident to it hides $z_i$, as it can locally check if $V(T_u) \cap V(\inside{F}e) \subseteq V(F_f)$. Node $u$ knows its combinatorial embedding $t_u$ and can verify if $z_i$ is inside $F_f$ by checking if $v'$ is not an ancestor of $z_i$ and $\leftpi(z_i) < \leftpi(v')$, or if $v'$ is an ancestor of $z_i$ and $z_i \in V(T{v'}) \cap V(F_f)$. Node $v'$ can verify this as it knows $t_{v'}$, can identify which of its children in $T_i$ are inside $F_f$, and knows the interval of positions $I_\omega \subseteq [n]$ that all its descendants in a subtree $T_\omega$ (with $\omega$ a $T_i$-child of $v'$) occupy. Therefore, $v'$ can verify if $z_i$ has a position in $\leftpi^{T_i}$ that lies within some $I_\omega$ with $\omega$ inside $F_f$. Nodes $u$ and $v'$ share this information to decide if $z_i$ is hidden by $f$ or not.

If $u' \neq u$, let $f = u'v'$ be a fundamental edge incident to $u'$ such that $\leftpi^{T_i}(u') < \leftpi^{T_i}(v')$. Node $u'$, similarly to the previous case, can verify if $\leftpi(u') < \leftpi(z_i)$, and $v'$ can verify if $z_i$ is inside $F_f$ in an analogous manner. Nodes $u'$ and $v'$ share this information to decide if $z_i$ is hidden by $f$ or not.\end{proof}

Along with \hiddenproblem, at some point in the algorithm of \Cref{teo:separatorcongest}, it is necessary to select a real fundamental edge that is not contained within any other. Formally, the problem is defined as follows.

\begin{algorithmbox}{\notcontainedproblem}
     \textbf{Input}: A planar configuration $(G,\mathcal{E},T)$ with diameter $D$, a partition of $\mathcal{P}$ of $V$, a spanning tree $T_i$ of each $G[P_i]$, and a collection of real fundamental edges $E_i = \{e_1,...,e_k\}$. This means that both extremes of each $e_j$ know that the edge $e_j$ is in $E_i$. \\

       \textbf{Output:} All the nodes in each $P_i$ acknowledge the \id\ of both extremes of a real fundamental edge $e\in E_i$ that is not contained in any other real fundamental edge of $E_i$.
\end{algorithmbox}

\begin{lemma}\label{teo:notcontainedcongest}
    There exists a deterministic algorithm in the \CONGEST\ model, which solves \notcontainedproblem\ in $\biglo{D}$ rounds.
\end{lemma}

\begin{proof}
    First, in $\biglo{D}$ rounds, the \leftorder\ of each $T_i$ can be computed, and all nodes in each $P_i$ acknowledge the size of $|P_i|$ (\Cref{lemma:somebroadcasts}).

Then, in each $P_i$, all nodes $v \in P_i$ define $x_v$ such that

\begin{equation*}
x_v = \begin{cases}
\leftpi^{T_i}(v) & \exists e \in E_i \text{ such that } v \in e \\
n & \text{ otherwise}
\end{cases}
\end{equation*}

Through a \minproblem, the nodes identify the \id\ of the node $v'$ such that $x_{v'} = \min_{v \in P_i} {x_v}$. This node $v'$ then transmits to all other nodes in $P_i$ the \id\ of itself and the other node $u$ such that $f = v'u \in E_i$. If more than one edge $f$ exists, node $v'$ selects $u$ to be the one with the highest position in $\leftpi^{T_i}$.

At this point, the nodes know the \id\ of an edge $f = v'u$ that can only be contained by another real fundamental edge $f = v'u'$, where $v'$ is an ancestor of $v$ or $u'$ is an ancestor of $u$. Otherwise, $f$ would have a position in $\leftpi^{T_i}$ less than $v$ or a position in $\leftpi^{T_i}$ greater than $u$, contradicting the selection of $v$ or $u$, respectively. By applying \detectancestor\ twice (once for $u$ and once for $v$), all nodes in $P_i$ can detect if they are ancestors of $u$ and/or $v$. With this information, all nodes that are ancestors of $u$ or $v$ can detect if they are incident to a fundamental edge in $E_i$ that contains $f$.

The nodes $v'$ that are ancestors of $v$ verify if they are incident to an edge $f' = v'u' \in E_i$ such that $\leftpi^{T_i}(u') > \leftpi^{T_i}(u)$, or $u'$ is an ancestor of $u$ and $t_{v'}(f') > t_{v'}(e)$, with $e$ being the $T$-edge between $v'$ and its $T$-child $\omega$ such that $v' \in T_\omega$. If there exists an edge $f' = v'u'$ that meets the above conditions, through procedures analogous to the selection of $f$ (using \Cref{broadcasts} and \Cref{lemma:somebroadcasts}), in $\biglo{D}$ rounds, the nodes in $P_i$ can know the \id\ of the endpoints of an edge $f' \in E_i$, which by construction cannot be contained in any other edge of $E_i$.
\end{proof}

An analogous problem consists in detect an edge that does not contain any other edge between an specific subset of edges.

\begin{algorithmbox}{\notcontainsproblem}
     \textbf{Input}: A planar configuration $(G,\mathcal{E},T)$ with diameter $D$, a partition of $\mathcal{P}$ of $V$, a spanning tree $T_i$ of each $G[P_i]$, and a collection of real fundamental edges $E_i = \{e_1,...,e_k\}$. This means that both extremes of each $e_j$ knows the edge $e_j$ is in $E_i$. \\

       \textbf{Output:} All the nodes in each $P_i$ acknowledge the \id\ of both extremes of a real fundamental edge $e\in E_i$ that does not contain any other real fundamental edge of $E_i$.
\end{algorithmbox}

\begin{lemma}
    There exists a deterministic algorithm in the \CONGEST\ model, which solves \notcontainsproblem\ in $\biglo{D}$ rounds.
\end{lemma}

\begin{proof}
    Analogous to \Cref{teo:notcontainedcongest}.
\end{proof}

\subsection{Separator cycle computation: proof of \Cref{teo:separatorcongest}.}\label{subsec:separatoralgcongest}

Recall that formally, the problem to be solved is the following

\begin{algorithmbox}{\separatorproblem}
	
	\textbf{Input}: A graph $G = (V,E)$, a planar combinatorial embedding $\mathcal{E}$ of $G$, a partition $\mathcal{P}=\{P_1,...,P_k\}$ of $V$.
 \medskip
	
	\textbf{Output:} A set of marked nodes $ S_i\subseteq P_i$ in each subset of nodes, such that $S_i$ is a cycle separator of $G[P_i]$.

\end{algorithmbox}

 The algorithm that solves \separatorproblem\ proceeds through $6$ main phases which we call {\tt Precomputation Phase, Phase 1, Phase 2, Phase 3, Phase 4} and {\tt Phase 5}. The separator set of each subgraph $G[P_i]$ is computed at then end of one of the last four phases. We explain how the algorithm proceeds in each subgraph $G[T_i]$ in parallel.\\

\noindent  {\bf Phase 1.} ({\tt Precomputation Phase.}) First, a planar combinatorial embedding $\mathcal{E}$ is computed in $\biglo{D}$ rounds. Then, in each subgraph $G[P_i]$ a spanning tree $T_i$ is computed (\Cref{teo:mstcongest}), along with the \leftorder\ and \rightorder\ of each $T_i$ (\Cref{lemma:dfsordercongest}), and the weight of each real fundamental face $F_e$ of $T_i$ (\Cref{teo:weightscongest}). Additionally, each node $v \in P_i$ computes the number of nodes $n_{T_i}(v)$ in each subtree $(T_i)_v$ and $n_i := |P_i|$ (\Cref{lemma:somebroadcasts}). All these computations can be done in $\biglo{D}$ rounds. \\

\noindent  {\bf Phase 2.} In this phase, each component $P_i$ aims to verify if $G[P_i]$ is a tree. If $P_i$ is a tree, a centroid node is selected and acknowledged by all nodes in each subgraph.

To achieve the desire verification, each subgraph $G[P_i]$, each node $v \in P_i$ defines $x_v = 1$ if a $T_i$-real fundamental edge is incident on $v$, and $0$ otherwise. Then, through a partwise-aggregation problem using $\bigoplus = \max$, all nodes in $P_i$ can acknowledge the value of $x(T_i) := \max_{v \in P_i} x_v$. Using \Cref{broadcasts} and \Cref{lemma:somebroadcasts}, this is done in $\biglo{D}$ rounds.

If $x(T_i) = 1$ (i.e., $G[P_i]$ is not a tree), the nodes in $P_i$ wait until the next phase starts. If $x(T_i) = 0$, then the nodes acknowledge that $G[P_i]$ is a tree, i.e., $G[P_i] = T_i$, and they seek a centroid node of $T_i$. To do so, all nodes define $x_v = n_{T_i}(v) := \left| V(T_i)v \right|$, and through a \rangeproblem\ with $R_i = [n_i/3, 2n_i/3]$, all nodes learn the \id\ of a node $v_0 \in P_i$ such that $n{T_i}(v_0) \in R_i$ (it is known that in every tree, such a node $v_0$ exists). This is also done in $\biglo{D}$ rounds (\Cref{lemma:somebroadcasts}).

Then, the nodes in $G[P_i]$ proceed to a \markpathproblem\ to mark all the nodes in the $T_i$-path between the root of $T_i$ and $v_0$. This can be done in $\biglo{D}$ rounds, and by the selection of $v_0$, the $T_i$-path between $v_0$ and the root of $T_i$ is a separator set of $G[P_i]$. Nodes belonging to sets $P_i$ that have computed a separator set in this phase do not perform further computations and wait until all phases are completed. \\

\noindent  {\bf Phase 3.}  If the nodes in a subgraph $G[P_i]$ reaches this phase, then there exists at least one $T_i$-real fundamental edge.

Since the endpoints of each $T_i$-real fundamental edge $e = uv$ have calculated the weight $\omega(F_e)$ of the real fundamental face $F_e$, through a \rangeproblem\ with $R_i = [n_i/3, 2n_i/3]$, and $x_v = \omega(F_e)$ if $v \in e$ and $\omega(F_e) \in R_i$, and $x_v=0$ otherwise, the nodes of $P_i$ learn the \id\ of a node $v$ (if it exists) that is part of a $T_i$-real fundamental edge $e$ such that $\omega(F_e) \in R_i$. This can be done in $\biglo{D}$ rounds (\Cref{lemma:somebroadcasts}).

If a subgraph $G[P_i]$ does not have any $T_i$-real fundamental edge with weight in the range $[n_i/3, 2n_i/3]$, then the nodes of $P_i$ wait until the start of the next phase. On the other hand, if the nodes of \( P_i \) acknowledge the \id\ of a node \( v \in P_i \) that is part of a fundamental $T_i$-real edge \( e \) such that \( \omega(F_e) \in [n_i/3, 2n_i/3] \), then this node \( v \) can communicate the IDs of both endpoints of \( e \) to all other nodes in \( \biglo{D} \) rounds (\Cref{broadcasts}).

Finally, all the nodes on the $T_i$-path between the endpoints of \( e \), which are now known by all nodes in \( P_i \), are marked in \( \biglo{D} \) rounds using \Cref{teo:markpathcongest}. According to \Cref{teo:balancedface}, this path forms a cycle separator of \( G[P_i] \). All subgraphs \( G[P_i] \) that computed a cycle separator in this phase will wait until the end of the algorithm.\\

\noindent  {\bf Phase 4.} All sets \( P_i \) that reach this round do so because they have at least one \( T_i \)-real fundamental edge, and no face defined by these real fundamental edges has a weight in the range \([n_i/3, 2n_i/3]\). Any set \( P_i \) participating in this phase begins by searching for the existence of a \( T_i \)-fundamental edge \( e \) such that \( \omega(F_e) > 2n_i/3 \) (if it exists). This is done analogously to the previous part. Each node \( v \in P_i \) defines

\begin{equation*}
    x_v = \begin{cases}
        \omega(F_e) &v\in e, e\text{ is a }T_i\text{-real fund. edge}, \omega(F_e)> 2n_i/3 \\
        0 &\text{ otherwise}
    \end{cases}
\end{equation*}

And through a \rangeproblem\ using $R_i = [2n_i/3, n_i]$, all nodes in $P_i$ determine the ID of a node $v_0 \in P_i$ such that $x_{v_0} \in R_i$. On one hand, if such a node does not exist, then all nodes in $P_i$ wait until the next and final phase begins. 

On the other hand, if such a node $v_0$ exists, each node in $P_i$ identifies all $T_i$-real fundamental edges that have a weight greater than $2n_i/3$. Then, through a \notcontainsproblem, the nodes in each $P_i$ determine the endpoint's ID of a real fundamental edge $f^i = u_iv_i$ that has a weight greater than $2n_i/3$, and $f^i$ does not include any other $T_i$-fundamental edge with weight greater than $2n_i/3$. Furthermore, since no fundamental edge in this phase has a weight between $[n_i/3, 2n_i/3]$, all real fundamental edges contained in $F_{f^i}$ (if any exist) have weights less than $n_i/3$.

Subsequently the nodes in $P_i$, using \detectfaceproblem, determine whether they are in $F_{f^i}$ or not. Additionally, each subgraph $G[P_i]$ computes \fullagumentationproblem\ of the face $F_{f^i}$ so that all nodes $z$ in $F_{f^i}$ know the weight of $F_{u_1z}^\ell$.

Once all nodes in $F_f$ have computed the weight of a full augmentation from $u_i$, each node $s\in P_i$ defines $x_s = \omega\left(F_{u_is}^\ell\right)$ if $s\in \inside{F}_{f^i}$ and $x_s = 0$ otherwise. Through a \rangeproblem\ with $R_i = [n_i/3, 2n_i/3]$ the nodes in $P_i$ acknowledge, if exists, a node $s_i\in F_f$ such that $x_{s_i} = \omega\left(F_{u_is_i}^\ell\right)$ and $\omega\left(F_{u_is_i}^\ell\right)\in [n_i/3, 2n_i/3]$. If such a node exists, the nodes of $P_i$ goes to {\tt Sub-phase 4.1}, and if does not exists, then nodes wait until {\tt Sub-phase 4.2} begins.\\

\noindent  {\bf Sub-phase 4.1.} If the nodes in $P_i$ reach this subphase, there exists a node $s^i \in P_i$ such that $\omega\left(F_{u_is_i}^\ell\right) \in [n_i/3, 2n_i/3]$ within a $T_i$-real fundamental face $F_{f^i}$. Without loss of generality, $s_i$ is a leaf of $T_i$ by \Cref{lemma:propaug}. If $s_i$ is not a leaf of $T_i$, then according to \Cref{lemma:propaug}, the $T_i$-leaf $t_i$ that is a descendant of $s_i$ with the highest position in $\leftpi^{T_i}$, provided that $f^{i}$ is such that $u_i$ is not an ancestor of $v_i$ or $f^{i}$ is $\mathcal{E}$-left oriented edge, satisfies $\omega\left(F_{u_is_i}^\ell\right) = \omega\left(F_{u_it_i}^\ell\right) \in R_i$. Using \Cref{broadcasts} and \Cref{lemma:somebroadcasts}, this leaf $t_i$ can be acknowledged by all nodes, and we proceed through this phase with $s_i = t_i$. If $f^{i}$ is a $\mathcal{E}$-right oriented edge, the $T_i$-leaf $t_i$ that is a descendant of $s_i$ satisfies $\omega\left(F_{u_is_i}^\ell\right) = \omega\left(F_{u_it_i}^\ell\right) \in R_i$, and again this phase proceeds with $s_i = t_i$. This adjustment of $s_i$ to ensure it is a leaf in $T_i$ takes $\biglo{D}$ rounds.

Then, through \hiddenproblem\ for this node $s_i$, all nodes in $F_{f^i}$ can determine if they are incident to any $T_i$-real fundamental edge that hides to $s_i$ in $F_{f^i}$. Following this, using \Cref{broadcasts}, the nodes can determine if there exists a $T_i$-real fundamental edge that hides to $s_i$. If such an edge does not exist, then $s_i$ is $(T,F_e)$-compatible with $u$, and by \Cref{lemma:balancedincritical}, the $T$-path between $s_i$ and $u_i$ serves as a cycle separator of $G$, which can be marked in $\biglo{D}$ rounds using \Cref{teo:markpathcongest}.

If there exists at least one fundamental edge that hides to $v_i$, then through a \notcontainedproblem, the nodes in $P_i$ can determine the IDs $\text{id}(z_1^i)$ and $\text{id}(z_2^i)$ such that $f'' = z_1^iz_2^i$ is a $T_i$-real fundamental edge that hides to $s_i$ and $f''$ is not contained in any other fundamental edge that also hides to $s_i$. Then, assuming $\leftpi(z_1^i)<\leftpi(z_2^i)$, according to \Cref{lemma:balancedincritical}, the nodes proceed to mark the nodes in the $T_i$-path between $u_i$ and $z_2^i$. The correctness of the cycle separator marked in this case is guarantee thanks to \Cref{lemma:balancedincritical}\\

\noindent  {\bf Sub-phase 4.2.}   If no leaf $z\in \inside{F}_{f^i}$ has weight $\omega\left(F_{u_1z}^\ell\right)\in [n_i/3, 2n_i/3]$, then all nodes have $\omega\left(F_{u_1z}^\ell\right)<n_i/3$. Recall that in the proof of \Cref{teo:fundamentalseparator}, it is proven that there exists a cycle separator of $G$ by analyzing two different cases. In both cases, the conclusion is that the $T_i$-path between $u_i$ and $v_i$ is a separator cycle of $G[P_i]$. Therefore, using the \markpathproblem, the nodes in the $T_i$-path between $u_i$ and $v_i$ can be marked in $\biglo{D}$ rounds.\\

\noindent {\bf Phase 5.} All sets $P_i$ that reach this round do so because they have at least one $T_i$-real fundamental edge, and all faces defined by these real fundamental edges have a weight smaller than $n_i/3$. In this final Phase, the nodes simulate the proof of \Cref{lemma:verysmall}.

First, by an analogous application of \notcontainedproblem, the nodes in each $P_i$ can acknowledge the endpoint's \id\ of a real fundamental edge $e_i=u_iv_i$ that is not contained in any other $T_i$-real fundamental edge.

Locally, $u_i$ and $v_i$ can compute the value of $|F_\ell^{e_i}|$ and $|F_r^{e_i}|$, respectively,\footnote{Definition of these values in proof of \Cref{lemma:verysmall}}, because by definition

\begin{align*}
    \left|F_\ell^{e_i}\right| &= \leftpi(u_i)-d_{T_i}(u_i) + p_{u_i}(v_i)\\
    \left|F_r^{e_i}\right| &= \rightpi(v_i)-d_{T_i}(v_i) + p_{v_i}(u_i)
\end{align*}

In $\biglo{D}$ rounds, these values $|F_\ell^{e_i}|, |F_r^{e_i}|$ can be acknowledge by all the nodes in $P_i$. If $|F_\ell^{e_i}|, |F_r^{e_i}|\leq n_i/3$ or $n_i/3\leq |F_\ell^{e_i}|, |F_r^{e_i}|\leq 2n_i/3$, then by \Cref{lemma:verysmall}, the $T_i$-path between $u_i$ and $v_i$ is a cycle separator of $G[P_i]$.

Finally, if $|F_r^{e_i}|\geq 2n_i/3$ (resp. $|F_\ell^{e_i}|\geq n_i/3$), then the nodes simulates Phase 4 using the face $F_{v_ir_{T_i}}$ such that $\inside{F}_{v_ir_{T_i}} = F_r^{e_i}$ (resp. the face $F_{u_ir_{T_i}}$ such that $\inside{F}_{u_ir_{T_i}} = F_\ell^{e_i}$). In these case a cycle separator is computed (\Cref{lemma:balancedincritical}).

In any case, in this point has been marked a set of nodes in each $P_i$ that are a cycle separator of $G[P_i]$. As a constant number of subroutines are needed, and each of them needs $\biglo{D}$ rounds, then the overall round complexity of the algorithm is of $\biglo{D}$ rounds.
\qed

\section{Application: DFS tree construction.}\label{sec:dfscongest}

Although Ghaffari and Parter~\cite{ghaffari:2017} gave a randomized DFS algorithm based on cycle separators, following essentially the same high-level idea that we use, we cannot simply replace their randomized cycle-separator subroutine with our deterministic one and treat their algorithm as a black box. The reason is that their approach also relies on additional randomized subroutines (beyond cycle separators), such as mark-a-path procedures. We therefore revisit the construction in detail and formally describe each deterministic protocol needed to obtain a deterministic DFS algorithm.
Analogous to the previous section, we first formalize and provide algorithmic solutions to two sub-problems necessary for computing a deterministic DFS. We then describe in detail the deterministic algorithm that computes a DFS in $\biglo{D}$ rounds.

\subsection{Deterministic subroutines for the DFS algorithm.}\label{subsec:dfscongestalg}

\subsubsection{Re-root a tree.}
As discussed in \Cref{subsec:dfshighlevelideas}, to achieve an algorithm with a round complexity of $\biglo{D}$ in the \CONGEST\ model, \joinpath\ must be resolved in $\biglo{D}$ rounds. Before describing the solution to this problem, a more fundamental problem is presented and solved: how to re-root a tree in $\biglo{D}$ rounds.

\begin{algorithmbox}{\rerootproblem}
     \textbf{Input}: A planar configuration $(G,\mathcal{E},T)$ with diameter $D$, a partition of $\mathcal{P}$ of $V$, a spanning tree $T_i$ of each $G[P_i]$, and the \id\ of a node $v_i\in P_i$ known by all the nodes in $P_i$ \\

       \textbf{Output:} Each $P_i$ has update their tree $T_i$ such that $T_i$ is rooted in $v_i$
\end{algorithmbox}

By re-rooting \( T_i \) at \( v_i \), it is meant that all nodes update their distances in \( T_i \) to be their distance to \( v_i \) in \( T_i \) and update their parents accordingly.

\begin{lemma}
    There exists a deterministic algorithm in the \CONGEST\ model. which solves \rerootproblem\ in $\biglo{D}$ rounds.
\end{lemma}

\begin{proof}
   In the distributed construction of each $T_i$, each node $v \in P_i$ knows its parent $p_{T_i}(v)$, which corresponds to the first node $z \neq v$ on the $T_i$-path between $v$ and $r_{T_i}$, starting from $v$, and its depth $d_{T_i}(v)$, which corresponds to the distance of the node $v$ to the root $r_{T_i}$ in $T_i$. To re-root $T_i$ at $v_i$, while keeping the same edges, the depth of each node $v$ must be updated to the length of the $T_i$-path between $v$ and $v_i$, and its parent $p_{T_i}(v)$ must be updated to the first node on the $T_i$-path between $v$ and $v_i$, different from $v$ and starting from $v$.

First, in $\mathcal{O}(D)$ rounds, all nodes in $P_i$ can determine whether they are ancestors or descendants of the node $v_i$, or neither, using \detectancestor\ and \detectdescendant. Then, also in $\mathcal{O}(D)$ rounds (\Cref{broadcasts}), $v_i$ can inform all other nodes in $P_i$ of its original depth $d_{T_i}(v_i)$. After sending this information, $v_i$ updates its own information to $d_{T_i}(v_i) = 0$ and $p_{T_i}(v_i) = \perp$.

All nodes $v \in P_i$ that are descendants of $v_i$ in the original $T_i$ keep their same parent but update their depth to $d_{T_i}(v) = d_{T_i}(v) - d_{T_i}(v_i)$. All nodes $v \in P_i$ that are ancestors of $v_i$ in the original $T_i$ update their parent $p_{T_i}(v)$ to be the unique one of their $T_i$-children (in the original $T_i$) who is also an ancestor of $v_i$, and update their depth to $d_{T_i}(v) = d_{T_i}(v_i) - d_{T_i}(v)$. All nodes $v \in P_i$ that are neither descendants nor ancestors of $v_i$ in the original $T_i$ keep their same parent but update their depth to $d_{T_i}(v) = d_{T_i}(v) + d_{T_i}(v_i)$.
\end{proof}

\subsubsection{Adding a cycle separator: proof of \Cref{teo:joinkpaths}.}\label{subsec:proofjoinkpaths}

Recall that formally, the problem to be solved is the following

\begin{algorithmbox}{\joinproblem}
	
	\textbf{Input}: A partial DFS tree $T_d$ of $G$, the connected components $\mathcal{C}=\{C_1,...,C_k\}$ of $G-T_d$, and a set of marked nodes $S_i$ in each $P_i$ such that each $S_i$ is a cycle separator of $G[P_i]$ \\
	\medskip
	\textbf{Output:} A partial DFS tree $\tilde{T}_d$ of $G$ that contains $T_d$ and all  the nodes in $\cup_{i\in [k]} S_i$. \footnote{Other nodes can also be added to $T_d$ besides the nodes in the cycle separators}

\end{algorithmbox}

Recall that, by definition of the problem, \(G\) comes with a partial DFS tree \(T_d\). Let \(C_1,\dots,C_k\) be the connected components of \(G - T_d\). In principle, \(k\) can be as large as \(\Theta(n)\). However, the low-congestion shortcut framework allows us to process all these components in parallel.

Crucially, although this is not an intrinsic property of the problem, it is a consequence of our approach based on cycle separators that each component \(C_i\) has size at most a constant fraction of \(G\). This fact is crucial later for the recursive application of the algorithm described below. Indeed, it guarantees that the recursion terminates after \(\bigo{\log n}\) levels, and hence the final DFS tree is computed in $\biglo{D}$ round.

To add all the nodes in each $S_i$ to $T_d$ following the \dfsrule, the process for joining a particular subset $S_i$ to $T_d$ is explained, because all the subroutines developed allow all the sets $S_1,...,S_k$ to be added in parallel. The algorithm operates in $\bigo{k \cdot \log n}$ phases, and each phase requires $\biglo{D}$ rounds.

In the first round, using \Cref{broadcasts}, the nodes in $P_i$ acknowledge the \id\ of the node $v_i$ with the deepest node in $T_d$. Then, using the \rerootproblem, $T_i$ is re-rooted at $v_i$ if necessary. Note that each node knows whether it is in $S_i$ or not, but the nodes in $P_i$ do not know the complete list of nodes in $S_i$. Then, with at most $k$ calls to \minproblem, all the nodes in $P_i$ can learn the \id\ of the at most $k$ marked nodes in $S_i$ that are leaves of $T_i$. This is achieved because each node $v\in P_i$ can define $x_v = 1$ if $v\in S_i$, is a $T_i$ leaf, and has not yet been acknowledged by the other nodes, and call \minproblem\ $k$ times with these inputs.\footnote{Note that in the \rerootproblem\ the same edges of $T_i$ are maintained.}

Then, each node can locally form the lexicographical order of the leaves $(v_1,v_2)$ $,(v_1,v_3)$ ,$ $...$ $, $ (v_{k-1},v_k)$ such that $(v_i,v_j)\leq(v_h,v_g) \iff \id(v_i)<\id(v_h)\lor (\id(v_i)=\id(v_h)\land \id(v_j)<\id(v_g))$, and since all the nodes know all the leaves of $S_i$, the same order is formed by all. Then, with at most $k^2$ uses of the \lcaproblem, in lexicographical order for each pair of leaves $(v_i,v_j)$, the nodes learn the \id\ of LCA$(v_i,v_j)$ for all $v_i,v_j$ leaves of $T_i$ in $S_i$. Given this, with a straightforward application of \minproblem, all the nodes in $P_i$ can acknowledge which node $z_1 = LCA(v_i,v_j)$ has the least depth in $T_i$. Since $S_i$ is a cycle separator of $T_i$, this node $z_1$ is unique and $z_1\in S_1$.

Next, using \maxproblem, it can be decided which leaf $h$ of $T_i$ in $S_i$ has the greatest distance from $z_1$, since each node $v'$ in $S_i$ knows that the distance to $z_1$ is $d_{T_i}(v')-d_{T_i}(z_1)$. To complete the first phase, using \markpathproblem, all the nodes in the $T_i$-path $P_1$ between the root $v_i$ and $h$ (which by construction contains $z_1$) are marked.

These marked nodes are then added to $T_d$. In $\biglo{D}$ rounds, the node $v_i$ can inform all other nodes in $P_i$ of the depth $p_i$ of its deepest neighbor in $T_d$. With this information, all marked nodes $v'$ in $P_1$ update their depth in $T_d$ as $d_{T_d}(v') = p_i + 1 + d_{T_i}(v')$ and update their parent in $T_d$ to their parent in $T_i$, i.e., $p_{T_d}(v') = p_{T_i}(v')$. With this, all nodes in $P_1$ are added to $T_d$ and since $v_i$ is the node with the deepest neighbor in $T_d$, the \dfsrule\ is followed.

Note that in this first phase, at least half of the nodes in the cycle separator $S_i$ of $T_i$ are added to $T_d$.

For the second phase, the same procedure is repeated, but now in each connected component $C_1^2,...,C_{j_2}^2$ of $G[P_i]-P_1$ such that $S' \cap V(C_h^2)\neq \emptyset$ with $S' = S_1 - V(P_1)$. In each of these connected components, an MST is first constructed with the weight function $\omega$ such that $\omega(uv) = 0$ if $u,v\in S'$ and $\omega(uv)=1$ otherwise. These MSTs $T_h^2$ can be constructed in $\biglo{D}$ rounds (\Cref{teo:mstcongest}) and the weight function ensures that all nodes in $S'$ in the same connected component $C_h^2$ will form a $j$-$T_h^2$-path in the MST $T_h^2$ with $j\leq k$. From this point, the same procedure as in the first phase is followed, so by the end of the second phase, after $\biglo{D}$ rounds, a path $P_h^2$ has been joined to $T_d$ in each connected component $C_h^2$, containing at least half of one of the $j\leq k$ paths that form the $j$-$T_h^2$-path marked in $C_h^2$.

Subsequently, in each phase $\ell$, the connected components $C_1^\ell,...,C_{j_\ell}^\ell$ of $G[P_i]-P_1-\cup_{s< \ell,t\leq j_s} P_t^s$ are considered such that $S' \cap V(C_h^\ell)\neq \emptyset$ with $S' = S_1 - V(P_1) - \cup_{s< \ell,t\leq j_s} V(P_t^s)$. In each of these connected components, the procedure is analogous to the second phase. Thus, in each phase $\ell$, at least half of the nodes in one of the remaining paths in $S' = S_1 - V(P_1) - \cup_{s< \ell,t\leq j_s} V(P_t^s)$ are added to $T_d$. Since there are at most $k$ paths, after $\bigo{k \times \log n}$ rounds, all the nodes in $S_i$ have been added to $T_d$, achieving the round complexity of $\biglo{D}$ rounds for a constant $k$. \qed

Given all these results, all the subroutines required to prove \Cref{teo:separatorcongest}
 and \Cref{teo:dfscongest} are available.

\subsection{DFS tree construction: proof of \Cref{teo:dfscongest}.}\label{dfscongestproof}

Given a planar input graph $F=(V,E)$ with diameter $D$, first the algorithm computes a planar combinatorial embedding $\mathcal{E}$ of $G$ in $\biglo{D}$ rounds. As explained in \Cref{subsec:dfshighlevelideas}, the deterministic DFS algorithm, also called \emph{main algorithm}, is the following: In each recursive call, referred to as a \emph{phase}, we have a partial DFS-tree $T_{d}$ already constructed. this means that each node in $T_d$ knows its distance of the root and its parent in $T_d$. This information does not change once a one joins $T_d$.

In the first phase, $T_{d}$ consists just of the node $r$ without any edge. Then, in  each phase we compute the following steps, in parallel, in each connected component $C_i$ of $G-T_d$:

\begin{enumerate}[align= left,label=\textbf{Step \arabic*.}, ref = {\bf Step \arabic*}.]
	\item  Compute a cycle separator $S_i$ of $C_i$. More precisely, we mark a set of nodes $S_i\subseteq V(C_i)$ that induce a separator set of $C_i$.
	\item  Add the marked cycle separator \(S_i\) to the partial DFS-tree $T_{d}$ following the \dfsrule.
\end{enumerate} 

By \Cref{teo:separatorcongest} (whose proof is in \Cref{subsec:proofteo}), it can be shown that Step 1 can be completed in $\biglo{D}$ rounds. Since the nodes of the computed cycle separators forms a path of a spanning tree $T_i$ of each connected component $C_i$ of $G - T_d$, Step 2 can also be resolved in $\biglo{D}$ rounds by \Cref{teo:joinkpaths} (whose proof is in \Cref{subsec:proofjoinkpaths}). In each phase $j$, a cycle separator from each connected component is added to $T_d$. Consequently, in the next phase $j+1$, the maximum size of a connected component decreases by at least a factor of $1/3$ compared to the maximum size of a connected component in phase $j$. Therefore, after $\mathcal{O}(\log n)$ phases, no nodes remain to be added to $T_d$. Thus, in $\biglo{D}$ rounds, a DFS tree of $G$ rooted at $r$ is computed.\qed

\section{Other deterministic applications}
\label{sec:otherapp}

Our deterministic algorithm for computing cycle separators does not, by itself, yield a deterministic DFS algorithm, since additional subproblems must be handled (\Cref{subsec:dfscongestalg}). In contrast, several state-of-the-art randomized algorithms for planar graphs can be derandomized directly using our cycle separator algorithm, as their only source of randomness is the computation of cycle separators. For completeness, we briefly outline the main ideas of these algorithms and state formally the deterministic algorithms that we obtain, but refer to the original papers for full details.

{\bf Bounded Diameter Decompositions.} Li and Parter (Definition 4.1., \cite{parter}) introduced the notion of Bounded Diameter Decomposition (BDD) of a planar graph. A BDD of an $n$-node planar graph $G=(V,E)$ with diameter $D$ is a rooted tree $\mathcal T=(V_{\mathcal T}, E_{\mathcal T})$ such that each node $X\in V_{\mathcal T}$ is a subset of nodes of $G$, i.e. $X\subseteq V$, and therefore they are called bags. A BDD decomposition $\mathcal T$ satisfies important properties that allow solving certain problems recursively over $\mathcal T$ in $\biglo{poly(D)}$ rounds, even dispensing with the use of low-congestion shortcuts. For example, the tree $\mathcal T$ has low depth ($\bigo{\log n}$), each bag $X\in V_{\mathcal T}$ induces a connected graph in $G$ with diameter close to $D$ (namely, diameter of $\bigo{D\log n}$). See \cite{parter} for an extensive description.

Algorithmically, it was shown in \cite{parter} that a BDD can be computed recursively as follows: Given an $n$-node planar graph $G$ with diameter $D$.

\begin{itemize}
    \item Start with $G=G'$, compute a cycle separator $S$ of $G'$.
    \item Let $V_1,\dots, V_k$ be the nodes of each connected component of $G'-S$. Each bag $X_i$ of this iteration is given by a set of nodes $V_i$ plus some nodes of $S$, in order to guarantee that the bag satisfies the properties of Definition 4.1. in  \cite{parter}.
    \item Continue recursively, to define the children bags of each bag $X_i$, in each connected component $G[V_i]$.
\end{itemize}

In their algorithm to compute a BDD, step $2.$ of the recursion consists of straightforward deterministic algorithms to decide which nodes of the separator will join to each connected component to define the bags. Therefore, given that for step $1.$ we can use the deterministic cycle separator given by \Cref{teo:separatorcongest}, we obtain the following Corollary.

\begin{corollary}
     There exists a deterministic algorithm in the \CONGEST\ model which, given a planar graph $G~=~(V,E)$ with diameter $D$, computes a BDD $\mathcal T$ of $G$ in $\biglo{D}$ rounds.
\end{corollary}

\begin{proof}
    Follows directly by using the deterministic cycle separator algorithm given by \Cref{teo:separatorcongest} in the algorithm given by Theorem 4.2. in \cite{parter}.
\end{proof}

{\bf Application to different problems.} Using cycle separators or BDDs, several randomized algorithms are known whose only source of randomness (so far) is the procedure that computes the cycle separator or the BDD. We describe a collection of problems for {\it directed} planar graphs and the complexity of the new deterministic algorithm for them. This list of problems is not extensive and we refer to the references for other applications or modifications of the problems.

\begin{itemize}[label=$\triangleright$]
    \item {\it Single Source Shortest Path (SSSP)}. Given an $n$-node directed planar graph $G=(V,E)$, source node $s\in S$ and a weight function $\omega\colon V\to \mathbb Q$, the objective is that each node $v$ in $G$ learns its weighted distance to $s$ (respect to $\omega$).
    \item {\it $s$--$t$-Maximum Flow}. Given an $n$-node directed planar graph $G=(V,E)$, two nodes $s,t\in V$ and an integer capacity assignment in the edges $c\colon E\to \mathbb N$, the objective is to compute the maximum $s$--$t$ flow and the respective assignment of flow in the edges.
    \item {\it Reachability.} Given an $n$-node directed planar graph $G=(V,E)$ and a source node $s\in V$, the objective is to identify all the nodes of $G$ that are reachable from $s$ through a directed path.

    \item {\it Strongly Connected Components (SCCs).} Given an $n$-node directed planar graph $G=(V,E)$, the goal is to compute all the strongly connected components (SCCs) of $G$. In the distributed setting, this means that each node must output the identifier ($\id$) of the strongly connected component to which it belongs.
\end{itemize}

\begin{corollary}
    There exist deterministic algorithms $\mathcal A _ {1}, \mathcal A _ {2}, \mathcal A_3 $ and $\mathcal A_4$ in the \CONGEST\ model such that, given a directed planar graph $ G~ ~(V,E)$ with diameter $D$.

    \begin{itemize}
        \item $\mathcal A _ {1} $ solves the SSSP problem in $\biglo{D^2}$ rounds and, moreover, computes a directed SSSP tree with root in the source node $s$.
        \item $\mathcal A_2$ solves {\it Reachability} problem in $\biglo{D}$ rounds.

        \item $\mathcal A_3$ solves $s$--$t$-Maximum Flow problem in $\biglo{D^2}$ rounds.
        
        \item $\mathcal A_4$ solves {\it SCCs} problem in $\biglo{D}$ rounds.
        
    \end{itemize}
\end{corollary}

Each deterministic algorithm is obtained by plugging our deterministic cycle separator algorithm from \Cref{teo:separatorcongest} into the corresponding randomized algorithm from the references listed in Table~1, replacing their cycle separator routine.

\section{Conclusion.}\label{sec:conclusion}

In this work, we have developed a deterministic algorithm in the \CONGEST\ model to compute non-trivial cycle separators in planar graphs. This result enabled the creation of a deterministic, nearly-optimal algorithm in the \CONGEST\ model that computes a DFS tree in undirected planar graphs in 
$\biglo{D}$ rounds.

Several interesting questions remain open. On one hand, this paper initiates the study and application of cycle separators to solve problems that still do not have deterministic algorithms matching their lower bound, or their randomized counterpart. While for some of these problems a deterministic algorithm can be obtained directly, since their only source of randomness lies in the computation of cycle separators, for others this is not the case. A notable example is diameter computation \cite{parter}, which relies on randomization through the use of hash functions to compress the computed information; hence, the mere existence of deterministic cycle separator algorithms is not sufficient to derive a fully deterministic solution.

Another promising research direction is to investigate the existence of deterministic algorithms for computing non-trivial separator sets in graph classes beyond planar graphs—such as bounded-genus graphs, minor-free graphs, and graphs with bounded treewidth. Interestingly, for these graph families, several works have already established the existence of deterministic algorithms for computing low-congestion shortcuts \cite{gyhdense, haeupler2018round, haeupler2016low, haeupler2016near}. However, computing separator sets in these classes may be significantly more challenging, as one cannot rely on tools like those developed by Lipton and Tarjan for computing cycle separators in planar graphs.

\bibliographystyle{plain}
\bibliography{bibliografia}

@inproceedings{ghaffari2013distributed,
  title={Distributed minimum cut approximation},
  author={Ghaffari, Mohsen and Kuhn, Fabian},
  booktitle={International Symposium on Distributed Computing},
  pages={1--15},
  year={2013},
  organization={Springer}
  }

@inproceedings{abd2025distributed,
  title={Distributed Maximum Flow in Planar Graphs},
  author={Abd-Elhaleem, Yaseen and Dory, Michal and Parter, Merav and Weimann, Oren},
  booktitle={Proceedings of the ACM Symposium on Principles of Distributed Computing},
  pages={278--286},
  year={2025},
  doi = {https://doi.org/10.1145/3732772.3733521}
}

@book{peleg2000distributed,
  title={Distributed computing: a locality-sensitive approach},
  author={Peleg, David},
  year={2000},
  publisher={SIAM},
  address = {Philadelphia, PA, USA.},
  doi = {https://doi.org/10.1137/1.9780898719772}
}

@inproceedings{haeupler2016low,
  title={Low-congestion shortcuts without embedding},
  author={Haeupler, Bernhard and Izumi, Taisuke and Zuzic, Goran},
  booktitle={Proceedings of the 2016 ACM Symposium on Principles of Distributed Computing},
  pages={451--460},
  year={2016},
  doi = {https://doi.org/10.1145/2933057.2933112}
}

@inproceedings{haeupler2016near,
  title={Near-optimal low-congestion shortcuts on bounded parameter graphs},
  author={Haeupler, Bernhard and Izumi, Taisuke and Zuzic, Goran},
  booktitle={International Symposium on Distributed Computing},
  pages={158--172},
  year={2016},
  organization={Springer},
  doi = {https://doi.org/10.1007/978-3-662-53426-7_12}
}

@inproceedings{ghaffari2016distributed-1,
  title={Distributed algorithms for planar networks I: Planar embedding},
  author={Ghaffari, Mohsen and Haeupler, Bernhard},
  booktitle={Proceedings of the 2016 ACM Symposium on Principles of Distributed Computing},
  pages={29--38},
  year={2016},
  doi = {https://doi.org/10.1145/2933057.2933109}
}

@inproceedings{ghaffari2016distributed-2,
  title={Distributed algorithms for planar networks II: Low-congestion shortcuts, MST, and Min-Cut},
  author={Ghaffari, Mohsen and Haeupler, Bernhard},
  booktitle={Proceedings of the twenty-seventh annual ACM-SIAM Symposium on Discrete Algorithms, SODA '16'},
  pages={202--219},
  year={2016},
  organization={SIAM},
  doi = {https://doi.org/10.1137/1.9781611974331.ch15}
}

@inproceedings{ghaffari:2017,
  title={Near-optimal distributed DFS in planar graphs},
  author={Ghaffari, Mohsen and Parter, Merav},
  booktitle={31st International Symposium on Distributed Computing (DISC 2017)},
  year={2017},
  doi = {https://doi.org/10.4230/LIPIcs.DISC.2017.21}}

@inproceedings{dfs1987aggarwal,
  title={A random {NC} algorithm for depth first search},
  author={Aggarwal, Alok and Anderson, Richard},
  booktitle={Proceedings of the Nineteenth Annual ACM Symposium on Theory of Computing},
  pages={325--334},
  year={1987}
}

@article{lipton:1979,
    author = "Richard Lipton and Robert Tarjan",
    title = "A separator theorem for planar graphs",
    year = "1979",
    journal = "SIAM J. Appl. Math.",
    doi = {https://doi.org/10.1137/0136016}
}

@article{GILBERT1984391,
title = {A separator theorem for graphs of bounded genus},
journal = {Journal of Algorithms},
volume = {5},
number = {3},
pages = {391-407},
year = {1984},
issn = {0196-6774},
doi = {https://doi.org/10.1016/0196-6774(84)90019-1},
author = {John R Gilbert and Joan P Hutchinson and Robert Endre Tarjan}
}

@article{seymournonplanar,
 ISSN = {08940347, 10886834},
 author = {Noga Alon and Paul Seymour and Robin Thomas},
 journal = {Journal of the American Mathematical Society},
 number = {4},
 pages = {801--808},
 doi = {https://doi.org/10.2307/1990903},
 publisher = {American Mathematical Society},
 title = {A Separator Theorem for Nonplanar Graphs},
 volume = {3},
 year = {1990}
}

@inproceedings{haeupler2018round,
  author       = {Bernhard Haeupler and
                  D. Ellis Hershkowitz and
                  David Wajc},
  editor       = {Calvin Newport and
                  Idit Keidar},
  title        = {Round- and Message-Optimal Distributed Graph Algorithms},
  booktitle    = {Proceedings of the 2018 {ACM} Symposium on Principles of Distributed
                  Computing, {PODC} 2018, Egham, United Kingdom, July 23-27, 2018},
  pages        = {119--128},
  publisher    = {{ACM}},
  year         = {2018},
  doi          = {https://doi.org/10.1145/3212734.3212737},
  address = {Egham, United Kingdom}
}

@inproceedings{Ghaffari2022, series={PODC ’22},
   title={Universally-Optimal Distributed Exact Min-Cut},
   doi={http://dx.doi.org/10.1145/3519270.3538429},
   booktitle={Proceedings of the 2022 ACM Symposium on Principles of Distributed Computing, PODC},
   publisher={ACM},
   author={Ghaffari, Mohsen and Zuzic, Goran},
   month=jul, collection={PODC ’22},
   address ={Salerno, Italy}
   }

@inproceedings{gyhdense,
  author       = {Mohsen Ghaffari and
                  Bernhard Haeupler},
  editor       = {Avery Miller and
                  Keren Censor{-}Hillel and
                  Janne H. Korhonen},
  title        = {Low-Congestion Shortcuts for Graphs Excluding Dense Minors},
  booktitle    = {{PODC} '21: {ACM} Symposium on Principles of Distributed Computing},
  pages        = {213--221},
  publisher    = {{ACM}},
  doi          = {https://doi.org/10.1145/3465084.3467935},
  address = {Online}
}

@article{Fox2013ApplicationsOA,
  title={Applications of a New Separator Theorem for String Graphs},
  author={Jacob Fox and J{\'a}nos Pach},
  journal={Combinatorics, Probability and Computing},
  year={2013},
  volume={23},
  pages={66 - 74},
  doi = {https://doi.org/10.1017/S0963548313000412}
}

@inproceedings{separatorapplications,
  author       = {Richard J. Lipton and
                  Robert Endre Tarjan},
  title        = {Application of a Planar Separator Theorem},
  booktitle    = {18th Annual Symposium on Foundations of Computer Science, FOCS, 1977},
  pages        = {162--170},
  publisher    = {{IEEE} Computer Society},
  year         = {1977},
  doi          = {https://doi.org/10.1109/SFCS.1977.6},
  address      = {Providence, RI, USA}
}

@inproceedings{parter, author = {Li, Jason and Parter, Merav}, title = {Planar diameter via metric compression}, year = {2019}, isbn = {9781450367059}, publisher = {Association for Computing Machinery}, address = {New York, NY, USA}, doi = {10.1145/3313276.3316358}, location = {Phoenix, AZ, USA}, series = {STOC 2019} }

@inproceedings{parter2020,
  author       = {Merav Parter},
  editor       = {Hagit Attiya},
  title        = {Distributed Planar Reachability in Nearly Optimal Time},
  booktitle    = {34th International Symposium on Distributed Computing, {DISC} 2020,
                  October 12-16, 2020, Virtual Conference},
  series       = {LIPIcs},
  volume       = {179},
  pages        = {38:1--38:17},
  year         = {2020},
  doi          = {https://doi.org/10.4230/LIPIcs.DISC.2020.38}
}

@article{awerbuch1985new,
  title={A new distributed depth-first-search algorithm},
  author={Awerbuch, Baruch},
  journal={Information Processing Letters},
  volume={20},
  number={3},
  pages={147--150},
  year={1985},
  publisher={Elsevier},
  doi = {https://doi.org/10.1016/0020-0190(85)90083-3}
}

\end{document}